%% file: xpathbib.tex
\documentclass{elsarticle}
\biboptions{sort&compress}
\journal{}

\usepackage{array}
\usepackage{rotating}
\usepackage{varioref}
\usepackage{amsmath,amssymb,amsfonts,amsthm}
\usepackage{algorithmic,algorithm}
\usepackage{graphicx}
\usepackage{xspace}
\usepackage{xcolor}

\newcommand{\ch}[1]{\mathop{\smash{\textrm{ch}_{\geq #1}}}}
\newcommand{\BL}{\ensuremath{\mathcal{X}}}
\newcommand{\CL}{\ensuremath{\mathcal{C}}}
\newcommand{\ass}{\ensuremath{\mathrel{:=}}}
\newcommand{\sig}{\textrm{sig}}
\newcommand{\down}{\ensuremath{\mathord{\downarrow}}}
\newcommand{\up}{\ensuremath{\mathord{\uparrow}}}
\newcommand{\updown}{\ensuremath{\mathord{\updownarrow}}}
\renewcommand{\top}{\textrm{top}}
\newcommand{\siggeq}{\ensuremath{\mathrel{\gtrsim}}}
\newcommand{\sigequiv}{\ensuremath{\mathrel{\cong}}}
\newcommand{\expgeq}{\ensuremath{\mathrel{\geq_{\rm exp}}}}
\newcommand{\expequiv}{\ensuremath{\mathrel{\equiv_{\rm exp}}}}
\newcommand{\downgeq}[1]{\ensuremath{\mathrel{\geq_{\down}^{#1}}}}
\newcommand{\downequiv}[1]{\ensuremath{\mathrel{\equiv_{\down}^{#1}}}}
\newcommand{\weakdownequiv}[1]{\ensuremath{\mathrel{\approxeq_{\down}^{#1}}}}
\newcommand{\weakupdownequiv}[1]{\ensuremath{\mathrel{\approxeq_{\updown}^{#1}}}}
\newcommand{\updowngeq}[1]{\ensuremath{\mathrel{\geq_{\updown}^{#1}}}}
\newcommand{\updownequiv}[1]{\ensuremath{\mathrel{\equiv_{\updown}^{#1}}}}
\newcommand{\upgeq}{\ensuremath{\mathrel{\geq_{\up}}}}
\newcommand{\upequiv}{\ensuremath{\mathrel{\equiv_{\up}}}}
\newcommand{\pairs}[2]{\ensuremath{\mathrel{{#1}_{#2}}}}

\newtheorem{theorem}{Theorem}[section]
\newtheorem{corollary}[theorem]{Corollary}
\newtheorem{lemma}[theorem]{Lemma}

\newtheorem{proposition}[theorem]{Proposition}
\newdefinition{definition}[theorem]{Definition}
\newdefinition{notation}[theorem]{Notation}
\newdefinition{example}[theorem]{Example}
\newdefinition{remark}[theorem]{Remark}
\numberwithin{equation}{section}

\begin{document}

\begin{frontmatter}

\title{Structural characterizations of the navigational expressiveness of relation
algebras on a tree\tnoteref{pods-ref}}
\tnotetext[pods-ref]{A preliminary version of some of the results given here were presented at
  the Twenty-Fifth ACM Symposium on Principles of Database Systems \cite{GyssensPGF06}.}
\author[inst1]{George H.\,L.\ Fletcher\corref{cor1}}
\ead{g.h.l.fletcher@tue.nl}
\author[inst2]{Marc Gyssens}
\ead{marc.gyssens@uhasselt.be}
\author[inst3]{Jan Paredaens}
\ead{jan.paredaens@ua.ac.be}
\author[inst4]{Dirk Van Gucht}
\ead{vgucht@cs.indiana.edu}
\author[inst4]{Yuqing Wu}
\ead{yuqwu@indiana.edu}
\cortext[cor1]{Corresponding author. Phone: +31 (0)40 247 26 24. Fax: +31 (0)40 243 66 85.}
\address[inst1]{Eindhoven University of Technology, P.O. Box 513, 5600 MB Eindhoven, 
The Netherlands} 
\address[inst2]{Hasselt University and Transnational University of Limburg,
School for Information Technology, Belgium}
\address[inst3]{University of Antwerp, Belgium} 
\address[inst4]{Indiana University, Bloomington, USA} 

%\titlerunning{Semantics of XPath as Navigation Tool on a Document}
%\authorrunning{Fletcher, Gyssens, Paredaens, Van Gucht, and Wu}% 

\begin{abstract}
Given a document $D$ in the form of an unordered node-labeled tree, we study the
expressiveness on $D$ of various basic fragments of XPath, the core navigational
language on XML documents.  Working from the perspective of these languages as
fragments of Tarski's relation algebra, we give characterizations, in terms of
the structure of $D$, for when a binary relation on its nodes is definable by an
expression in these algebras. Since each pair of nodes in such a relation
represents a unique path in $D$, our results therefore capture the sets of paths
in $D$ definable in each of the fragments. We refer to this perspective on
language semantics as the ``global view.'' In contrast with this global view,
there is also a ``local view'' where one is interested in the nodes to which one
can navigate starting from a particular node in the document. In this view, we
characterize when a set of nodes in $D$ can be defined as the result of applying
an expression to a given node of $D$. All these definability results, both in
the global and the local view, are obtained by using a robust two-step
methodology, which consists of first characterizing when two nodes cannot be
distinguished by an expression in the respective fragments of XPath, and then
bootstrapping these characterizations to the desired results.
\end{abstract}


\end{frontmatter}

%======================================================================
\section{Introduction}
\label{sec-introduction}

In this paper, we investigate the expressive power of several basic fragments of
Tarski's relation algebra \cite{Tarski41} on finite tree-structured graphs.
Tarski's algebra is a fundamental tool in the field of algebraic logic 
which finds various applications in computer science
\cite{HirshHodkinson,Givant06,Maddux,TarskiGivant}.  Our investigation is
specifically motivated by the role the relation algebra plays in the study of
database query languages
\cite{GyssensSG91,tenCate:2007,Marx:2005,FletcherGLBGVW12,FletcherGLBGVW11,FletcherGWGBP09,SarathySG93}.
In particular, the algebras we consider in this paper correspond to
natural fragments of XPath.  XPath is a simple language for navigation in XML
documents (i.e., a standard syntax for representing node-labeled trees), which
is at the heart of standard XML transformation languages and other XML
technologies \cite{xpath}.  Keeping in the spirit of XML, we will continue to
speak in what follows of trees as ``documents'' and the algebras we study as
``XPath'' algebras.

XPath can be viewed as a query language in which an expression associates to
every document a binary relation on its nodes representing all navigation paths
in the document defined by that
expression~\cite{BenediktFK05,GottlobKoch,Marx:2005}.  From this query-level
perspective, several natural semantic issues have been investigated in recent
years for various fragments of XPath. These include expressibility, closure
properties, and complexity of evaluation
\cite{BenediktFK05,BenediktKoch:2009a,GottlobKP05,Marx:2005,tenCate:2007}, as well as
decision problems such as satisfiability, containment, and
equivalence~\cite{BenediktFG08,BojanczykMSS09,MiklauS04}.

Alternatively, we can view XPath as a navigational tool on a
particular given document, and study expressiveness issues from this
document-level perspective. (A similar duality exists in the
relational database model, where Bancilhon~\cite{Bancilhon78} and
Paredaens~\cite{Paredaens78} considered and characterized
expressiveness at the instance level, which, subsequently, Chandra and
Harel~\cite{ChandraHarel} contrasted with expressiveness at the
query level.) 

In this setting, our goal is to characterize, for various
natural fragments of XPath, when a binary relation on the nodes of a given
document (i.e., a set of navigation paths) is definable by an
expression in the fragment.

To achieve this goal, we develop a robust two-step methodology. The
first step consists of characterizing when two
nodes in a document cannot be distinguished by an expression in the
fragment under consideration. It turns out for those fragments we
consider that this notion of 
expression equivalence of nodes is equivalent to an appropriate
generalization of the classic notion of bisimilarity \cite{Sangiorgi}. 
The second step of our methodology
then consists of bootstrapping this result to a characterization for
when a  binary relation on the nodes of a given document is definable
by an expression in the fragment (in the sense of the previous
paragraph).

We refer to this perspective on the semantics of XPath at the document
level as the ``global view.'' In contrast with this global view, there
is also a ``local view'' which we consider. In this view, one is only
interested in the nodes to which one can navigate starting from a
particular given node in the document under consideration.  From this
perspective, a set of nodes of that document can be seen as the end
points of a set of paths starting at the given node.  For each of the
XPath fragments considered, we characterize when such a set represents the
set of \emph{all} paths starting at the given node defined by some
expression in the fragment. These characterizations are derived from
the corresponding characterizations in the ``global view,'' and turn
out to be particularly elegant in the important special case where the
starting node is the root.

In this paper, we study several natural XPath fragments.  The most expressive
among them is the \emph{XPath algebra} which permits the self, parent, and child
operators, predicates, compositions, and the boolean operators union,
intersection, and difference. (Since we work at the document level, i.e., the
document is given, there is no need to include the ancestor and descendant
operators as primitives.) We also consider the \emph{core XPath algebra}, which
is the XPath algebra without intersection and difference at the expression
level. The core XPath algebra is the adaptation to our setting of Core XPath of
Gottlob et al.~\cite{BenediktKoch:2009a,tenCate2010,GottlobKoch}.  Of both of
these algebras, we also consider various ``downward'' and ``upward'' fragments  
without the parent and child operator, respectively.   We also study ``positive''
variants of all the fragments considered, without the difference operator. 

Our strategy is to introduce and characterize generalizations of each of these
practical fragments, towards a broader perspective on relation algebras on
trees.  These generalizations are based on a simple notion of path counting, a
feature which also appears in XPath.

The robustness of the characterizations provided in this paper is
further strengthened by their feasibility.  As discussed in Section
\ref{sec-conclusions},  the global and local definability problems for
each of the XPath fragments are decidable in polynomial time.  This
feasibility hints towards efficient partitioning and reduction
techniques on both the set of nodes and the set of paths in a document.
Such techniques may fruitfully applied towards, e.g., document compression
\cite{Buneman:2003}, access control \cite{Fundulaki:2004}, and
designing indexes for query processing
\cite{FletcherGWGBP09,KaushikSBG02,MiloS99}.

We proceed in the paper as follows.
In Section \ref{subsec-docandnav}, we formally define documents and the algebras, and then in 
Section \ref{subsec-sig}, we define a notion of ``signatures'' which 
will be essential in the sequel.
In Section \ref{sec-distinguishability}, we define the semantic and syntactic notions of node
distinguishability necessary to obtain our desired structural characterizations.
In the balance of the paper, we apply our two-step methodology to link
semantic expression equivalence in the languages to appropriate structural syntactic
equivalence notions.
In particular, we give structural characterizations, under both the global and
local views, 
\begin{itemize}
    \item of ``strictly'' (Section \ref{sec-strictlydownward}) and ``weakly'' 
        (Section \ref{sec-weaklydownward}) downward languages, and their positive variants;
    \item  of upward languages and their positive variants (Section \ref{sec-upward});
        and,
    \item of languages with both downward and upward navigation, and their positive variants
        (Section \ref{sec-general}).
\end{itemize}
Along the way, we also establish the equivalence of some of these fragments, using
the structural characterizations obtained.
We conclude in Section~\ref{sec-conclusions} with a discussion of some
ramifications of our results and directions for further study.

%======================================================================
\section{Documents and navigation}
\label{subsec-docandnav}

In this paper, we are interested in navigating over documents in the
form of unordered labeled trees. Formally, we denote such a document
as $D=(V,\textit{Ed},r,\lambda)$, with $D$ the document name, $V$ the
set of nodes of the tree, $\textit{Ed}$ the set of edges of the tree,
$r$ the root of the tree, and $\lambda: V\to\mathcal{L}$ a function
assigning to each node a label from some infinite set of labels
$\mathcal{L}$. 

\begin{example}
\label{ex-document}

Figure~\ref{fig-document} shows an example of a document that will be
used throughout the paper.  Here, $r=v_1$ is the root of the tree with label
$\lambda(v_1) = a$.

\end{example}

\begin{figure}
\begin{center}
\resizebox{0.9\textwidth}{!}{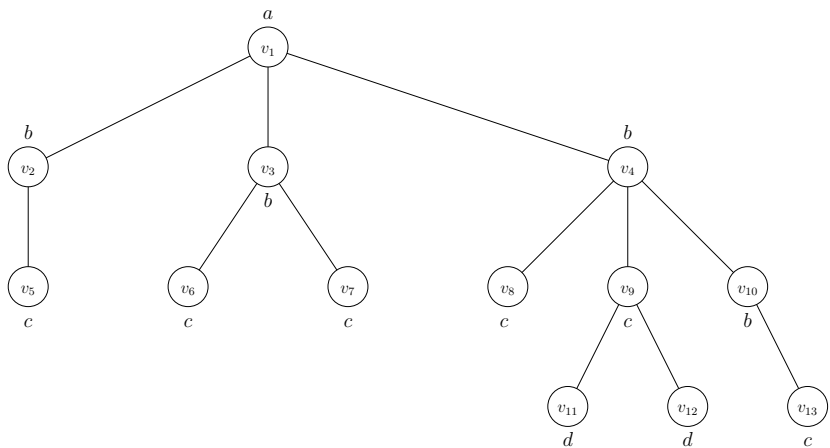}
\end{center}
\caption[a]{Example document.}
\label{fig-document}
\end{figure}

We next define a set of operations on documents, as tabulated
in Table~\ref{tab-binops}. The left column shows the syntax of the
operation, and the right column its semantics, given a document
$D=(V,\textit{Ed},r,\lambda)$. Notice that, in each case, the result
is a binary relation on the nodes of the document.

\begin{table}
    \centering
\caption{Binary operations on documents. The left column shows the
  syntax of the operation, and the right column its semantics, given a
  document $D=(V,\textit{Ed},r,\lambda)$. Below, $\ell$ is a label in
  $\mathcal{L}$ and $k\geq 1$ a natural number. Furthermore, in the
  recursive definitions, $e$, $e_1$, and $e_2$ represents expressions
  built with the operations.}
\label{tab-binops}
\small{
$$\begin{array}{|c|c|}
\hline
\textrm{Syntax}&\textrm{Semantics}\\
\hline
\emptyset&\emptyset(D)=\emptyset\\
\varepsilon&\varepsilon(D)=\{(v,v)\mid v\in V\}\\
\hat{\ell}&\hat{\ell}(D)=\{(v,v)\mid v\in V\ \&\ \lambda(v)=\ell\}\\
\down&\down(D)=\textit{Ed}\\
\up&\up(D)=\textit{Ed}^{-1}\\
%\righta&\righta(D)=\{(v,w)\mid v\neq w\ \&\ (\exists
%u)((u,v)\in\textit{Ed}\ \&\ (u,w)\in\textit{Ed})\}\\
\pi_1(e)&\pi_1(e)(D)=\{(v,v)\mid(\exists w)(v,w)\in e(D)\}\\
\pi_2(e)&\pi_2(e)(D)=\{(w,w)\mid(\exists v)(v,w)\in e(D)\}\\
e^{-1}&e^{-1}(D)=e(D)^{-1}\\
\ch{k}(e)&\ch{k}(e)(D)=\{(v,v)\mid v\in V\ \&\ |\{w\mid
(v,w)\in\textit{Ed}\ \&\ (w,w)\in\pi_1(e)(D)|\geq k\}\\
e_1/e_2&e_1/e_2(D)=\{(u,w)\mid (\exists v)((u,v)\in e_1(D)\ \&\ (v,w)\in
e_2(D))\}\\
%e_1[e_2]&e_1[e_2](D)=\{(u,v)\mid (\exists w)((u,v)\in e_1(D)\ \&\ (v,w)\in
%e_2(D))\}\\
e_1\cup e_2&e_1\cup e_2(D)=e_1(D)\cup e_2(D)\\
e_1\cap e_2&e_1\cap e_2(D)=e_1(D)\cap e_2(D)\\
e_1 - e_2&e_1 - e_2(D)=e_1(D) - e_2(D)\\
\hline
\end{array}$$
}
\end{table}

The \emph{basic algebra}, denoted $\BL$, is the language consisting of
all expressions built from $\emptyset$, $\varepsilon$, $\hat{\ell}$
with $\ell\in\mathcal{L}$, composition (``$/$''), and union
(``$\cup$''). The basic algebra $\BL$ can be extended by adding some
of the other operations in Table~\ref{tab-binops}, which we call
\emph{nonbasic}. If $E$ is a set of nonbasic operations, then $\BL(E)$
denotes the algebra obtained by adding the operations in $E$ to the
basic algebra $\BL$. When writing expressions, we assume that unary
operations take precedence over binary operations, and that
composition takes precedence over the set operations.

Notice that we do not consider transitive closure operations such as
the descendant (``$\down^\ast$'') or ancestor
(``$\up^\ast$'') operations of XPath. The reason for this is
that, in this paper, we only consider navigation within a single,
given document. 

\begin{example}
\label{ex-operations}

Consider the document $D$ in Figure~\ref{fig-document}. Let $e$ be the
expression $\up/\pi_1(\down/\hat{b}/\down/\hat{c})-\ch{2}(\varepsilon)/\up$ in
the language $\BL(\down,\up,\pi_1,\ch{2},-)$ (or, for that matter, in any language
$\BL(E)$ with  $\{\down,\up,\pi_1,\ch{2},-\}\subseteq E$). Then,
$e(D)=\{(v_2,v_1),(v_8,v_4),(v_{10},v_4)\}$. 

\end{example}

Not all the above operations are primitive, however. For instance,
intersection (``$\cap$'') is expressible as soon a set
difference (``$-$'') is expressible, since, for any two sets $A$ and
$B$, $A\cap B=A-(A-B)$. Even more eliminations are possible in the
following setting.

\begin{proposition}
\label{prop-eliminate}

Let $E$ be a set of nonbasic operations containing set difference
``$-$'') or intersection
(``$\cap$'') for which ``$\down$'' and ``$\up$'' are both
contained in $E$ or both not contained in $E$. Then, for each
expression $e$ in $\BL(E)$, there is an equivalent expression in
$\BL(E-\{\pi_1,\pi_2,\mathstrut^{-1}\})$.

\end{proposition}

\begin{proof}
First, we eliminate both projections using the identities
\begin{eqnarray*}
\pi_1(e)&=&(e/e^{-1})\cap\varepsilon\rm;\\
\pi_2(e)&=&(e^{-1}/e)\cap\varepsilon\rm.
\end{eqnarray*}
Hence, each expression in
$\BL(E)$ can be replaced by an equivalent expression in
$\BL((E\cup\{\mathstrut^{-1}\})-\{\pi_1,\pi_2\})$. It remains to
show that we can eliminate inverse (``$\mathstrut^{-1}$''). This
follows from the following identities. In these,
$D=(V,\textit{Ed},r,\lambda)$ is a document, $\ell\in\mathcal{L}$ is a
label, $k\geq 1$ is a natural number, and $e$, $e_1$ and $e_2$ are
expressions in $\BL(E)$. 
\begin{itemize}
\item $\emptyset^{-1}(D)=\emptyset(D)$;
\item $\varepsilon^{-1}(D)=\varepsilon(D)$;
\item $\hat{\ell}^{-1}(D)=\hat{\ell}(D)$;
\item $\down^{-1}(D)=\up(D)$;
\item $\up^{-1}(D)=\down(D)$;
%\item $\righta^{-1}(D)=\righta(D)$;
\item $(e^{-1})^{-1}(D)=e(D)$;
%\item $\pi_1(e)^{-1}(D)=\pi_1(e)(D)$;
%\item $\pi_2(e)^{-1}(D)=\pi_2(e)(D)$;
\item $(e_1/e_2)^{-1}(D)=e_2^{-1}/e_1^{-1}(D)$;
\item $\ch{k}(e)^{-1}(D)=\ch{k}(e)(D)$;
\item $(e_1\cup e_2)^{-1}(D)=e_1^{-1}\cup e_2^{-1}(D)$;
\item $(e_1\cap e_2)^{-1}(D)=e_1^{-1}\cap e_2^{-1}(D)$;
\item $(e_1 - e_2)^{-1}(D)=e_1^{-1} - e_2^{-1}(D)$.
\end{itemize}
\end{proof}

Notice that in a language with both upward (``$\up$'') and downward
(``$\down$'') navigation, the identities
$\pi_1(e)(D)=\pi_2(e^{-1})(D)$ and
$\pi_2(e)(D)=\pi_1(e^{-1})(D)$ imply that one projection
operation can be eliminated in favor of the other. Hence, it does not
make sense to consider the projection operations separately.

Some counting operations (``$\ch{k}(e)$'') can also be
simulated. One can easily verify the following.

\begin{proposition}
\label{prop-counting}

Let $D=(V,\textit{Ed},r,\lambda)$ be a document. 
%and let $\ell\in\mathcal{L}$ be a label. 
Then,
\begin{enumerate}

\item $\ch{1}(e)(D)=\pi_1(\down/e)(D)$;

\item $\ch{2}(e)(D)=
  \pi_1(\down/(\pi_1(e)/\up/\down/\pi_1(e)-\varepsilon))(D)$; and
%  \pi_1(\down/(\pi_1(e)\righta/\pi_1(e)))$; and

\item $\ch{3}(e)(D)= \pi_1(\down/
  ((\pi_1(e)/\up/\down/\pi_1(e)-\varepsilon)/
  (\pi_1(e)/\up/\down/\pi_1(e)-\varepsilon)- \varepsilon)(D)$\\
%  $\phantom{\ch{3}(D)}=\pi_1(\down/(\pi_1(e)\righta/\pi_1(e)/\righta/\pi_1(e)-\varepsilon)$.   

\end{enumerate}
\end{proposition}

\begin{example}
\label{ex-primitive}

Consider again the expression $e \mathrel{:=}
\up/\pi_1(\down/\hat{b}/\down/\hat{c})-\ch{2}(\varepsilon)/\up$ of
Example~\ref{ex-operations}. Using 
Proposition~\ref{prop-counting}, and making some straightforward
simplifications, we can rewrite $e$ as
$\up/\pi_1(\down/\hat{b}/\down/\hat{c})-
\pi_1(\down/(\up/\down-\varepsilon))/\uparrow$, an expression of
$\BL(\down,\up,\pi_1,-)$. Alternatively, one can use
Proposition~\ref{prop-eliminate} and the techniques exhibited in its
proof to rewrite $e$ as
$$\up/(\down/\hat{b}/\down/\hat{c}/\up/\down) -
\down/((\up/\down-\varepsilon)/(\up/\down-\varepsilon)
\cap\varepsilon)/\up,$$
an expression in $\BL(\down,\up,\cap,-)$. Finally, we invite the
reader to verify that $e$ can also be rewritten as 
$$\pi_1(\varepsilon-\pi_1(\down/(\up/\down-
\varepsilon)))/\up/\pi_1(\down/\hat{b}/\down/\hat{c}),$$
also an expression of $\BL(\down,\up,\pi_1,-)$.

\end{example}

We shall call the language
$\BL(\down,\up,\pi_1,\pi_2,.^{-1},\cap,-)$, which by
Proposition~\ref{prop-eliminate} is equivalent to
$\BL(\down,\up,-)$, the \emph{XPath algebra}.\footnote{
Note that the XPath algebra corresponds to the (full) relation algebra of Tarski
\cite{Tarski41}, adapted to our setting (cf.\ \cite{tenCate:2007}).} This is justified by the following result.

\begin{proposition}
\label{prop-XPath}

Given a single document $D=(V,\textit{Ed},r,\lambda)$, the
XPath algebra is equivalent to XPath.

\end{proposition}

\begin{proof}
    Notice that $\ell$ in XPath \cite{xpath} is simulated by $\down/\hat{\ell}$ in
the XPath algebra.  Furthermore, $\hat{\ell}$ in the XPath algebra is
simulated by  $\varepsilon[\textrm{label}=\ell]$ in XPath.
The proof is complete if, for each predicate $P$ in XPath, there exists
an XPath algebra expression $e$ such that $e(D)=\{(n,n)\mid n\in P(D)\}$.
This is proved by structural induction:
\begin{enumerate}

\item if $P$ is an XPath expression without predicates, then take
    $e := \pi_1(f)$, with  $f$ the XPath algebra expression
    obtained from $P$ by replacing everywhere $\ell$ by
    $\down/\hat{\ell}$.

\item if $P$ is $\textrm{label}= \ell$, then take $e:= \hat{\ell}$.

\item if $P$ is $\lnot Q$, with $Q$ an XPath predicate, then take
  $e:= \varepsilon - f$, with $f$ the 
    XPath algebra expression corresponding to $Q$.

\item if $P$ is $Q_1\land Q_2$, with $Q_1$ and $Q_2$ XPath predicates,
  then take $e:= f_1\cap f_2$, with $f_1$ and $f_2$ the 
    XPath algebra expressions corresponding to $Q_1$ and $Q_2$,
    respectively. 

\item if $P$ is $Q_1\lor Q_2$, with $Q_1$ and $Q_2$ XPath predicates,
  then take $e:= f_1\cup f_2$, with $f_1$ and $f_2$ the 
    XPath algebra expressions corresponding to $Q_1$ and $Q_2$,
    respectively.

\end{enumerate}
\end{proof} 

Besides the \emph{standard languages} $\BL(E)$, with $E$ a set of
nonbasic operations, we also consider the so-called \emph{core languages}
$\CL(E)$. More concretely, $\CL(E)$ is defined recursively in the same way as
$\BL(E-\{\cap,-\})$, except that in expressions of the form 
$\pi_1(f)$, and $\pi_2(f)$, $f$ may be a boolean combination of
expressions of the language using union and the operations in
$E\cap\{\cap,-\}$, rather than just an expression of the language.

The above terminology is inspired by the fact that
$\CL(\down,\up,\pi_1,\pi_2,-,\cap)$, the language which we call the 
\emph{core XPath algebra\/}, is the adaptation to our setting of
Core XPath of Gottlob and Koch \cite{GottlobKoch}.

\begin{example}
\label{ex-core}

Continuing with Example~\ref{ex-primitive}, we consider again the
expression $e \mathrel{:=}
\up/\pi_1(\down/\hat{b}/\down/\hat{c})-\ch{2}(\varepsilon)/\up$
of Example~\ref{ex-operations}. Obviously, there is no core language of
which $e$ is an expression, as set difference (``$-$'') occurs at the
outer level, and not in a subexpression $f$ which in turn is embedded
in a subexpression of the form $\pi_1(f)$ or
  $\pi_2(f)$. However, in Example~\ref{ex-primitive}, the expression
  $e$ has been shown to be equivalent to
$$\pi_1(\varepsilon-\pi_1(\down/(\up/\down-
\varepsilon)))/\up/\pi_1(\down/\hat{b}/\down/\hat{c}),$$
which is an expression of $\CL(\down,\up,\pi_1,-,\cap)$, and hence
also of the core XPath algebra.

\end{example}

Given a set of nonbasic operators $E$, an expression in $\BL(E)$ can
in general \textit{not} be converted to an equivalent expression in
$\CL(E)$, however, as will follow from the results of this paper, even
though there are exceptions (Section~\ref{sec-strictlydownward},
Theorem~\ref{theo-downwardcorecoincides}).

\begin{table}
\centering
\caption{Languages studied in this paper.}\label{table:language-summary}
\footnotesize{
    \begin{tabular}{|>{\centering}m{.3\textwidth}<{\centering}>{\centering}m{.6\textwidth}<{\centering}|}
    \hline
    {\em Language} & {\em Relation algebra fragment} \tabularnewline
    \hline
    strictly downward (core) XPath algebra with counting up to $k$ & 
    \begin{multline*}\BL(\down,\pi_1,\ch{1}(.),\ldots,\ch{k}(.),-) \\ = \CL(\down,\pi_1,\ch{1}(.),\ldots,\ch{k}(.),-)\end{multline*}\tabularnewline
    \hline
    strictly downward (core) positive XPath algebra & $\BL(\down,\pi_1,\cap) = \CL(\down,\pi_1,\cap)$  \tabularnewline
    \hline
    weakly downward (core) XPath algebra with counting up to $k$ & 
    \begin{multline*}\BL(\down,\pi_1,\pi_2,\ch{1}(.),\ldots,\ch{k}(.),-) \\ = \CL(\down,\pi_1,\pi_2,\ch{1}(.),\ldots,\ch{k}(.),-)\end{multline*}\tabularnewline
    \hline
    weakly downward (core) positive XPath algebra & 
    $\BL(\down,\pi_1,\pi_2) = \BL(\down,\pi_1,\pi_2,\cap) =  \CL(\down,\pi_1,\pi_2,\cap)$ \tabularnewline
    \hline
    strictly upward (core) XPath algebra & $\BL(\up,\pi_1,-) = \CL(\up,\pi_1,-)$ \tabularnewline
    \hline
    strictly upward (core) positive XPath algebra & $\BL(\up,\pi_1,\cap) = \CL(\up,\pi_1,\cap)$\tabularnewline
    \hline
    weakly upward languages &{\em see Section \ref{subsec-weaklyupward}} \tabularnewline
    \hline
    XPath algebra & $\BL(\down,\up,\pi_1,\pi_2,.^{-1},\cap,-) = \BL(\down,\up,-)$ \tabularnewline
    \hline
    XPath algebra with counting up to $k$ & $\BL(\down,\up,\ch{1}(.),\ldots,\ch{k}(.),-)$  \tabularnewline
    \hline
    core XPath algebra & $\CL(\down,\up,\pi_1,\pi_2,-,\cap)$ \tabularnewline
    \hline
    core XPath algebra with counting up to $k$ & $\CL(\down,\up,\pi_1,\pi_2,\ch{1}(.),\ldots,\ch{k}(.),-)$ \tabularnewline
    \hline
    (core) positive XPath algebra (\cite{WuGGP11})& $\BL(\down,\up,\cap) =\BL(\down,\up,\pi_1,\pi_2)
    = \CL(\down,\up,\pi_1,\pi_2,\cap)$ \tabularnewline
    \hline
\end{tabular}
}
\end{table}

Table \ref{table:language-summary} gives an overview of the various relation
algebra fragments we investigate below.

To conclude this section, we observe that, given a document and an
expression, we have defined the semantics of that expression as a
binary relation over the nodes of the document, i.e., as a set of
pair of nodes. From the perspective of navigation, however, it is
useful to be able to say that an expression allows one to navigate
from one node of the document to another. For this purpose,
we introduce the following notation.

\begin{definition}
\label{def-efunction}

Let $e$ be an arbitrary expression, and let
$D=(V,\textit{Ed},r,\lambda)$ be a document. For $v\in V$, $e(D)(v)
\mathrel{:=} \{w\mid (v,w)\in e(D)\}$.
\end{definition}

Definition~\ref{def-efunction} reflects the ``local'' perspective of
an expression working on particular nodes of a document, rather than
the ``global'' perspective of working on an entire document.

\begin{example}
\label{ex-local}

Consider again the expression
$e\ass \up/\pi_1(\down/\hat{b}/\down/\hat{c})-\ch{2}(\varepsilon)/\up$ of
Example~\ref{ex-operations}. We have established that, for the
document $D$ in Figure~\ref{fig-document},
$e(D)=\{(v_2,v_1),(v_8,v_4),(v_{10},v_4)\}$. Hence,
$e(D)(v_8)=\{v_4\}$ and $e(D)(v_1)=\emptyset$.

\end{example}
%======================================================================
\section{Signatures}
\label{subsec-sig}

Given a pair of nodes in a document, there is a unique path in that
document (not taking into account the direction of the edges) to
navigate from the first to the second node, in general by going a
few steps upward in the tree, and then going a few steps downward. We
call this the \emph{signature} of that pair of nodes, and shall
formally represent it by an expression in $\BL(\down,\up)$.

\begin{definition}
\label{def-signature}

Let $D=(V,\textit{Ed},r,\lambda)$ be a document, and let $v,w\in
V$. The \emph{signature} of the pair $(v,w)$, denoted $\sig(v,w)$, is
the expression in $\BL(\down,\up)$ that is recursively
defined, as follows:
\begin{itemize}

\item if $v=w$, then $\sig(v,w) \mathrel{:=} \varepsilon$;

\item if $v$ is an ancestor of $w$, and $z$ is the child of $v$ on the
  path from $v$ to $w$, then $\sig(v,w) \mathrel{:=} \down/\sig(z,w)$;

\item otherwise\footnote{In particular, $v\neq r$.}, if $z$ is the
  parent of $v$, then $\sig(v,w) \mathrel{:=} \up/\sig(z,w)$.

\end{itemize}
\end{definition}

Given nodes $v$ and $w$ of a document $D=(V,\textit{Ed},r,\lambda)$,
we denote by $\top(v,w)$ the unique node on the undirected path from
$v$ to $w$ that is an ancestor of both $v$ and $w$. Clearly,
$$\sig(v,w)=\sig(v,\top(v,w))/\sig(\top(v,w),w)=\up^m/\down^n,$$
where $m$, respectively $n$, is the distance from $\top(v,w)$ to $v$,
respectively $w$; and, for an expression $e$ and a natural number
$i\geq 1$, $e^i$ denotes the $i$-fold composition of
$e$.\footnote{Here, and elsewhere in this paper, equality between
  expressions must be interpreted at the semantic and not at the
  syntactic level, i.e., for two expressions $e_1$ and $e_2$ in one of
  the languages considered here, $e_1=e_2$ means that, for each
  document~$D$, $e_1(D)=e_2(D)$.} (We put $e^0\mathrel{:=} \varepsilon$.)

The signature of a pair of nodes of a document can be seen as a
description of the unique path connecting these nodes, but also as an
expression that can be applied to the document under consideration. We
shall often exploit this duality.

\begin{example}
\label{ex-signature}

For the document $D$ in Figure~\ref{fig-document},
$\sig(v_1,v_1)=\varepsilon$, 
$\sig(v_1,v_2)=\down$, 
$\sig(v_6,v_4)=\up^2/\down$, and
$\sig(v_{11},v_5)=\up^3/\down^2$.  We have that 
$$\begin{array}{lcl}
\sig(v_{11},v_5)(D)&=&\{(v_{11},v_5),(v_{12},v_5),(v_{13},v_5),
(v_{11},v_6),(v_{12},v_6),(v_{13},v_6),\\
&&\phantom{\{}(v_{11},v_7),(v_{12},v_7),(v_{13},v_7),
(v_{11},v_8),(v_{12},v_8),(v_{13},v_8),\\
&&\phantom{\{}(v_{11},v_9),(v_{12},v_9),(v_{13},v_9),
(v_{11},v_{10}),(v_{12},v_{10}),(v_{13},v_{10})\}.
\end{array}$$
Notice that not each pair in the result has the same signature as
$(v_{11},v_5)$. For instance, $\sig(v_{11},v_8)=\up^2/\down$ and
$\sig(v_{11},v_9)=\up$.

\end{example}

Now, let $(v_1,w_1)$ and $(v_2,w_2)$ be two pairs of nodes in a
document~$D=(V,\textit{Ed},r,\lambda)$. We say that $(v_1,w_1)$ 
\emph{subsumes\/} $(v_2,w_2)$, denoted $(v_1,w_1)\siggeq
(v_2,w_2)$, if $(v_2,w_2)$ is in $\sig(v_1,w_1)(D)$.  We say that
$(v_1,w_1)$ are $(v_2,w_2)$ \emph{congruent}, denoted
$(v_1,w_1)\sigequiv (v_2,w_2)$, if $(v_1,w_1)\siggeq (v_2,w_2)$ and
$(v_2,w_2)\siggeq (v_1,w_1)$. It can be easily seen that, in this case,
$\sig(v_1,w_1)=\sig(v_2,w_2)$. Informally speaking, the path
from~$v_1$ to~$w_1$ has then the same shape as the path 
from~$v_2$ to~$w_2$.

\begin{example}
\label{ex-subsumption}

Consider again Example~\ref{ex-signature}. Clearly, $(v_{11},v_5)$
subsumes each pair of nodes in $\sig(v_{11},v_5)(D)$,
e.g., $(v_{11},v_5)\siggeq (v_{12},v_6)$ and $(v_{11},v_5)\siggeq
(v_{12},v_9)$.  Notice that also $(v_{12},v_6)\siggeq (v_{11},v_5)$, and
hence $(v_{11},v_5)\sigequiv (v_{12},v_6)$. However,
$(v_{12},v_9)\not\siggeq (v_{11},v_5)$. Hence, these pairs are not congruent.

\end{example}

By definition, subsumption is captured by the
``$\sig$'' expression. One may wonder if there also exists an
expression that precisely captures congruence. This is the
case in the following situations.

\begin{proposition}
\label{prop-congruence}

Let $D=(V,\textit{Ed},r,\lambda)$ be a document and let
$v_1,v_2,w_1,w_2\in V$. Then,

\begin{enumerate}

\item\label{congruence-1} if $v_1$ is an ancestor of $w_1$ or vice
  versa, $(v_1,w_1)\sigequiv (v_2,w_2)$ if and only if
  $(v_2,w_2)\in\sig(v_1,w_1)(D)$;

\item\label{congruence-2} otherwise, let
  $\sig(v_1,w_1)=\up^m/\down^n$. Then, as $m\geq 1$ and $n\geq 1$,
%
%\begin{enumerate}
%
%\item\label{congruence-2a} 
$(v_1,w_1)\sigequiv (v_2,w_2)$ if and only if
  $(v_2,w_2)\in\up^m/\down^n-\up^{m-1}/\down^{n-1}(D)$.
%
%\item\label{congruence-2b} $(v_1,w_1)\sigequiv (v_2,w_2)$ if and only if
%  $(v_2,w_2)\in\up^{m-1}/\righta/\down^{n-1}(D)$.
%
%\end{enumerate}
\end{enumerate}
\end{proposition}

\begin{proof}
\begin{enumerate}
\item As the ``only if'' is trivial, it
suffices to consider the ``if,'' which follows from a straightforward
induction argument.

\item As the ``only if'' is straightforward, we only consider the
  ``if.'' Let $t_2 \mathrel{:=} \up^m(D)(v_2)$. Since $w_2\in\down^n(D)(t_2)$,
$t_2$ is a common ancestor. Let $v_2'$ and $w_2'$ be the children of
$t_2$ on the path to $v_2$ and $w_2$, respectively.
If $v_2'=w_2'$, then
$(v_2,w_2)\in\up^{m-1}/\down^{n-1}(D)$, a contradiction, Hence,
$v_2'\neq w_2'$ and $t_2=\top(v_2,w_2)$, and
$\sig(v_2,w_2)=\up^m/\down^n= \sig(v_1,w_1)$. 

\end{enumerate}
\end{proof}

For later use, but also because of their independent interest, 
we finally note the following fundamental properties of
subsumption and congruence.

\begin{proposition}
\label{prop-subsumption}

Let $v$, $w$, $v_1$, $w_1$, $z_1$, $v_2$, $w_2$, and $z_2$ be nodes of
a document~$D$. Then the following properties hold.

\begin{enumerate}

\item $(v,v)\siggeq(w,w)$.

\item \label{subsumption-2}$(v_1,w_1)\siggeq (v_2,w_2)$ implies that
  $(w_1,v_1)\siggeq (w_2,v_2)$.

\item \label{subsumption-3}If $\top(v_1,z_1)$ is also an ancestor of
  $w_1$, then $(v_1,w_1)\siggeq (v_2,w_2)$ and
  $(w_1,z_1)\siggeq (w_2,z_2)$ imply that $(v_1,z_1)\siggeq (v_2,z_2)$.

\item\label{subsumption-4} All properties above also hold when
  subsumption is replaced by congruence, provided that, in
  item~\ref{subsumption-3}, the condition ``$\top(v_2,z_2)$ is also an
  ancestor of $w_2$'' is added.

\end{enumerate}
\end{proposition}

\begin{proof}
All properties are straightforward, except for
Property~\ref{subsumption-3}. So, assume that $(v_1,w_1)\siggeq (v_2,w_2)$ and
  $(v_1,z_1)\siggeq (v_2,z_2)$. Hence, $(v_2,w_2)\in\sig(v_1,w_1)(D)$ and
$(v_2,z_2)\in\sig(v_1,z_1)(D)$, as a consequence of which
$$(v_2,z_2)\in\sig(v_1,w_1)/\sig(w_1,z_1)(D).$$
For the sake of abbreviation, let $t_1\ass\top(v_1,w_1)$ and $u_1\ass
\top(w_1,z_1)$. Using these nodes, we can write
$$\sig(v_1,w_1)/\sig(w_1,z_1)=
\sig(v_1,t_1)/\sig(t_1,w_1)/\sig(w_1,u_1)/\sig(u_1,z_1),$$ which
is equal to $\sig(v_1,s_1)/\sig(s_1,z_1)$,
where $s_1$ is the higher of~$t_1$ and~$u_1$ in~$D$. Notice that
$s_1$ is a common ancestor of~$v_1$ and~$z_1$, as a consequence of
which it is also an ancestor of~$\top(v_1,z_1)$, the least common
ancestor of~$v_1$ and~$z_1$.  By assumption, $\top(v_1,z_1)$ is a
common ancestor of~$v_1$, $w_1$, and~$z_1$, and hence also of
$\top(v_1,w_1)$ and $\top(w_1,z_1)$, the highest of which is~$s_1$.  Thus,
$s_1=\top(v_1,z_1)$, and, therefore,
$\sig(v_1,s_1)/\sig(s_1,z_1)=\sig(v_1,z_1)$. In 
summary, $(v_2,z_2)\in\sig(v_1,z_1)(D)$, and hence
$(v_1,z_1)\siggeq(v_2,z_2)$.
\end{proof}

Observe that the condition in Proposition~\ref{prop-subsumption},
(\ref{subsumption-3}), is necessary for that part of the
proposition to hold, as shown by the following counterexample.

\begin{figure}
\begin{center}
\resizebox{0.8\textwidth}{!}{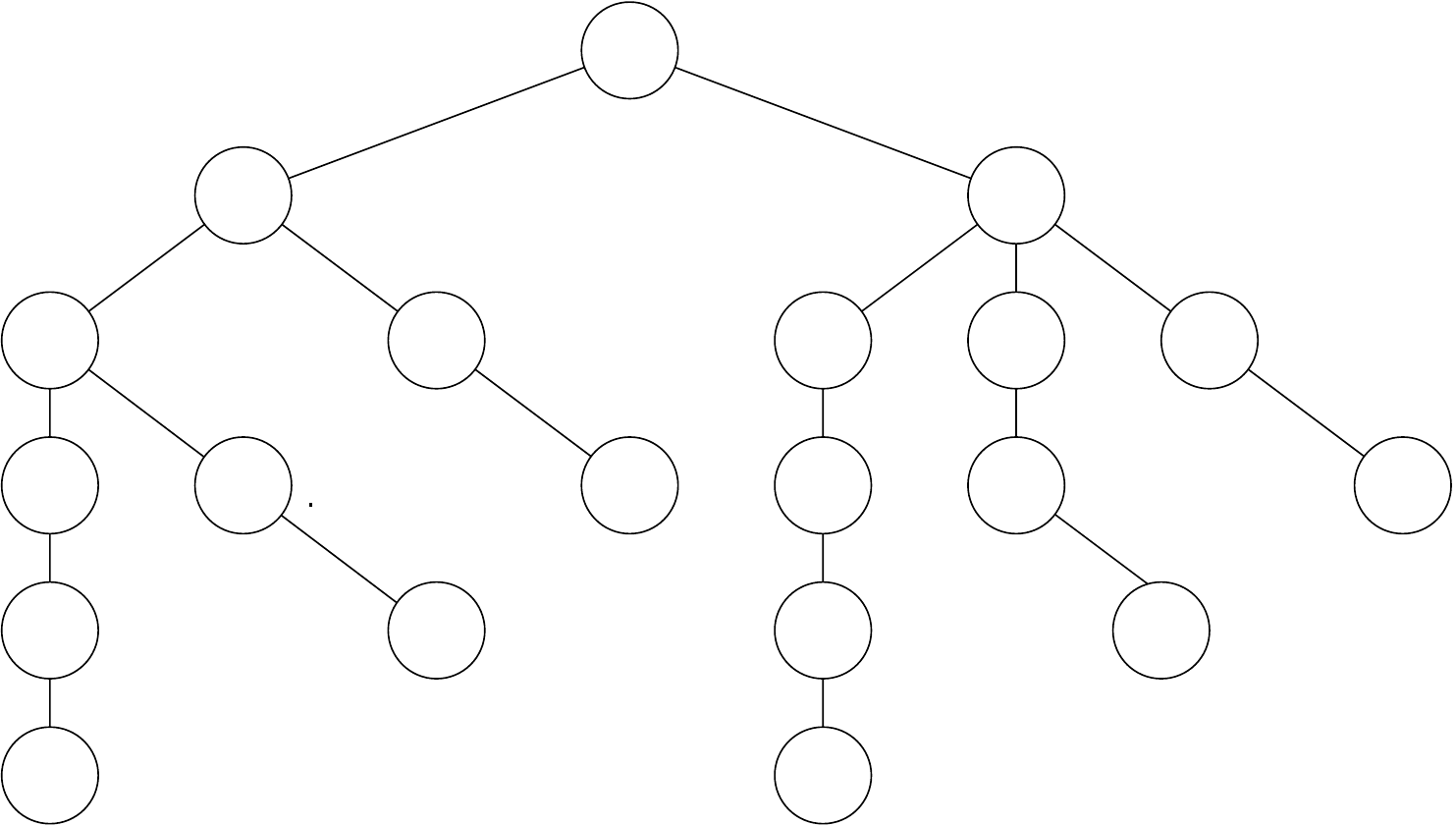}
\end{center}
\caption[a]{Document of Example~\ref{ex-nottrans}.}
\label{fig-bp-nottrans}
\end{figure}

\begin{example}
\label{ex-nottrans}

Consider the document in Figure~\ref{fig-bp-nottrans}. Labels have
been omitted, because they are not relevant in this discussion. (We
assume all nodes have the same label.) 
Observe that
$(v_1,w_1)\sigequiv(v_2,w_2)$ and $(w_1,z_1)\sigequiv
(w_2,z_2)$. However, $\top(v_1,z_1)$ is \emph{not\/} an ancestor of
$w_1$, hence, Proposition~\ref{prop-subsumption},
(\ref{subsumption-3}), is \emph{not\/} applicable. We see that,
indeed, $(v_1,z_1)$ does \emph{not\/} subsume $(v_2,z_2)$, let alone
that $(v_2,z_2)$ and $(v_2,w_2)$ would be congruent.
\end{example}

%======================================================================
\section{Distinguishability of nodes in a document}
\label{sec-distinguishability}

We wish to link the distinguishing power of a navigational language on
a document to syntactic conditions which can readily be verified on
that document. As argued before, the action of an expression on a
document can be interpreted as (1) returning pairs of nodes, or (2)
given a node, returning the set of nodes that can be reached from that
node. We shall refer to the first interpretation as the \emph{pairs
  semantics}, and to the second interpretation as the \emph{node
  semantics}. In this section, we propose suitable semantic and syntactic
notions of distinguishability for the node semantics.

\subsection{Distinguishability of nodes at the semantic level}
\label{subsec-semanticdistinguishability}

We propose the following distinguishability criterion based on the
emptiness or nonemptiness of the set of nodes that can be reached by
applying an arbitrary expression of the language under consideration.

\begin{definition}
\label{def-ee}

Let $L$ be one of the languages considered in Section~\ref{subsec-docandnav}.
Let $D=(V,\textit{Ed},r,\lambda)$ be a document, and let
$v_1,v_2\in V$. Then, 
\begin{enumerate}

\item $v_1$ and $v_2$ are \emph{expression-related}, denoted $v_1
  \expgeq v_2$, if, for each expression $e$ in $L$,
  $e(D)(v_1)\neq\emptyset$ implies $e(D)(v_2)\neq\emptyset$; and

\item $v_1$ and $v_2$ are \emph{expression-equivalent}, denoted $v_1
  \expequiv v_2$, if $v_1\expgeq v_2$ and $v_2\expgeq v_1$.

\end{enumerate}
\end{definition}

In principle, we should have reflected the language under
consideration in the notation for expression-equivalence. As the
language under consideration will always be clear from the context, we
chose not to do so in order to avoid overloaded notation.

The following observation is useful.

\begin{proposition}
\label{prop-notempty-empty}

Let $E$ be a set of nonbasic operations containing 
first projection (``$\pi_1$'') and set difference (``$-$'').
Consider expression-equivalence with respect to~$\BL(E)$. Let
$D=(V,\textit{Ed},r,\lambda)$ be a document, and let 
$v_1,v_2\in V$. Then, $v_1\expequiv v_2$ if and only if $v_1\expgeq v_2$.

\end{proposition}

\begin{proof}
Assume that $v_2\not\expgeq v_1$. Then there exists an expression $f$
in $\BL(E)$ such that $f(D)(v_2)\neq\emptyset$ and
$f(D)(v_1)=\emptyset$. Now consider $e\ass
\pi_1(\varepsilon-\pi_1(f))$. Clearly, $e(D)(v_2)=\emptyset$ and
$e(D)(v_1)\neq\emptyset$, hence $v_1\not\expgeq v_2$. By
contraposition, $v_1\expgeq v_2$ implies $v_2\expgeq v_1$, and hence
also $v_1\expequiv v_2$.
\end{proof}

%-------------------------------------------------------------------------
\subsection{Distinguishability of nodes at the syntactic level}
\label{subsec-syntacticdistinguishability}

Our syntactic criterion of distinguishability is based on the
similarity of the documents locally around the nodes under
consideration. In order to decide this similarity, we shall consider a
hierarchy for the degree of coarseness by which we compare the
environments of those nodes. We shall also consider variants for the
cases where from the given nodes of the document we (1) only look
downward; (2) only look upward; or (3) look in both directions.
%.........................................................................
\subsubsection{Downward distinguishability}
\label{subsubsec-downwardcase}

For the downward case, we consider the following syntactic notions of
distinguishability of nodes. They are all defined recursively on the
height of the first node.

\begin{definition}
\label{def-kdownwardequivalent}

Let $D=(V,\textit{Ed},r,\lambda)$ be a document, let
$v_1,v_2\in V$, and let $k\geq 1$. Then, $v_1$ and $v_2$ are
\emph{downward-$k$-equivalent}, denoted $v_1\downequiv{k} v_2$,
if
\begin{enumerate}

\item $\lambda(v_1)=\lambda(v_2)$;

\item for each child $w_1$ of $v_1$, there exists a child
$w_2$ of $v_2$ such that $w_1\downequiv{k} w_2$, and vice versa;

\item for each child $w_1$ of $v_1$ and $w_2$ of $v_2$ such that
  $w_1\downequiv{k} w_2$, $\min(|\bar{w}_1|,k)=\min(|\bar{w}_2|,k)$,
  where, for $i=1,2$, $\bar{w}_i$ is the set of all siblings of $w_i$
  (including $w_i$ itself) that are downward $k$-equivalent to
  $w_i$.\footnote{For a set $A$, $|A|$ denotes the cardinality of
    $A$.}

\end{enumerate}
\end{definition}

For $k=1$, the third condition in the above definition is trivially
satisfied. In the literature, downward 1-equivalence is usually
referred to as \emph{bisimilarity} \cite{Sangiorgi}.

\begin{example}
\label{ex-downward}

Consider again the example document in
Figure~\ref{fig-document}. Notice that $v_2\downequiv{k} v_{10}$ for
any value of $k\geq 1$. We also have that $v_2\downequiv{1} v_3$, and,
for any value of $k\geq 2$, $v_2\not\downequiv{k} v_3$. Finally,
notice that $v_3\not\downequiv{k} v_4$ for any value of $k\geq 1$.

\end{example}

The following is immediate from the second condition in the
Definition~\ref{def-kdownwardequivalent}.

\begin{proposition}
\label{prop-kdownwardheight}

Let $D=(V,\textit{Ed},r,\lambda)$ be a document, let
$v_1,v_2\in V$, and let $k\geq 1$. If $v_1\downequiv{k} v_2$, then
$v_1$ and $v_2$ have equal height\footnote{By the height of a node, we
  mean the length of the longest path from that node to a leaf.} in~$D$.

\end{proposition}

The following property of downward-$k$-equivalence will turn out to be
very useful in the sequel.

\begin{proposition}
\label{prop-coarse}

Let $D=(V,\textit{Ed},r,\lambda)$ be a document, and let $k\geq 1$.
Let ``$\equiv$'' be an equivalence relation on~$V$ such that, for all
$v_1,v_2\in V$ with $v_1\equiv v_2$,
\begin{enumerate}

\item $\lambda(v_1) = \lambda(v_2)$;

\item for each child $w_1$ of $v_1$, there exists a child $w_2$ of
    $v_2$ such that $w_1\equiv w_2$, and vice versa; and

\item for each child $w_1$ of $v_1$ and each child $w_2$ of $v_2$ such
  that $w_1\equiv w_2$, $\min(|\tilde{w}_1|,k)=\min(|\tilde{w}_2|,k)$,
where, for $i=1,2$, $\tilde{w}_i$ is the set of all siblings of $v_i$
  (including $v_i$ itself) that are equivalent to $v_i$ under ``$\equiv$.''

\end{enumerate}
Then, for all $v_1,v_2\in V$, $v_1\equiv v_2$ implies $v_1\downequiv{k}
v_2$.

\end{proposition}

\begin{proof}
By induction of the height of~$v_1$.

If $v_1$ is a leaf, the second condition above implies that $v_2$ must
also be a leaf. By the first condition,
$\lambda(v_1)=\lambda(v_2)$. Hence, $v_1\downequiv{k} v_2$.

If $v_1$ is not a leaf, we still have, by the first condition, that
$\lambda(v_1)=\lambda(v_2)$. Hence the first condition in the
definition of $v_1\downequiv{k} v_2$
(Definition~\ref{def-kdownwardequivalent}) is satisfied.

The second condition in the definition of $v_1\downequiv{k} v_2$
follows from the second condition above and the induction hypothesis.

It remains to show that also the third condition in the definition of
$v_1\downequiv{k} v_2$ holds. Thereto, let $w_1$ be a child of $v_1$
and $w_2$ be a child of $v_2$ such that $w_1\downequiv{k} w_2$. We
show that $\min(|\bar{w}_1|,k)=\min(|\bar{w}_2|,k)$, where, for
$i=1,2$, $\bar{w}_i$ is the set of all siblings of $w_i$ (including
$w_i$ itself) that are downward $k$-equivalent to $w_i$.  Let
$\{W_{11},\ldots,W_{1\ell}\}$ be the coarsest partition of $\bar{w}_1$
in $\equiv$-equivalent nodes, and let $\{W_{21},\ldots,W_{2\ell}\}$ be
the coarsest partition of $\bar{w}_2$ in $\equiv$-equivalent nodes. By
the induction hypothesis and the second condition above, both
partitions have indeed the same size. It follows furthermore that no
node of $\bar{w}_1$ is $\equiv$-equivalent with a child of $v_1$
outside $\bar{w}_1$, and that no node of $\bar{w}_2$ is
$\equiv$-equivalent with a child of $v_2$ outside $\bar{w}_2$.
Without loss of generality, we may assume that, for $i=1,\ldots,\ell$,
every node in $W_{1i}$ is $\equiv$-equivalent to every node in
$W_{2i}$. Hence, by the third condition above, $\min(|W_{1i}|,k)=
\min(|W_{2i}|,k)$. We now distinguish two cases.
\begin{enumerate}

\item For all $i=1,\ldots,\ell$, $|W_{1i}|<k$. Then, for all
  $i=1,\ldots,\ell$, $|W_{1i}|=|W_{2i}|$. It follows that
  $|\bar{w}_1|=|\bar{w}_2|$, and, hence, also
that $\min(|\bar{w}_1|,k)=\min(|\bar{w}_2|,k)$.

\item For some $i$, $1\le i\le\ell$, $|W_{1i}|\geq k$. Then,
$|W_{2i}|=|W_{1i}|\geq k$. Hence, $|\bar{w}_1|\geq k$ and
  $|\bar{w}_2|\ge k$.
It follows that $\min(|\bar{w}_1|,k)=\min(|\bar{w}_2|,k)=k$.

\end{enumerate}
We conclude that, in both cases, the third condition in the
definition of $v_1\downequiv{k} v_2$ is also satisfied. 
\end{proof}

So, given a document $D=(V,\textit{Ed},r,\lambda)$,
downward-$k$-equivalence is the coarsest equivalence relation on $V$
satisfying Proposition~\ref{prop-coarse}.

A straightforward application of Proposition~\ref{prop-coarse} yields

\begin{corollary}
\label{cor-kplusoneisfiner}

Let $k\ge 1$. Let $D=(V,\textit{Ed},r,\lambda)$ be a document, and let
$v_1,v_2\in V$. If $v_1\downequiv{k+1}v_2$, then $v_1\downequiv{k}v_2$.

\end{corollary}

\begin{proof}
It suffices to observe that ``$\downequiv{k+1}$'' is an equivalence
relation satisfying Proposition~\ref{prop-coarse} for the value of $k$ in
the statement of the Corollary, above. For the first two conditions in
Proposition~\ref{prop-coarse}, this follows immediately from the
corresponding conditions in Definition~\ref{def-kdownwardequivalent}.
For the third condition in Proposition~\ref{prop-coarse}, this also follows
from the third condition in Definition~\ref{def-kdownwardequivalent}
if one takes into account that, for arbitrary sets $A$ and $B$,
$\min(|A|,k+1)=\min(|B|,k+1)$ implies that $\min(|A|,k)=\min(|B|,k)$.
\end{proof}
%.........................................................................
\subsubsection{Upward distinguishability}
\label{subsubsec-upwardcase}

If we only look upward in the document, there is only one reasonable
definition of node distinguishability, as each node has at most one
parent. In contrast with the downward case, the recursion in the
definition is on the \emph{depth} of the first node.

\begin{definition}
\label{def-upwardequivalent}

Let $D=(V,\textit{Ed},r,\lambda)$ be a document, and let
$v_1,v_2\in V$. Then, $v_1$ and $v_2$ are \emph{upward-equivalent},
denoted $v_1\upequiv v_2$, if 
\begin{enumerate}

\item $\lambda(v_1)=\lambda(v_2)$;

\item $v_1$ is the root if and only if $v_2$ is the root;

\item if $v_1$ and $v_2$ are not the root, and
  $u_1$ and $u_2$ are the parents of $v_1$ and
  $v_2$, respectively, then $u_1\upequiv u_2$.

\end{enumerate}
\end{definition}

It is easily seen that two nodes are upward-equivalent if the paths
from the root to these two nodes are isomorphic in the sense that
they have the same length and corresponding nodes have the same label.

\begin{example}
\label{ex-upward}

In the example document of Figure~\ref{fig-document} we have, e.g., that
$v_6\upequiv v_7$, $v_8\upequiv v_9$, $v_{11}\upequiv v_{12}$, but
$v_8\not\upequiv v_{13}$.

\end{example}
%.........................................................................
\subsubsection{Two-way distinguishability}

If we look both upward and downward in a document, we can define a
notion of equivalence by combining the definitions of upward- and
$k$-downward-equivalence: two nodes are $k$-equivalent if they are
upward-equivalent, and if corresponding nodes on the isomorphic paths
from the root to these nodes are $k$-downward-equivalent. More
formally, we have the following recursive definition, where the
recursion is on the depth of the first node.

\begin{definition}
\label{def-kequivalent}

Let $D=(V,\textit{Ed},r,\lambda)$ be a document, let
$v_1,v_2\in V$, and let $k\geq 1$. Then, $v_1$ and $v_2$ are
\emph{$k$-equivalent}, denoted $v_1\updownequiv{k} v_2$,
if
\begin{enumerate}

\item $v_1 \downequiv{k} v_2$;

\item $v_1$ is the root if and only if $v_2$ is the root; and

\item if $v_1$ and $v_2$ are not the root, and $u_1$ and $u_2$ are the
  parents of $v_1$ and $v_2$, respectively, then $u_1\updownequiv{k} u_2$.

\end{enumerate}
\end{definition}

Stated in a nonrecursive way, two nodes are $k$-equivalent if the
paths from the root to these two nodes have equal length and
corresponding nodes on these two paths are downward-$k$-equivalent.

\begin{example}
\label{exa-bothways}

Consider again the example document in Figure~\ref{fig-document}. We
have that, e.g, $v_5\updownequiv{1} v_6\updownequiv{1} v_7$, but no two
of these nodes are $k$-equivalent for any value of $k\geq 2$. Also,
$v_5\not\updownequiv{k} v_8$ and $v_8\not\updownequiv{k} v_{13}$, for
any value of $k\geq 1$.

\end{example}

By a straightforward inductive argument, the following is immediate
from Corollary~\ref{cor-kplusoneisfiner}.

\begin{proposition}
\label{prop-kplusoneisfiner}

Let $k\ge 1$. Let $D=(V,\textit{Ed},r,\lambda)$ be a document, and let
$v_1,v_2\in V$. If $v_1\updownequiv{k+1}v_2$, then $v_1\updownequiv{k}v_2$.

\end{proposition}
%-------------------------------------------------------------------------
\subsection{Distinguishability of pairs of nodes at the syntactic
  level}
\label{subsec-pairsofnodes}

We also define notions of distinguishability of \emph{pairs} of nodes, by
requiring that the pairs have subsumed or congruent signatures and that
corresponding nodes on the (undirected) paths between begin and end
points of both pairs are related under one of the notions defined
in Subsection~\ref{subsec-syntacticdistinguishability}.

\begin{definition}
\label{def-pairsofnodes}
Let $D=(V,\textit{Ed},r,\lambda)$ be a document,
let $\vartheta$ be
  one of the syntactic relationships between nodes defined
in Subsection \ref{subsec-syntacticdistinguishability}, and let
$v_1,w_1,v_2,$ and $w_2$ be nodes in $V$. Then, 
$(v_1,w_1)$ \emph{$\vartheta$-subsumes\/} $(v_2,w_2)$, denoted
  $(v_1,w_1)\pairs{\siggeq}{\vartheta}(v_2,w_2)$ (respectively, $(v_1,w_1)$ and
  $(v_2,w_2)$ are \emph{$\vartheta$-congruent\/}, denoted 
  $(v_1,w_1)\pairs{\sigequiv}{\vartheta}(v_2,w_2)$) if
\begin{enumerate}

\item $(v_1,w_1)\siggeq (v_2,w_2)$ (respectively,
$(v_1,w_1)\sigequiv (v_2,w_2)$); and

\item for each node $y_1$ on the path form $v_1$ to $w_1$,
  $y_1\vartheta y_2$, where $y_2$ is the unique ancestor of $v_2$ or
  $w_2$ or both for which $(v_2,y_2)\in\sig(v_1,y_1)(D)$ (or, equivalently,
  $(y_2,w_2)\in\sig(y_1,w_1)(D)$).\footnote{In the sequel, we call
    $y_1$ and $y_2$ \emph{corresponding\/} nodes.}

\end{enumerate}
\end{definition}

\begin{example}
\label{ex-pairsofnodes}

Consider again the example document in Figure~\ref{fig-document}. We
have that, e.g., $(v_2,v_5)\pairs{\sigequiv}{\downequiv{k}} (v_3,v_6)$
for $k=1$ but not for any higher value of $k$;
$(v_2,v_5)\pairs{\sigequiv}{\downequiv{k}} (v_{10},v_{13})$ for any
value of $k\geq 1$; 
$(v_2,v_5)\pairs{\sigequiv}{\upequiv} (v_4,v_9)$;
$(v_5,v_6)\pairs{\sigequiv}{\updownequiv{k}} (v_5,v_7)$ for any value of 
$k\geq 1$; and $(v_6,v_7)\pairs{\siggeq}{\updownequiv{1}}(v_2,v_5)$,
but not the other way around. 

\end{example}

The following observation is obvious from the definition.

\begin{proposition}
\label{prop-subpairs}

Let $D=(V,\textit{Ed},r,\lambda)$ be a document, let
$\varphi\in\{\siggeq,\sigequiv\}$, let $\vartheta$ be one
of the syntactic relationships between nodes defined in
Subsection~\ref{subsec-syntacticdistinguishability}, and 
let $v_1$, $w_1$, $v_2$, and $w_2$ be nodes of $D$ such that
$(v_1,w_1)\pairs{\varphi}{\vartheta}(v_2,w_2)$. Let $y_1$ and $y_2$ be
nodes on the path from $v_1$ to $w_1$, and let $z_1$ and $z_2$ be
ancestors of $v_2$ or $w_2$ or both corresponding to
$y_1$ and $y_2$, respectively. Then
$(y_1,z_1)\pairs{\varphi}{\vartheta}(y_2,z_2)$. 

\end{proposition}

The mutual position of the nodes in the statement of
Proposition~\ref{prop-subpairs} is illustrated in Figure~\ref{fig-subpairs}.

\begin{figure}[!thb]
\begin{center}
\resizebox{0.4\textwidth}{!}{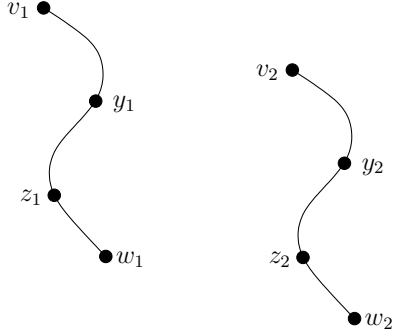}
\end{center}
\caption{Mutual position of the nodes mentioned in the statement
  of Proposition~\ref{prop-subpairs}.}
\label{fig-subpairs}
\end{figure}

From Proposition~\ref{prop-congruence}, (\ref{congruence-1}), the
following is also obvious.

\begin{proposition}
\label{prop-congissub}

Let $D=(V,\textit{Ed},r,\lambda)$ be a document, let $\vartheta$ be one
of the syntactic relationships between nodes defined in
Subsection~\ref{subsec-syntacticdistinguishability}, and let
$v_1,v_2,w_1,$ and $w_2$ be nodes in $V$. If $v_1$ is an ancestor of $w_1$
or vice versa, $(v_1,w_1)\pairs{\sigequiv}{\vartheta}(v_2,w_2)$
if and only if $(v_1,w_1)\pairs{\siggeq}{\vartheta}(v_2,w_2)$.

\end{proposition}

Finally, from Definitions~\ref{def-kequivalent}
and~\ref{def-pairsofnodes}, the following is immediate.

\begin{proposition}
\label{prop-kequivalent}

Let $D=(V,\textit{Ed},r,\lambda)$ be a document, let
$v_1,v_2\in V$, and let $k\geq 1$. Then, $v_1\updownequiv{k} v_2$ 
if and only if $(r,v_1)\pairs{\sigequiv}{\downequiv{k}} (r,v_2)$.

\end{proposition}

Table \ref{table:summary-syntactic} summarizes  all of the distinguishability notions
presented in this section.
The balance of the paper is devoted to identifying the languages which
correspond in expressive power to each of these notions.

\begin{table}[hbt]
    \caption{Distinguishability notions of Section \ref{sec-distinguishability}.}
    \label{table:summary-syntactic}
\begin{center}
    \begin{tabular}{|ccc|}
        \hline
        {\em distinguishability  notion} & {\em notation} & {\em defined in}\\
        \hline
        expression-related & \expgeq         & Definition \ref{def-ee}\\
     expression-equivalent & \expequiv       & Definition \ref{def-ee}\\
   downward-$k$-equivalent & \downequiv{k}   & Definition \ref{def-kdownwardequivalent}\\
         upward-equivalent & \upequiv        & Definition \ref{def-upwardequivalent}\\
            $k$-equivalent & \updownequiv{k} & Definition \ref{def-kequivalent}\\
      $\vartheta$-subsumes & \pairs{\siggeq}{\vartheta} & Definition \ref{def-pairsofnodes}\\
     $\vartheta$-congruent & \pairs{\sigequiv}{\vartheta} &  Definition \ref{def-pairsofnodes}\\
        \hline
    \end{tabular}
\end{center}
\end{table}

%======================================================================
\section{Strictly downward languages}
\label{sec-strictlydownward}

We call a language \emph{downward\/} if, for any expression $e$ in that
language, and for any node $v$ of the document $D$ under
consideration, all nodes in $e(D)(v)$ are descendants of $v$.

In this section, we consider languages with the stronger property
that, for any expression $e$ in the language, and for any node $v$ of
the document $D$ under consideration, $e(D)(v)=e(D')(v)$, where $D'$
is the subtree of~$D$ rooted at~$v$. We shall call such languages
\emph{strictly downward\/}.

Downward languages that are \emph{not} strictly downward will be
called \emph{weakly downward} and are the subject of
Section~\ref{sec-weaklydownward}. 

Considering the nonbasic operations in Table~\ref{tab-binops}, the
language $\BL(E)$ is strictly downward if and only if $E$ does
\emph{not} contain upward navigation (``$\up$''), 
second projection (``$\pi_2$''), and
inverse (``$.^{-1}$''). It is  
the purpose of this section to investigate the expressive power of
these languages at the document level, both for query expressiveness and
navigational expressiveness, and, in some cases, derive actual
characterizations for these.
%-------------------------------------------------------------------------
\subsection{Sufficient conditions for expression equivalence}
\label{subsec-strictlydownsufficient}

If $e$ is an expression in a downward language $\BL(E)$,
then it follows immediately from the definition that, given a node $v$
of the document~$D$ under consideration, each node in $e(D)(v)$ is a
descendant of~$v$. Therefore, we only need to consider
ancestor-descendant pairs of nodes, for which corresponding notions of
subsumption and congruence coincide (Proposition~\ref{prop-congissub}).

The following property of $\downequiv{k}$-congruence, $k\geq 1$, for 
ancestor-descendant pairs of nodes will turn out to be very useful.

\begin{lemma}
\label{lem-strictlydown1}

Let $D=(V,\textit{Ed},r,\lambda)$ be a document, let $v_1$, $w_1$, and
$v_2$ be nodes of $D$ such that $w_1$ is a descendant of $v_1$, and let
$k\geq 1$. If $v_1\downequiv{k} v_2$, then $v_2$ has a descendant
$w_2$ in $D$ such that $(v_1,w_1)\pairs{\sigequiv}{\downequiv{k}} (v_2,w_2)$.

\end{lemma}

\begin{proof}
The proof is by induction of the length of the path from $v_1$ to
$w_1$. If $w_1=v_1$, then, obviously,
Lemma~\ref{lem-strictlydown1} is satisfied for $w_2\ass v_2$.
If $w_1\neq v_1$, then let $y_1$ be the child of $v_1$ on the path to
$w_1$. By Definition~\ref{def-kdownwardequivalent}, $v_2$ has a child
$y_2$ such that $y_1\downequiv{k} y_2$. By the induction hypothesis, 
$y_2$ has a descendant $w_2$ in $D$ such that
$(y_1,w_1)\pairs{\sigequiv}{\downequiv{k}}(y_2,w_2)$. From
Definition~\ref{def-pairsofnodes}, it is now 
straightforward that $(v_1,w_1)\pairs{\sigequiv}{\downequiv{k}}(v_2,w_2)$.
\end{proof}

We now link $\downequiv{k}$-congruence of ancestor-descendant pairs of
nodes with expressibility in strictly downward languages.

\begin{proposition}
\label{prop-strictlydown}

Let $k\geq 1$, and let $E$ be the set of all nonbasic operations in
Table~\ref{tab-binops}, except for upward navigation (``$\up$''),
second projection
(``$\pi_2$''), inverse (``$.^{-1}$''), and selection on at least $m$
children satisfying some condition (``$\ch{m}(.)$'') for $m>k$. Let
$e$ be an expression in $\BL(E)$.  Let $D=(V,\textit{Ed},r,\lambda)$
be a document, let $v_1$, $w_1$, $v_2$, and $w_2$ be nodes of $D$ such
that $w_1$ is a descendant of $v_1$ and $w_2$ is a descendant of
$v_2$.  Assume furthermore that
$(v_1,w_1)\pairs{\sigequiv}{\downequiv{k}} (v_2,w_2)$. Then,
$(v_1,w_1)\in e(D)$ if and only if $(v_2,w_2)\in e(D)$.

\end{proposition}

\begin{proof}
By symmetry, it suffices to show
that $(v_1,w_1)\in e(D)$ implies $(v_2,w_2)\in e(D)$.
We prove this by structural induction.
For the atomic operators $\emptyset$, $\varepsilon$, $\hat{\ell}$
($\ell\in\mathcal{L}$), and $\downarrow$, it is
straightforward that Proposition~\ref{prop-strictlydown} holds.
We have now settled the base case and turn to the induction step.
\begin{enumerate}

\item $e \ass e_1/e_2$, with $e_1$ and $e_2$ satisfying 
Proposition~\ref{prop-strictlydown}. Assume that $(v_1,w_1)\in
e(D)$. Then there exists $y_1\in V$ such that 
$(v_1,y_1)\in e_1(D)$ and $(y_1,w_1)\in e_2(D)$. By the strictly
downward nature of $\BL(E)$, $y_1$ is on the path from $v_1$ to $w_1$.
Let $y_2$ be the node on the path from $v_2$ to $w_2$ corresponding to
$y_1$. By Proposition~\ref{prop-subpairs}, $(v_1,y_1)\downequiv{k}
(v_2,y_2)$ and $(y_1,w_1)\downequiv{k} (y_2,w_2)$. By the induction
hypothesis, $(v_2,y_2)\in e_1(D)$ and $(y_2,w_2)\in e_2(D)$. Hence,
$(v_2,w_2)\in e(D)$.

\item $e \ass \pi_1(f)$, with $f$ satisfying
  Proposition~\ref{prop-strictlydown}. Assume that $(v_1,w_1)\in
  e(D)$. Then, necessarily $v_1=w_1$, and, consequently, $v_2=w_2$.
  From $(v_1,v_1)\in\pi_1(f)(D)$, it follows that there exists $z_1\in
  V$ such that $(v_1,z_1)\in f(D)$.  Since $v_1\downequiv{k} v_2$, it
  also follows, by Lemma~\ref{lem-strictlydown1}, that there exists a
  descendant $z_2$ of $w_2$ such that
  $(v_1,z_1)\pairs{\sigequiv}{\downequiv{k}} (v_2,z_2)$. By the
  induction hypothesis, $(v_2,z_2)\in f(D)$. Hence, $(v_2,v_2)\in
  e(D)$.

\item $e \ass \ch{m}(f)$, with $m\leq k$ and $f$ satisfying
  Proposition~\ref{prop-strictlydown}. Assume that $(v_1,w_1)\in
  \ch{m}(f)(D)$. Hence, $v_1=w_1$, which in turn implies $v_2=w_2$.
  Let $\down/\pi_1(f)(D)(v_1)=Y_1$ and let
  $\down/\pi_1(f)(D)(v_2)=Y_2$. By assumption, $|Y_1|\geq m$.
  Now, let $y$ be a child of $v_1$ in $Y_1$ or a child of $v_2$ in
  $Y_2$, and let $z$ be a child of $v_1$ not in $Y_1$ or a child of
  $v_2$ not in $Y_2$. By assumption, there exists a node $y'$ such
  that $(y,y')\in f(D)$. Now, suppose that $y\downequiv{k} z$. Then,
  by Proposition~\ref{lem-strictlydown1}, there exists a node $z'$
  such that $(y,y')\pairs{\sigequiv}{\downequiv{k}} (z,z')$. But then,
  by the induction hypothesis, $(z,z')\in f(D)$, contrary to our
  assumptions. We may therefore conclude that $y\not\downequiv{k}
  z$. Since furthermore $v_1\downequiv{k} v_2$, it follows that, for
  all $y_1\in Y_1$, there exists $y_2\in Y_2$ such that
  $y_1\downequiv{k} y_2$, and vice versa.  Hence, for some $n\geq 1$,
  we can write $Y_1=Y_{11}\cup\ldots\cup Y_{1n}$ and
  $Y_2=Y_{21}\cup\ldots\cup Y_{2n}$ such that
\begin{enumerate}

\item $Y_{11},\ldots,Y_{1n}$ are maximal sets of mutually
  downward-$k$-equivalent children of $v_1$, and are hence pairwise
  disjoint;

\item $Y_{21},\ldots,Y_{2n}$ are maximal sets of mutually
  downward-$k$-equivalent children of $v_2$, and are hence pairwise
  disjoint; and

\item for all $i=1,\ldots,n$, each node of $Y_{1i}$ is
  downward-$k$-equivalent to each node of $Y_{2i}$.

\end{enumerate}
If, for some $i$, $|Y_{1i}|\geq k$, it follows from $v_1\downequiv{k}
v_2$ that $|Y_{2i}|\geq k$, and, hence, that $|Y_2|\geq k\geq m$. If, on the
other hand,  for all $i=1,\ldots,n$, $|Y_{1i}|<k$, it
follows from $v_1\downequiv{k} v_2$ that $|Y_{1i}|=|Y_{2i}|$, and,
hence, that $|Y_1|=|Y_2|$. Since $|Y_1|\geq m$, it follows that, also
in this case, $|Y_2|\geq m$. We may thus conclude that, in all cases,
$|Y_2|\geq m$, and, hence, that $(v_2,v_2)\in \ch{m}(f)(D)=e(D)$.

\item $e \ass e_1\cup e_2$, with $e_1$ and $e_2$ satisfying
  Proposition~\ref{prop-strictlydown}. Assume that $(v_1,w_1)\in e(D)$.
Then, $(v_1,w_1)\in e_1(D)$ or $(v_1,w_1)\in e_2(D)$. Without loss of
generality, assume the former. Then, by the induction hypothesis,
$(v_2,w_2)\in e_1(D)$. Hence, $(v_2,w_2)\in e(D)$.
 
\item $e \ass e_1\cap e_2$, with $e_1$ and $e_2$ satisfying
  Proposition~\ref{prop-strictlydown}. Assume that $(v_1,w_1)\in
  e(D)$.  Then, $(v_1,w_1)\in e_1(D)$ and $(v_1,w_1)\in e_2(D)$.  It
  follows by the induction hypothesis that $(v_2,w_2)\in e_1(D)$ and
  $(v_2,w_2)\in e_2(D)$. Hence, $(v_2,w_2)\in e(D)$.

\item $e \ass e_1-e_2$, with $e_1$ and $e_2$ satisfying
  Proposition~\ref{prop-strictlydown}. Assume that $(v_1,w_1)\in e(D)$.
Then $(v_1,w_1)\in e_1(D)$ and $(v_1,w_1)\notin e_2(D)$. By the
induction hypothesis, $(v_2,w_2)\in 
  e_1(D)$ and $(v_2,w_2)\notin e_2(D)$. (Indeed, if $(v_2,w_2)\in
  e_2(D)$, then, again by the induction hypothesis, $(v_1,w_1)\in
  e_2(D)$, a contradiction.) Hence, $(v_2,w_2)\in e(D)$.
\end{enumerate}
\end{proof}

\begin{corollary}
\label{cor-strictlydown1}

Let $k\geq 1$, and let $E$ be the set of all nonbasic operations in
Table~\ref{tab-binops}, except for upward navigation (``$\up$''),
second projection
(``$\pi_2$''), inverse (``$.^{-1}$''), and selection on at least $m$
children (``$\ch{m}(.)$'') for $m>k$. Let $e$ be an expression in
$\BL(E)$. Let $D=(V,\textit{Ed},r,\lambda)$ be a document, let $v_1$
and $v_2$ be nodes of $D$ such that $v_1\downequiv{k} v_2$ and let
$w_1$ be a descendant of $v_1$. If $(v_1,w_1)\in e(D)$, then there
exists a descendant $w_2$ of $v_2$ such that $(v_2,w_2)\in e(D)$.

\end{corollary}

\begin{proof}
By Lemma~\ref{lem-strictlydown1}, there exists a descendant $w_2$ of
$v_2$ such that $(v_1,w_1)\pairs{\sigequiv}{\downequiv{k}}(v_2,w_2)$. By
Proposition~\ref{prop-strictlydown}, it now follows that
$(v_2,w_2)\in e(D)$. 
\end{proof}

\begin{corollary}
\label{cor-strictlydown2}

Let $k\geq 1$, and let $E$ be a set of nonbasic operations in
Table~\ref{tab-binops} not  
containing upward navigation (``$\up$''), second projection
(``$\pi_2$''), inverse  
(``$.^{-1}$''), or selection on at least $m$ children satisfying some
condition (``$\ch{m}(.)$'') for $m>k$.  Consider the language $\BL(E)$
or $\CL(E)$. Let $D=(V,\textit{Ed},r,\lambda)$ be a document, and let
$v_1$ and $v_2$ be nodes of $D$. If $v_1\downequiv{k} v_2$, then
$v_1\expequiv v_2$.

\end{corollary}

\begin{proof}
Let $e$ be an expression in the language under consideration
such that $e(D)(v_1)\neq\emptyset$. Hence, there exists a descendant $w_1$ of
$v_1$ such that $(v_1,w_1)\in e(D)$. Notice that $e$ is also an
expression in the language considered in
Corollary~\ref{cor-strictlydown1}. Hence, there exists a descendant $w_2$ of
$v_2$ such that $(v_2,w_2)\in e(D)$, so $e(D)(v_2)\neq\emptyset$. By
symmetry, the converse also holds. We may thus conclude that
$v_1\expequiv v_2$. 
\end{proof}

We may thus conclude that downward-$k$-equivalence is a sufficient
condition for expression-equivalence under a strictly downward
language provided $\ch{m}$ cannot be expressed for $m>k$.

Even more, Corollary~\ref{cor-strictlydown2} does no longer hold if
this restriction is removed, as shown by the following counterexample.

\begin{example}
\label{ex-downsuff-counterex}

Consider again the example document in
Figure~\ref{fig-document}. We established in Example~\ref{ex-downward}
that $v_2\downequiv{1} v_3$, but $v_2\not\downequiv{2} v_3$. In the
language $\BL(\ch{2})$, clearly $v_2\not\expequiv v_3$, as
$\ch{2}(\varepsilon)(D)(v_2)=\emptyset$, while
$\ch{2}(\varepsilon)(D)(v_3)\neq\emptyset$.

\end{example}
%-------------------------------------------------------------------------
\subsection{Necessary conditions for expression equivalence}
\label{subsec-strictlydownnecessary}

We now explore requirements on the set of nonbasic operations
expressible in the language under which downward-$k$-equivalence
($k\geq 1$) is a necessary condition for expression-equivalence.  As
we have endeavored to make as few assumptions as possible,
Proposition~\ref{prop-downnecessary} 
also holds for a class of languages
that are \emph{not} (strictly) downward.

\begin{proposition}
\label{prop-downnecessary}

Let $k\geq 1$, and let $E$ be a set of nonbasic operations containing
set difference (``$-$'').  Consider the language $\BL(E)$ or
$\CL(E)$. Assume that, in this language, 
first projection (``$\pi_1$'') can be expressed, as well as
selection on at least $m$
children satisfying some condition (``$\ch{m}(.)$''),
for all $m=1,\ldots,k$.  Let $D=(V,\textit{Ed},r,\lambda)$ be a
document, and let $v_1$ and $v_2$ be nodes of $D$. If $v_1\expequiv
v_2$, then $v_1\downequiv{k} v_2$.

\end{proposition}

\begin{proof}
Since expression-equivalence in the context of $\BL(E)$ implies
expression-equivalence in the context of $\CL(E)$, we may assume
without loss of generality that the language under consideration is $\CL(E)$.
To prove Proposition~\ref{prop-downnecessary}, it suffices to show
that expression-equivalence (``$\expequiv$'') satisfies the conditions
of Proposition~\ref{prop-coarse}.
\begin{enumerate}

\item If $v_1\expequiv v_2$, then $\lambda(v_1)=\lambda(v_2)$, for,
  otherwise, $\widehat{\lambda(v_1)}(D)(v_1)\neq\emptyset$, while
$\widehat{\lambda(v_1)}(D)(v_2)=\emptyset$, a contradiction.

\item If $v_1\expequiv v_2$ and $v_1$ is not a leaf, then $v_2$ is not a leaf
either, for, otherwise, $\ch{1}(\varepsilon)(D)(v_1)\neq\emptyset$, while
$ch_1(\varepsilon)(D)(v_2)=\emptyset$, a contradiction. Let $w_1$ be a child of
$v_1$, and let $w_2^1,\ldots,w_2^n$ be all children of $w_2$.
Suppose for the sake of contradiction that, for all $i=1,\ldots,n$,
$w_1\not\expequiv w_2^i$. Then, by Proposition~\ref{prop-notempty-empty},
there exists an expression $e_i$ in $\CL(E)$ such that
$e_i(D)(w_1)\not=\emptyset$ and $e_i(D)(w_2^i)=\emptyset$, for all
$i=1,\ldots,n$. Now, let $e\ass
\pi_1(e_1)\cap\ldots\cap\pi_1(e_n)$, which can be
expressed in $\CL(E).$\footnote{Let $f_1$ and $f_2$ be expressions in
  $\CL(E)$ such that $f_1(D)\subseteq\varepsilon(D)$ and
  $f_2(D)\subseteq\varepsilon(D)$. Then, $f_1\cap f_2$ can be
  expressed in $\CL(E)$ as $\pi_1(\varepsilon-
  \pi_1(\varepsilon-f_1)\cup\pi_1(\varepsilon-f_2))$.}
 Then, $\ch{1}(e)(D)(v_1)\neq\emptyset$ while
$\ch{1}(e)(D)(v_2)=\emptyset$, contradicting $v_1\expequiv v_2$. Hence,
there does exist a child $w_2$ of $v_2$ such that $w_1\expequiv w_2$. 
Of course, the same also goes with the roles of $v_1$ and $v_2$ reversed.

\item Finally, let $v_1$ and $v_2$ be non-leaf nodes such that
  $v_1\expequiv v_2$, and let $w_1$ and $w_2$ be children of $v_1$ and
  $v_2$, respectively, such that $w_1\expequiv w_2$. For $i=1,2$, let
  $\tilde{w}_i$ be the set of all siblings of $w_i$ (including $w_i$
  itself) that are expression-equivalent to $w_i$. As in the previous
  item, we can construct an expression $e$ in $\CL(E)$ such that
  $e(D)(w_1)\neq\emptyset$ (and hence $e(D)(w)\neq\emptyset$ for each
  node $w$ in $\tilde{w}_1$ or $\tilde{w_2}$)
  and $e(D)(w)=\emptyset$ for each sibling of $w_1$ not in $\tilde{w}_1$
  \emph{and} for each sibling of $w_2$ not in $\tilde{w}_2$. For the
  sake of contradiction, assume that $\min(|\tilde{w}_1|,k)\neq
  \min(|\tilde{w}_2|,k)$. Without loss of generality, assume that
  $\min(|\tilde{w}_1|,k)<\min(|\tilde{w}_2|,k)$. Hence, 
  $\min(|\tilde{w}_1|,k)=|\tilde{w}_1|$. Let $m\ass
  \min(|\tilde{w}_2|,k)$. Then, $\ch{m}(e)(D)(v_1)=\emptyset$, while
  $\ch{m}(e)(D)(v_2)\neq\emptyset$, contradicting $v_1\expequiv v_2$.
  We may thus conclude that $\min(|\tilde{w}_1|,k)=\min(|\tilde{w}_2|,k)$. 

\end{enumerate}
\end{proof}

Notice that the languages satisfying the statement of
Proposition~\ref{prop-downnecessary} need not contain any navigation
operations (``$\down$'' or ``$\up$''). Of course, in the context of
this Section, we are interested in languages in which downward navigation
(``$\down$'') is possible. Specializing 
Proposition~\ref{prop-downnecessary} to this case, 
we may thus conclude that downward-$k$-equivalence is a
necessary condition for expression-equivalence under a strictly
downward language containing first
projection (``$\pi_1$'') and set difference (``$-$''), provided 
selection on at least $m$ children satisfying some condition
(``$\ch{m}$'') for all $m=1,\ldots,k$ can be expressed.
%-------------------------------------------------------------------------
\subsection{Characterization of expression equivalence}
\label{subsec-strictlydowncharacterization}

The languages containing downward navigation (``$\down$'') and
satisfying both Corollary~\ref{cor-strictlydown2} 
of Subsection~\ref{subsec-strictlydownsufficient}
and Proposition~\ref{prop-downnecessary} 
of Subsection~\ref{subsec-strictlydownnecessary} 
are \\ 
$\BL(\down,\pi_1,\ch{1}(.),\ldots,\ch{k}(.),-)$ and
$\CL(\down,\pi_1,\ch{1}(.),\ldots,\ch{k}(.),-)$.
We call these languages
the \emph{strictly downward XPath algebra with counting up to~$k$\/} and the
\emph{strictly downward core XPath algebra with counting up to~$k$\/},
respectively. Combining the aforementioned results, we get the following.

\begin{theorem}
\label{theo-downequivalent}

Let $k\geq 1$, and consider the strictly downward (core) XPath algebra with
counting up to~$k$.  Let $D=(V,\textit{Ed},r,\lambda)$ be a document,
and let $v_1$ and $v_2$ be nodes of $D$. Then $v_1\expequiv v_2$, if
and only if $v_1\downequiv{k} v_2$.

\end{theorem}

A special case arises when $k=1$, since selection on at least one
child satisfying some condition (``$\ch{1}(.)$'') can be expressed in
terms of the other operations required by
Theorem~\ref{theo-downequivalent}, by Proposition~\ref{prop-counting}.
The languages we then obtain, $\BL(\down,\pi_1,-)$ and
$\CL(\down,\pi_1,-)$, are called the \emph{strictly downward XPath 
  algebra} and the \emph{strictly downward core XPath algebra}, respectively.
We have the following.

\begin{corollary}
\label{cor-downequivalent}

Consider the strictly downward (core) XPath algebra.  Let
$D=(V,\textit{Ed},r,\lambda)$ be a document, and let $v_1$ and $v_2$
be nodes of $D$. Then $v_1\expequiv v_2$, if and only if
$v_1\downequiv{1} v_2$.

\end{corollary}
%-------------------------------------------------------------------------
\subsection{Characterization of navigational expressiveness}
\label{subsec-strictlybp}

We shall now investigate the expressiveness of strictly downward
languages at the document level. In other words, we shall address the
question whether, given a document, we can characterize when a set of
pairs of nodes of that document is the result of some query in the
language under consideration applied to that document. Such type of
results are often referred to as BP-characterizations, after
Bancilhon~\cite{Bancilhon78} and Paredaens~\cite{Paredaens78} who
first proved such results for Codd's relational calculus and algebra,
respectively (cf.\ \cite{ChandraHarel}). 

We start by proving a converse to Proposition~\ref{prop-strictlydown}.

\begin{proposition}
\label{prop-strictlydownconverse}

Let $k\geq 1$, and let $E$ be a set of nonbasic operations containing
downward navigation (``$\down$'') and
set difference (``$-$'').  Consider the language $\BL(E)$ or
$\CL(E)$. Assume that, in this language, first projection
(``$\pi_1$'') can be expressed, as well as
selection on at least $m$
children satisfying some condition (``$\ch{m}(.)$''),
for all $m=1,\ldots,k$.  Let $D=(V,\textit{Ed},r,\lambda)$ be a
document, and let
$v_1$, $w_1$, $v_2$, and $w_2$ be nodes of $D$ such that $w_1$ is a
descendant of~$v_1$ and $w_2$ is a descendant of~$v_2$. Assume
furthermore that, for each expression~$e$ in the language, $(v_1,w_1)\in
e(D)$ if and only if $(v_2,w_2)\in e(D)$. Then
$(v_1,w_1)\pairs{\sigequiv}{\downequiv{k}}(v_2,w_2)$.

\end{proposition}

\begin{proof}
First notice that, by assumption, $(v_2,w_2)\in\sig(v_1,w_1)(D)$, and
vice versa. Hence, $(v_1,w_1)\sigequiv (v_2,w_2)$. Let $y_1$ be a node
on the path from $v_1$ to~$w_1$, and let $y_2$ be the corresponding
node on the path from $v_2$ to $w_2$. By construction,
$(v_1,y_1)\sigequiv (v_1,y_2)$ and $(y_1,w_1)\sigequiv
(y_2,w_2)$. Now, let $f$ be any expression in the language such that
$f(D)(y_1)\neq\emptyset$. Then, $(y_1,y_1)\in \pi_1(f)(D)$. Let
$e\ass \sig(v_1,y_1)/\pi_1(f)/\sig(y_1,w_1)$. By construction,
$(v_1,w_1)\in e(D)$. Hence, by assumption, $(v_2,w_2)\in e(D)$, which
implies $(y_2,y_2)\in \pi_2(f)(D)$ or $f(D)(y_2)\neq\emptyset$. The same
holds vice versa, and we may thus conclude that $y_1\expequiv y_2$, and,
hence, by Proposition~\ref{prop-downnecessary}, $y_1\downequiv{k} y_2$.
We may thus conclude that $(v_1,w_1)\downequiv{k}(v_2,w_2)$.
\end{proof}

Combining Propositions~\ref{prop-strictlydown}
and~\ref{prop-strictlydownconverse}, we obtain the following.

\begin{corollary}
\label{cor-strictlydown}

Let $k\geq 1$, and consider the strictly downward (core) XPath algebra with
counting up to~$k$. Let
$D=(V,\textit{Ed},r,\lambda)$ be a document, and let $v_1$, $w_1$,
$v_2$, and $w_2$ be nodes of $D$ such that $w_1$ is a descendant of
$v_1$ and $w_2$ is a descendant of $v_2$. 
Then, the property that, for each expression $e$ in the
language under consideration, $(v_1,w_1)\in e(D)$ if and only if
$(v_2,w_2)\in e(D)$ is equivalent to the property
$(v_1,w_1)\pairs{\sigequiv}{\downequiv{k}} (v_2,w_2)$.

\end{corollary}

In order to state our first BP-result, we need the following two lemmas.

\begin{lemma}
\label{lem-separationdown1}

Let $k\geq 1$. Let $D=(V,\textit{Ed},r,\lambda)$ be a document, and
let $v_1$ be a node of $D$.  There exists an expression $e_{v_1}$ in
the strictly downward core XPath algebra with counting up to~$k$ such that, for
each node $v_2$ of $D$, $e_{v_1}(D)(v_2)\neq\emptyset$ if and only if
$v_1\downequiv{k} v_2$. 

\end{lemma}

\begin{proof}
Let $w$ be any node of $D$ such that $v_1\not\downequiv{k} w$. By
Theorem~\ref{theo-downequivalent}, $v_1\not\expequiv w$. By
Proposition~\ref{prop-notempty-empty}, there exists an expression
$f_{v_1,w}$ in the strictly downward core XPath algebra with counting
up to~$k$ such that
$f_{v_1,w}(D)(v_1)\neq\emptyset$ and 
$f_{v_1,w}(D)(w)=\emptyset$. Now consider the expression
$$e_{v_1}\ass \pi_1\left(\bigcap_{w\in V\ \&\ v_1\not\downequiv{k} w}
\pi_1(f_{v_1,w})\right),$$ 
which is also in the strictly downward core XPath algebra with
counting up to~$k$. By construction,
$e_{v_1}(D)(v_1)\neq\emptyset$. Now consider a node $v_2$ of $D$. If
$v_1\downequiv{k} v_2$, then, by Theorem~\ref{theo-downequivalent},
$v_1\expequiv v_2$. Hence, by definition,
$e_{v_1}(D)(v_2)\neq\emptyset$. If, on the other hand,
$v_1\not\downequiv{k} v_2$, then, by construction,
$e_{v_1}(D)(v_2)=\emptyset$.
\end{proof}

\begin{lemma}
\label{lem-separationdown2}

Let $k\geq 1$.  Let $D=(V,\textit{Ed},r,\lambda)$ be a document, and
let $v_1$ and $w_1$ be nodes of $D$ such that $w_1$ is a descendant
of $v_1$. There exists an expression $e_{v_1,w_1}$ in the strictly downward
core XPath algebra with counting up to~$k$ such that, for all nodes
$v_2$ and $w_2$ of $D$ with $w_2$ a descendant of $v_2$, $(v_2,w_2)\in
e_{v_1,w_1}(D)$ if and only if
$(v_1,w_1)\pairs{\sigequiv}{\downequiv{k}} (v_2,w_2)$.

\end{lemma}

\begin{proof}
From Lemma~\ref{lem-separationdown1}, we know that, for node $y_1$ of $D$,
there exists an expression $e_{y_1}$ in the strictly downward core
XPath algebra with counting up to~$k$
such that, for each node $y_2$ of $D$,
$e_{y_1}(D)(y_2)\neq\emptyset$ if and only if $y_1\downequiv{k}
y_2$. Now, let $v_1$ and $w_1$ be nodes of $D$ such that $w_1$ is a
descendant of $v_1$, and let $v_1=y_{11},\ldots,y_{1n}=w_1$ be the
path from $v_1$ to $w_1$ in $D$. Define
$$e_{v_1,w_1}\ass
\pi_1(e_{y_{11}})/\down/\pi_1(e_{y_{12}})/\ldots\down/\pi_1(e_{y_{1n}}),$$
which is also in the strictly downward core XPath algebra with
counting up to~$k$. By
construction, $(v_1,w_1)\in e_{v_1,w_1}(D)$. 
Let $v_2$ and $w_2$ be nodes of $D$ such
that $w_2$ is a descendant of $v_2$. If
$(v_1,w_1)\pairs{\sigequiv}{\downequiv{k}} (v_2,w_2)$, then, by
Corollary~\ref{cor-strictlydown}, $(v_2,w_2)\in
e_{v_1,w_1}(D)$. Conversely, if $(v_2,w_2)\in e_{v_1,w_1}(D)$,
then, by construction, $(v_1,w_1)\sigequiv (v_2,w_2)$.  Thus, let
$v_2=y_{21},\ldots,y_{2n}=w_2$ be the path from $v_2$ to $w_2$ in $D$.
Again by construction, it follows that, for $j=1,\ldots,n$,
$e_{y_{1j}}(D)(y_{2j})\neq\emptyset$, or, equivalently, that
$y_{1j}\downequiv{k} y_{2j}$. Hence, $(v_1,w_1)\downequiv{k}
(v_2,w_2)$.
\end{proof}

We are now ready to state the actual result.

\begin{theorem}
\label{theo-strictlydown-bp}

Let $k\geq 1$.
Let $D=(V,\textit{Ed},r,\lambda)$ be a document, and
let $R\subseteq V\times V$. Then, there exists an expression $e$ in the
strictly downward (core) XPath algebra with counting up to~$k$ such that
$e(D)=R$ if and only if, 
\begin{enumerate}

\item for all $v,w\in V$, $(v,w)\in R$ implies $w$ is a descendant
  of~$v$; and, 

\item for all $v_1,w_1,v_2,w_2\in V$ with $w_1$ a descendant of $v_1$,
  $w_2$ a descendant of $v_2$, and
  $(v_1,w_1)\pairs{\sigequiv}{\downequiv{k}} (v_2,w_2)$, 
$(v_1,w_1)\in R$ implies $(v_2,w_2)\in R$. 

\end{enumerate}
\end{theorem}

\begin{proof}
To see the ``only if,'' it suffices to notice that the first condition
follows from the downward character of the language, and the second
from Corollary~\ref{cor-strictlydown}.
The remainder of the proof concerns the ``if.'' From
Lemma~\ref{lem-separationdown2}, we know that, for all nodes $v_1$ and
$w_1$ of $D$ such that $w_1$ is a descendant of $v_1$, there exists an
expression $e_{v_1,w_1}$ in $\CL(E)$ such that, for all nodes $v_2$
and $w_2$ of $D$, $(v_2,w_2)\in e_{v_1,w_1}(D)$ if and only if
$(v_1,w_1)\pairs{\sigequiv}{\downequiv{k}} (v_2,w_2)$. Now consider the
expression 
$$e\ass \bigcup_{(v_1,w_1)\in R}e_{v_1,w_1}.$$
This expression, which is well defined because $(v_1,w_1)\in R$
by assumption implies that $w_1$ is a descendant of $v_1$, is also in
$\CL(E)$ (and hence also in $\BL(E)$). It remains to show that $e(D)=R$.
Clearly, $R\subseteq e(D)$. We prove the reverse inclusion. Thereto,
let $v_2$ and $w_2$ be nodes such that $(v_2,w_2)\in e(D)$. By
construction, there exist nodes $v_1$ and
$w_1$ in $D$ such that $w_1$ is a descendant of $v_1$ and
$(v_2,w_2)\in e_{v_1,w_1}(D)$. Hence,
$(v_1,w_1)\pairs{\sigequiv}{\downequiv{k}} (v_2,w_2)$. But then, by
assumption, also $(v_2,w_2)\in R$. So, $e(D)\subseteq R$. 
\end{proof}

As before, we can specialize Theorem~\ref{theo-strictlydown-bp} to the
strictly downward (core) XPath algebra.

\begin{corollary}
\label{cor-downXPath-bp}

Let $D=(V,\textit{Ed},r,\lambda)$ be a document, and let $R\subseteq
V\times V$. There exists an expression $e$ in the strictly downward (core)
XPath algebra such that $e(D)=R$ if and only if,
\begin{enumerate}

\item for all $v,w\in V$, $(v,w)\in R$ implies $w$ is a descendant
  of~$v$;

\item for all $v_1,w_1,v_2,w_2\in V$ with $w_1$ a descendant of $v_1$,
  $w_2$ a descendant of $v_2$, and $(v_1,w_1)\downequiv{1} (v_2,w_2)$,
  $(v_1,w_1)\in R$ implies $(v_2,w_2)\in R$. 

\end{enumerate}
\end{corollary}

We can also recast Theorem~\ref{theo-strictlydown-bp} in terms of
node-level navigation.

\begin{theorem}
\label{theo-strictlydown-nodelevel}

Let $k\geq 1$.  Let $D=(V,\textit{Ed},r,\lambda)$ be a document, let
$v$ be a node of $D$, and let $W\subseteq V$. Then there exists an
expression $e$ in the strictly downward (core) XPath algebra with counting up
to~$k$ such that $e(D)(v)=W$ if and only if all nodes of~$W$ are
descendants of $v$, and, for all $w_1,w_2\in W$ with
$(v,w_1)\pairs{\sigequiv}{\downequiv{k}} (v,w_2)$, $w_1\in W$ implies
$w_2\in W$.

\end{theorem}

\begin{proof}
\emph{Only if}. Let $e$ be an expression in the language under
consideration such that $e(D)(v)=W$. Let $w_1,w_2\in V$ be descendants
of~$v$ with $(v,w_1)\downequiv{k} (v,w_2)$, and assume 
that $w_1\in W=e(D)(v)$. Hence, $(v,w_1)\in e(D)$.
By Corollary~\ref{cor-strictlydown}, $(v,w_2)\in e(D)$. Hence, $w_2\in
e(D)(v)=W$. 

\emph{If}. Let $W\subseteq V$ satisfy the property that all nodes of
$W$ are descendants of $v$, and, for all $w_1,w_2\in V$ with
$w_1\downequiv{k} w_2$, $w_1\in W$ implies $w_2\in W$. 
Let $R := \{(v',w_2)\mid\ \textrm{there exists}\ w_1\in
W\ \textrm{such that}\ (v,w_1)\pairs{\sigequiv}{\downequiv{k}} (v',w_2)\}$. 
Clearly, $R$ satisfies the properties 
of Theorem~\ref{theo-strictlydown-bp}. Hence, there exists an expression $e$
in the language under consideration such that $R=e(D)$. Clearly, $W\subseteq
e(D)(v)$. We prove the reverse inclusion. Therefore, let $w_2\in
e(D)(v)$, i.e., $(v,w_2)\in R$. Then there exists $w_1\in R$ such that
$(v,w_1)\pairs{\sigequiv}{\downequiv{k}} (v,w_2)$. By the property that
$W$ satisfies, $w_2\in W$. Hence, $e(D)(v)\subseteq W$, and,
therefore, $e(D)(v)=W$. 
\end{proof}

Again, we can specialize Theorem~\ref{theo-strictlydown-nodelevel} to the
strictly downward (core) XPath algebra.

\begin{corollary}
\label{cor-strictlydown-nodelevel}

Let $D=(V,\textit{Ed},r,\lambda)$ be a document, let $v$ be a node of
$D$, and let $W\subseteq V$. Then there exists an 
expression $e$ in the strictly downward (core) XPath algebra such that
$e(D)(v)=W$ if and only if all nodes of~$W$ are descendants of $v$,
and, for all nodes $w_1$ and $w_2$ of $D$ with
$(v,w_1)\pairs{\sigequiv}{\downequiv{1}} (v,w_2)$, 
$w_1\in W$ implies $w_2\in W$. 

\end{corollary}

A special case of Theorem~\ref{cor-strictlydown-nodelevel} is when we
are only interested in navigation from the \emph{root}.

\begin{theorem}
\label{theo-strictlydown-root}

Let $k\geq 1$.  Let $D=(V,\textit{Ed},r,\lambda)$ be a document, and
let $W\subseteq V$. Then there exists an expression $e$ in the
strictly downward (core) XPath algebra with counting up to~$k$ such that
$e(D)(r)=W$ if and only if, for all nodes $w_1$ and $w_2$ of $D$ with
$w_1\updownequiv{k} w_2$, $w_1\in W$ implies $w_2\in W$.

\end{theorem}

\begin{proof}
From Theorem~\ref{theo-strictlydown-nodelevel}, it immediately follows that
there exists an expression $e$ in the language under consideration such
that $e(D)(r)=W$ if and only if, for $w_1,w_2\in V$ with
$(r,w_1)\pairs{\sigequiv}{\downequiv{k}}(r,w_2)$, $w_1\in W$ implies
$w_2\in W$. By Proposition~\ref{prop-kequivalent},
$(r,w_1)\pairs{\sigequiv}{\downequiv{k}} (r,w_2)$ is 
equivalent to $w_1\updownequiv{k} w_2$. 
\end{proof}

The specialization of Theorem~\ref{theo-strictlydown-root} to the
case of the strictly downward (core) XPath algebra is as follows.

\begin{corollary}
\label{cor-strictlydown-root}

Let $D=(V,\textit{Ed},r,\lambda)$ be a document, 
and let $W\subseteq V$. Then there exists an 
expression $e$ in the strictly downward (core) XPath algebra such that
$e(D)(r)=W$ if and only if, for all nodes $w_1$ and $w_2$ of $D$ with
$w_1\updownequiv{1} w_2$, $w_1\in W$ implies $w_2\in W$. 

\end{corollary}

To conclude this section, we observe that none of the characterization
results above distinguish between the language $\BL(E)$ and the
corresponding core language $\CL(E)$.  This is not surprising, as, for
\emph{all\/} downward languages, they have the same expressive power,
not only at the navigational level for a given document, but also at
the level of queries, i.e., for each expression $e$ in $\BL(E)$, there
exists an equivalent expression $e'$ in the corresponding core
language $\CL(E)$, meaning that, \emph{for each\/} document $D$,
$e(D)=e'(D)$. Thereto, we prove a slightly stronger result.

\begin{theorem}
\label{theo-downwardcorecoincides}

Let $E$ be a set of nonbasic operations containing
downward navigation (``$\down$'') and
first projection (``$\pi_1$''),
and not containing upward navigation
(``$\up$''), and inverse
(``$.^{-1}$''). Let $e$ be an
expression in the language under consideration. With the exception of
intersection (``$\cap$'') and set difference (``$-$'') operations used
as operands in boolean combinations of subexpressions of the language
within a first projection or conditional operation, all intersection
and set difference operations can be eliminated, to the extent that
these operations occur in the language under consideration.

\end{theorem}

\begin{proof}
The proof goes by structural induction. Therefore, 
consider the expression $e_1\cap e_2$,
respectively, $e_1-e_2$ (to the extent these operations occur in the
language under consideration), where $e_1$ and $e_2$ are expressions
not containing eliminable intersection and set difference operations.
For $i=1,2$, we may write
$$e_i=c_{i0}/\down/c_{i1}/\down/\ldots/\down/c_{in_{i-1}}/\down/c_{in_i},$$
where, for $j=0,\ldots,n_i$, $c_{ij}$ is an expression in $\CL(E)$
with the property that, for each document $D$, $c_{ij}(D)\subseteq
\varepsilon(D)$. From here on, we consider both cases separately.
\begin{enumerate}

\item \emph{Intersection}. Clearly, if $n_1\neq n_2$, then, for each
  document $D$, $e_1\cap e_2(D)=\emptyset=\emptyset(D)$. In the other
  case, let $n\ass n_1=n_2$. For $j=0,\ldots,n$, let $c_j\ass
  \pi_1(c_{1j}\cap c_{2j})$, which is an expression of $\CL(E)$,
  equivalent to $c_{1j}\cap c_{2j}$. Let
$$e'\ass c_0/\down/c_1/\down/\ldots/\down/c_{n-1}/\down/c_n.$$
A straightforward set-theoretical argument
reveals that, for each document $D$, $e'(D)=e_1\cap e_2(D)$.

\item \emph{Difference}.  Clearly, if $n_1\neq n_2$, then, for each
  document $D$, $e_1-e_2(D)=e_1(D)$. In the other case, let $n\ass
n_1=n_2$. For $j=0,\ldots,n$, let $e'_j$ be $e_1$ in which $c_{1j}$
is replaced by $\pi_1(c_{1j}-c_{2j})$, which is an expression of $\CL(E)$,
  equivalent to $c_{1j}- c_{2j}$. Let
$$e'=e'_0\cup e'_1\cup\ldots\cup e'_{n-1}\cup e_n'.$$
which is also in $\CL(E)$. A straightforward set-theoretical argument
reveals that, for each document $D$, $e'(D)=e_1-e_2(D)$.

\end{enumerate}
\end{proof}

\begin{corollary}
\label{cor-downwardcorecoincides}

Let $E$ be a set of nonbasic operations containing
downward navigation (``$\down$'') and
first projection (``$\pi_1$''), and not containing upward navigation
(``$\up$''), and inverse
(``$.^{-1}$''). Then, for each
expression $e$ in $\BL(E)$, there exists an expression $e'$ in
$\CL(E)$ such that, for each document $D$, $e(D)=e'(D)$.

\end{corollary}

By Theorem~\ref{theo-downwardcorecoincides}, we may even disallow
set difference or intersection operations (to the extent they occur in
the language under consideration) except those used as operands
of boolean combinations of subexpressions inside a 
projection operation without loosing expressive
power.
%-------------------------------------------------------------------------s-
\subsection{Strictly downward languages not containing set difference}
\label{subsec-strictlydownnodifference}

So far, the characterizations of strictly downward languages involved
only languages containing the set difference operator. One could,
therefore, wonder if it is possible to provide similar
characterizations for languages not containing set difference. However,
the absence of set difference and the logical negation that is
inherently embedded in it has as a side effect that it is no longer
always possible to exploit equivalences or derive them.
%.........................................................................
\subsubsection{Weaker notions of downward and two-way distinguishability}
\label{subsubsec-downwardrelatedness}

Therefore, one would like to consider an asymmetric version 
of downward $k$-equivalence, say ``downward $k$-relatedness,'' which,
for the appropriate language could correspond to expression relatedness.
For $k=1$, such an approach could lead to the following definitions.

\begin{definition}
\label{def-downrelated}

Let $D=(V,\textit{Ed},r,\lambda)$ be a document, and let
$v_1,v_2\in V$. Then, 
\begin{enumerate}

\item $v_1$ and $v_2$ are \emph{downward-related}, denoted $v_1
  \downgeq{} v_2$, if
\begin{enumerate}

\item $\lambda(v_1)=\lambda(v_2)$; and

\item for each child $w_1$ of $v_1$, there exists a child
$w_2$ of $v_2$ such that $w_1\downgeq{} w_2$.

\end{enumerate}

\item $v_1$ and $v_2$ are \emph{weakly downward-equivalent}, denoted $v_1
  \weakdownequiv{} v_2$, if $v_1\downgeq{} v_2$ and $v_2\downgeq{} v_1$.

\end{enumerate}
\end{definition}

Obviously, downward 1-equivalence implies weak downward
equivalence. The converse, however, is \emph{not\/} true, as
illustrated by the following, simple example.

\begin{example}
\label{ex-weakdownwardequivalence}

Consider the document in Figure~\ref{fig-weak}. Labels have
been omitted, because they are not relevant in this discussion. (We
assume all nodes have the same label.) Obviously,
$x_1\downequiv{1} x_2$, hence $x_1\weakdownequiv{} x_2$. In particular,
$x_1\downgeq{} x_2$ and $x_2\downgeq{} x_1$. Also, $y_1\downgeq{}
x_2$, as the second condition to be verified is voidly satisfied in
this case. We may thus conclude that $v_1\weakdownequiv{}
v_2$. However, $v_1\not\downequiv{1} v_2$, as there is no child of
$v_2$ that is downward 1-equivalent to $y_1$.

\end{example}

\begin{figure}
\begin{center}
\resizebox{0.5\textwidth}{!}{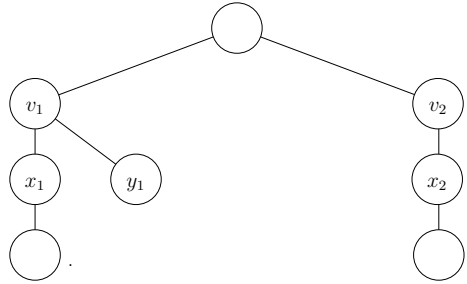}
\end{center}
\caption[a]{Document of Example~\ref{ex-weakdownwardequivalence}.}
\label{fig-weak}
\end{figure}

Notice that, in Example~\ref{ex-weakdownwardequivalence}, there is
even no child of $v_2$ that is \emph{weakly\/} downward equivalent to
$y_1$! Therefore, we shall not even attempt to generalize
Definition~\ref{def-downrelated} to the case where $k>1$, as there is
no straightforward way to adapt the third condition of
Definition~\ref{def-kdownwardequivalent}.

We conclude this digression on alternatives for downward
1-equivalence by providing analogue alternatives for 1-equivalence.

\begin{definition}
\label{def-weakequivalent}

Let $D=(V,\textit{Ed},r,\lambda)$ be a document, and let
$v_1,v_2\in V$. Then, 
\begin{enumerate}

\item $v_1$ and $v_2$ are \emph{related\/}, denoted $v_1\updowngeq{}
  v_2$, if
\begin{enumerate}

\item $v_1\downgeq{} v_2$;

\item $v_1$ is the root if and only if $v_2$ is the root; and

\item if $v_1$ and $v_2$ are not the root, and $u_1$ and $u_2$ are the
  parents of $v_1$ and $v_2$, respectively, then $u_1\updowngeq{} u_2$.

\end{enumerate}

\item $v_1$ and $v_2$ are \emph{weakly equivalent\/}, denoted
$v_1\weakupdownequiv{} v_2$,  if $v_1\updowngeq{} v_2$ and 
$v_2\updowngeq{} v_1$.

\end{enumerate}
\end{definition}

\begin{example}
\label{ex-weakequivalent}

Consider again the document in Figure~\ref{fig-weak}. 
Observe that $v_1\weakupdownequiv{} v_2$. Furthermore, $y_1\updowngeq{}
x_2$, but not the other way around.
\end{example}

\begin{table}
\caption{Distinguishability notions of Section \ref{subsubsec-downwardrelatedness}.}
\label{table:summary-weak-syntactic}
\begin{center}
\begin{tabular}{|ccc|}
    \hline
    {\em distinguishability  notion} & {\em notation} & {\em defined in}\\
    \hline
    downward-related & \downgeq{} & Definition \ref{def-downrelated}\\
    weakly-downward-equivalent & \weakdownequiv{} & Definition \ref{def-downrelated}\\
    related & \updowngeq{} & Definition \ref{def-weakequivalent}\\
    weakly-equivalent & \weakupdownequiv{} & Definition \ref{def-weakequivalent}\\
    \hline
\end{tabular}
\end{center}
\end{table}

Table \ref{table:summary-weak-syntactic} summarizes all of the distinguishability notions
presented in this section.

The following analogue of Proposition~\ref{prop-kequivalent} is
straightforward. 

\begin{proposition}
\label{prop-weakequivalent}

Let $D=(V,\textit{Ed},r,\lambda)$ be a document, and let
$v_1,v_2\in V$. Then,
\begin{enumerate}

\item\label{weakequivalent-1} $v_1\updowngeq{} v_2$ 
if and only if $(r,v_1)\pairs{\sigequiv}{\downgeq{}} (r,v_2)$; and

\item\label{weakequivalent-2}  $v_1\weakupdownequiv{} v_2$ 
if and only if $(r,v_1)\pairs{\sigequiv}{\weakdownequiv{}} (r,v_2)$.

\end{enumerate}
\end{proposition}
%.........................................................................
\subsubsection{Towards characterizing expression equivalence and
  navigational expressiveness}
\label{subsubsec-downwardpositivexpath}

The approach we shall take here is reviewing the results in
Sections~\ref{subsec-strictlydownsufficient}--\ref{subsec-strictlydowncharacterization}
and examine to which extent these results in the case where $k=1$
allow replacing downward 1-equivalence by weak downward equivalence.

We start by observing that the analogue of
Lemma~\ref{lem-strictlydown1} does not hold. Indeed, in the example
document of Example~\ref{ex-weakdownwardequivalence}, shown in
Figure~\ref{fig-weak}, $v_1\weakdownequiv{} v_2$. Also, there is no
child of $v_2$ that is weakly downward equivalent to $x_1$. Hence,
there is no node node $z$ for which
$(v_1,x_1)\pairs{\sigequiv}{\weakdownequiv{1}} (v_2,z)$. On the other
hand, we can restrict Lemma~\ref{lem-strictlydown1} to downward
relatedness:

\begin{lemma}
\label{lem-weakstrictlydown1}

Let $D=(V,\textit{Ed},r,\lambda)$ be a document, let $v_1$, $w_1$, and
$v_2$ be nodes of $D$ such that $w_1$ is a descendant of $v_1$. 
 If $v_1\downgeq{} v_2$, then $v_2$ has a descendant
$w_2$ in $D$ such that $(v_1,w_1)\pairs{\sigequiv}{\downgeq{}} (v_2,w_2)$.

\end{lemma}

Proposition~\ref{prop-strictlydown} relies on
Lemma~\ref{lem-strictlydown1} to prove the inductive step for the
first projection (``$\pi_1$''). It therefore comes as no surprise
that we cannot replace downward 1-equivalence by weak downward
equivalence, there. Indeed, consider the expression
$e\ass\pi_1(\down/(\varepsilon-\pi_1(\down)))$. In the example
document of Example~\ref{ex-weakdownwardequivalence}, shown in
Figure~\ref{fig-weak}, $v_1\weakdownequiv{} v_2$, and, hence,
$(v_1,v_1)\pairs{\sigequiv}{\weakdownequiv{}} (v_2,v_2)$. Moreover,
$(v_1,v_1)\in e(D)$. However, $(v_2,v_2)\notin e(D)$. However, we can
``save'' Proposition~\ref{prop-strictlydown} by replacing downward
1-equivalence by downward \emph{relatedness\/}, provided we omit set
difference (``$-$'') from the set of operations of the
language. Indeed, we can then recover the proof, using
Lemma~\ref{lem-weakstrictlydown1} instead of
Lemma~\ref{lem-strictlydown1}. (Notice that, for the induction
step for set difference in the original proof, we must exploit
equivalence in both directions to deal with the negation inherent to
the difference operation.) In summary, we have the following.

\begin{lemma}
\label{lem-weakstrictlydown}

Let $E$ be the set of all nonbasic operations in
Table~\ref{tab-binops}, except for upward navigation (``$\up$''),
second projection
(``$\pi_2$''), inverse (``$.^{-1}$''), selection on at least $k$
children satisfying some condition (``$\ch{k}(.)$'') for $k>1$, and
set difference (``$-$''). Let $e$ be an expression in $\BL(E)$.  Let
$D=(V,\textit{Ed},r,\lambda)$ be a document, let $v_1$, $w_1$, $v_2$,
and $w_2$ be nodes of $D$ such that $w_1$ is a descendant of $v_1$ and
$w_2$ is a descendant of $v_2$.  Assume furthermore that
$(v_1,w_1)\pairs{\sigequiv}{\downgeq{}} (v_2,w_2)$. Then,
$(v_1,w_1)\in e(D)$ implies $(v_2,w_2)\in e(D)$.

\end{lemma}

Two applications of Lemma~\ref{lem-weakstrictlydown} immediately yield
the following.

\begin{proposition}
\label{prop-weakstrictlydown}

Let $E$ be the set of all nonbasic operations in
Table~\ref{tab-binops}, except for upward navigation (``$\up$''),
second projection
(``$\pi_2$''), inverse (``$.^{-1}$''), selection on at least $k$
children satisfying some condition (``$\ch{k}(.)$'') for $k>1$, and
set difference (``$-$''). Let $e$ be an expression in $\BL(E)$.  Let
$D=(V,\textit{Ed},r,\lambda)$ be a document, let $v_1$, $w_1$, $v_2$,
and $w_2$ be nodes of $D$ such that $w_1$ is a descendant of $v_1$ and
$w_2$ is a descendant of $v_2$.  Assume furthermore that
$(v_1,w_1)\pairs{\sigequiv}{\weakdownequiv{}} (v_2,w_2)$. Then,
$(v_1,w_1)\in e(D)$ if and only $(v_2,w_2)\in e(D)$.

\end{proposition}

So, Proposition~\ref{prop-weakstrictlydown} is weaker than
Proposition~\ref{prop-strictlydown} in the sense that we had to exclude set
difference, but stronger in the sense that, in return, we were able to
replace the precondition by a weaker one.

The analogues of Corollaries~\ref{cor-strictlydown1} 
and~\ref{cor-strictlydown2} are now as follows.

\begin{corollary}
\label{cor-weakstrictlydown1}

Let $E$ be the set of all nonbasic operations in
Table~\ref{tab-binops}, except for upward navigation (``$\up$''),
second projection
(``$\pi_2$''), inverse (``$.^{-1}$''), selection on at least $k$
children satisfying some condition (``$\ch{k}(.)$'') for $k>1$, and
set difference (``$-$''). Let $e$ be an expression in $\BL(E)$.  Let
$D=(V,\textit{Ed},r,\lambda)$ be a document, let $v_1$ and $v_2$ be
nodes of $D$ such that $v_1\downgeq{} v_2$ and let $w_1$ be a
descendant of $v_1$. If $(v_1,w_1)\in e(D)$, then there exists a
descendant $w_2$ of $v_2$ such that $(v_2,w_2)\in e(D)$.

\end{corollary}

In other words, downward relatedness implies expression relatedness.

\begin{corollary}
\label{cor-weakstrictlydown2}

Let $E$ be a set of nonbasic operations not containing upward
navigation (``$\up$''), second
projection (``$\pi_2$''), inverse (``$.^{-1}$''), selection on at
least $k$ children satisfying some condition (``$\ch{k}(.)$'') for
$k>1$, and set difference (``$-$'').  Consider the language $\BL(E)$
or $\CL(E)$. Let $D=(V,\textit{Ed},r,\lambda)$ be a document, and let
$v_1$ and $v_2$ be nodes of $D$. If $v_1\weakdownequiv{} v_2$, then
$v_1\expequiv v_2$.

\end{corollary}

\begin{proof}
The condition $v_1\weakdownequiv{} v_2$ implies $v_1\downgeq{} v_2$
and $v_2\downgeq{} v_1$. By Corollary~\ref{cor-weakstrictlydown1}, these
conditions in turn imply $v_1\expgeq v_2$ and $v_2\expgeq v_1$, which
together are equivalent to $v_1\expequiv v_2$.
\end{proof}

We now look to necessary conditions for expression equivalence for
strictly downward languages not containing set difference. Provided
intersection (``$\cap$'') is available, the expressibility of set
  difference is used only once in the proof of
  Proposition~\ref{prop-downnecessary}, namely where
Proposition~\ref{prop-notempty-empty} is invoked. We do not need this
Proposition, however, in the following variation of
Proposition~\ref{prop-downnecessary}:

\begin{lemma}
\label{lem-weakdownnecessary}

Let $E$ be a set of nonbasic operations containing
first projection (``$\pi_1$''), and
intersection (``$\cap$'').  Consider the language $\BL(E)$ or
$\CL(E)$. Let $D=(V,\textit{Ed},r,\lambda)$ be a
document, and let $v_1$ and $v_2$ be nodes of $D$. If $v_1\expgeq
v_2$, then $v_1\downgeq{} v_2$.

\end{lemma}

Two applications of Lemma~\ref{lem-weakdownnecessary} immediately
yield the following.

\begin{proposition}
\label{prop-weakdownnecessary}

Let $E$ be a set of nonbasic operations containing
first projection (``$\pi_1$''), and
intersection (``$\cap$'').  Consider the language $\BL(E)$ or
$\CL(E)$. Let $D=(V,\textit{Ed},r,\lambda)$ be a
document, and let $v_1$ and $v_2$ be nodes of $D$. If $v_1\expequiv
v_2$, then $v_1\weakdownequiv{} v_2$.

\end{proposition}

So, Proposition~\ref{prop-weakdownnecessary} is weaker than
Proposition~\ref{prop-downnecessary} in the sense that the conclusion
is replaced by a weaker one, but stronger in the sense that, in
return, we no longer have to rely on the presence of difference.

The languages containing downward navigation (``$\down$'') and
satisfying both Corollary~\ref{cor-weakstrictlydown2} 
and Proposition~\ref{prop-weakdownnecessary} are
$\BL(\down,\pi_1,\cap)$ and $\CL(\down,\pi_1,\cap)$, which, moreover,
are equivalent, by Corollary~\ref{cor-downwardcorecoincides}. In
addition, we can eliminate intersection operations
except those used as operands of boolean combinations of
subexpressions inside a projection operation without
loosing expressive power. We call these languages the \emph{strictly
  downward positive XPath algebra\/} and the
\emph{strictly downward core positive XPath algebra\/},
respectively. Combining the aforementioned results, we 
get the following.

\begin{theorem}
\label{theo-weakdownequivalent}

Consider the strictly downward (core) positive XPath algebra.
Let $D=(V,\textit{Ed},r,\lambda)$ be a document, and let
$v_1$ and $v_2$ be nodes of $D$.  Then $v_1\expequiv v_2$ if and only
if $v_1\weakdownequiv{} v_2$.

\end{theorem}

We finally turn to the characterization of navigational
expressiveness. Proposition~\ref{prop-strictlydownconverse} and its
proof, and hence also Corollary~\ref{cor-strictlydown}, carry over to 
the current setting.

\begin{theorem}
\label{theo-weakstrictlydown}

Consider the strictly downward (core) positive XPath algebra.
Let $D=(V,\textit{Ed},r,\lambda)$ be a document, and let $v_1$, $w_1$,
$v_2$, and $w_2$ be nodes of $D$ such that $w_1$ is a descendant of
$v_1$ and $w_2$ is a descendant of $v_2$. 
Then, the property that, for each expression $e$ in the
language under consideration, $(v_1,w_1)\in e(D)$ if and only if
$(v_2,w_2)\in e(D)$ is equivalent to the property
$(v_1,w_1)\pairs{\sigequiv}{\weakdownequiv{}} (v_2,w_2)$.

\end{theorem}

To derive a BP-result for the strictly downward (core) positive XPath
algebra, we observe that Lemmas~\ref{lem-separationdown1}
and~\ref{lem-separationdown2} and Theorem~\ref{theo-strictlydown-bp}
carry over to the current context, provided we replace downward
1-equivalence by downward relatedness.

\begin{lemma}
\label{lem-weakseparationdown}

Let $D=(V,\textit{Ed},r,\lambda)$ be a document.
\begin{enumerate}

\item Let $v_1$ be a node of $D$.  There exists an expression $e_{v_1}$ in
the strictly downward (core) positive XPath algebra such that, for
each node $v_2$ of $D$, $e_{v_1}(D)(v_2)\neq\emptyset$ if and only if
$v_1\downgeq{} v_2$. 

\item Let $v_1$ and $w_1$ be a nodes of $D$ such that $w_1$ is a
  descendant of $v_1$. There exists an expression $e_{(v_1,w_1)}$ in
  the strictly downward (core) positive XPath algebra such that, for
  all nodes $v_2$ and $w_2$ of $D$ with $w_2$ a descendant of $v_2$,
  $(v_2,w_2)\in e_{(v_1,w_1)}(D)$ if and only if
  $(v_1,w_1)\pairs{\sigequiv}{\downgeq{}} (v_2,w_2)$.

\end{enumerate}
\end{lemma}

In the proof of the first claim, the role of
Theorem~\ref{theo-downequivalent} is taken over by
Corollary~\ref{cor-weakstrictlydown1}.

\begin{theorem}
\label{theo-weakstrictlydown-bp}

Let $D=(V,\textit{Ed},r,\lambda)$ be a document, and let
$R\subseteq V\times V$. Then, there exists an expression $e$ in the
strictly downward (core) positive XPath algebra
such that $e(D)=R$ if and only if,
\begin{enumerate}

\item for all $v,w\in V$, $(v,w)\in R$ implies $w$ is a descendant
  of~$v$; and,

\item for all $v_1,w_1,v_2,w_2\in V$ with $w_1$ a descendant of $v_1$,
  $w_2$ a descendant of $v_2$, and
  $(v_1,w_1)\pairs{\sigequiv}{\downgeq{}} (v_2,w_2)$, 
$(v_1,w_1)\in R$ implies $(v_2,w_2)\in R$. 

\end{enumerate}
\end{theorem}

The major difference between Theorems~\ref{theo-strictlydown-bp}
and~\ref{theo-weakstrictlydown-bp} is that, in the former, $R$ is a
partition of maximal sets of $\downequiv{k}$-congruent nodes, while,
in the latter, $R$ is merely closed under $\downgeq{}$-congruence.

We can also recast Theorem~\ref{theo-weakstrictlydown-bp} in terms of
node-level navigation, in much the same way as
Theorem~\ref{theo-strictlydown-bp}.

\begin{theorem}
\label{theo-weakstrictlydown-nodelevel}

Let $D=(V,\textit{Ed},r,\lambda)$ be a document, let
$v$ be a node of $D$, and let $W\subseteq V$. Then there exists an
expression $e$ in the strictly downward (core) positive XPath algebra
such that $e(D)(v)=W$ if and only if all nodes
of~$W$ are descendants of $v$, and, for all nodes $w_1$ and $w_2$ of $D$ with
$(v,w_1)\pairs{\sigequiv}{\downgeq{}} (v,w_2)$, $w_1\in W$ implies
$w_2\in W$.

\end{theorem}

\begin{corollary}
\label{cor-weakstrictlydown-root}

Let $D=(V,\textit{Ed},r,\lambda)$ be a document, 
and let $W\subseteq V$. Then there exists an 
expression $e$ in the strictly downward (core) positive XPath algebra such that
$e(D)(r)=W$ if and only if, for all nodes $w_1$ and $w_2$ of $D$ with
$w_1\updowngeq{} w_2$, $w_1\in W$ implies $w_2\in W$. 

\end{corollary}

For the last result of this section, we relied on
Proposition~\ref{prop-weakequivalent}, (\ref{weakequivalent-1}).

%======================================================================
\section{Weakly downward languages}
\label{sec-weaklydownward}

We now turn to \emph{weakly downward\/} languages: for any node $v$ of
the document $D$ under consideration, all nodes in $e(D)(v)$ are
descendants of $v$, but there are possibly nodes $v$ for which $e(D)(v)\neq
e(D')(v)$, with $D'$ the subtree of $D$ rooted at $v$.
%-------------------------------------------------------------------------
\subsection{Sufficient conditions for expression-equivalence}
\label{subsec-weaklydownsufficient}

The key notion in
Sections~\ref{subsec-weaklydownsufficient}--\ref{subsec-weaklydowncharacterization}
is $\updownequiv{k}$-congruence, $k\geq 1$, restricted to
ancestor-descendant pairs. We first explore some properties of this
notion.

\begin{lemma}
\label{lem-weaklydown1}

Let $D=(V,\textit{Ed},r,\lambda)$ be a document, let $v_1$, $w_1$,
$v_2$, and $w_2$ be nodes of $D$ such that $w_1$ is a descendant of
$v_1$ and $w_2$ is a descendant of $v_2$, and let
$k\geq 1$. Then, $(v_1,w_1)\pairs{\sigequiv}{\updownequiv{k}}
(v_2,w_2)$ if and only if $(v_1,w_1)\sigequiv(v_2,w_2)$ and
$w_1\updownequiv{k} w_2$. 

\end{lemma}

\begin{proof}
As the ``only if'' is obvious, we focus on the ``if.'' By
Proposition~\ref{prop-kequivalent}, $w_1\updownequiv{k} w_2$ implies
that $(r,w_1)\pairs{\sigequiv}{\downequiv{k}} (r,w_2)$. Let $y_1$ be a
node on the path from~$v_1$ to~$w_1$, and let $y_2$ be the
corresponding node on the path from~$v_2$ to~$w_2$. By
Proposition~\ref{prop-subpairs}, we also have that
$(r,y_1)\pairs{\sigequiv}{\downequiv{k}} (r,y_2)$. By another
application of Proposition~\ref{prop-kequivalent}, we finally deduce
that $y_1\updownequiv{k} y_2$.
\end{proof}

\begin{lemma}
\label{lem-weaklydown2}

Let $D=(V,\textit{Ed},r,\lambda)$ be a document, let $v_1$ and $w_1$
be nodes of $D$ such that $w_1$ is a descendant of $v_1$, and let
$k\geq 1$. Then,
\begin{enumerate}

\item\label{weakdown1} each node $v_2$ of $D$ for which
  $v_1\updownequiv{k} v_2$ has a descendant
$w_2$ in $D$ such that $(v_1,w_1)\pairs{\sigequiv}{\updownequiv{k}}
  (v_2,w_2)$; and

\item\label{weakdown2} each node $w_2$ of $D$ for which
  $w_1\updownequiv{k} w_2$ has an an ancestor
$v_2$ in $D$ such that $(v_1,w_1)\pairs{\sigequiv}{\updownequiv{k}}
  (v_2,w_2)$.

\end{enumerate}
\end{lemma}

\begin{proof}
To see (\ref{weakdown1}), we know by Lemma~\ref{lem-strictlydown1}
that $v_2$ has a descendant $w_2$ such that
$(v_1,w_1)\pairs{\sigequiv}{\downequiv{k}} (v_2,w_2)$. By
Proposition~\ref{prop-kequivalent}, we also have that 
$(r,v_1)\pairs{\sigequiv}{\downequiv{k}} (r,v_2)$. It now readily
follows that $(r,w_1)\pairs{\sigequiv}{\downequiv{k}} (r,w_2)$, or,
again by Proposition~\ref{prop-kequivalent}, $w_1\updownequiv{k}
w_2$. It now follows from Lemma~\ref{lem-weaklydown1} that 
$(v_1,w_1)\pairs{\sigequiv}{\updownequiv{k}}
  (v_2,w_2)$.

Claim (\ref{weakdown2}) can be shown by induction on the length of the
path from $v_1$ to $w_1$. If $v_1=w_1$, then obviously, we must choose
$v_2\ass w_2$. If $v_1\neq w_1$, we have in particular that $w_1\neq
r$, and, hence, by $w_1\updownequiv{k} w_2$, that $w_2\neq r$. Let
$y_1$ be the parent of $w_1$ and $y_2$ be the parent of $w_2$. By
definition, $y_1\updownequiv{k} y_2$, and, by the induction
hypothesis there is a node $v_2$ in $D$ such that
$(v_1,y_1)\pairs{\sigequiv}{\updownequiv{k}} (v_2,y_2)$. It now readily
follows that $(v_1,w_1)\pairs{\sigequiv}{\updownequiv{k}} (v_2,w_2)$.
\end{proof}

We now link $\updownequiv{k}$-congruence of ancestor-descendant pairs of
nodes with expressibility in weakly downward languages.

\begin{proposition}
\label{prop-weaklydown}

Let $k\geq 1$, and let $E$ be the set of all nonbasic operations in
Table~\ref{tab-binops}, except for upward navigation (``$\up$''),
inverse (``$.^{-1}$''), and
selection on at least $m$ children 
satisfying some condition (``$\ch{m}(.)$'') for $m>k$. Let $e$ be an
expression in $\BL(E)$. 
Let $D=(V,\textit{Ed},r,\lambda)$ be a
document, and let $v_1$, $w_1$, $v_2$, 
and $w_2$ be nodes of $D$ such that $w_1$ is a descendant of $v_1$ and
$w_2$ is a descendant of $v_2$.  Assume furthermore that
$(v_1,w_1)\pairs{\sigequiv}{\updownequiv{k}} (v_2,w_2)$. Then,
$(v_1,w_1)\in e(D)$ if and only if $(v_2,w_2)\in e(D)$.

\end{proposition}

\begin{proof}
The proof goes along the same lines of the proof of
Proposition~\ref{prop-strictlydown}. Actually, since
$\updownequiv{k}$-congruence implies $\downequiv{k}$-congruence,
almost all of the proof by structural induction can be used here
verbatim, except, of course, for the inductive step for the second
projection (``$\pi_2$''), which we consider next. Thus, let $e\ass
\pi_2(f)$, with $f$ satisfying Proposition~\ref{prop-weaklydown}.  If
$(v_1,w_1)\in\pi_2(f)$, then, of course, $v_1=w_1$ as a consequence of
which $v_2=w_2$. Also, there exists $y_1\in V$ such that $(y_1,v_1)\in
f(D)$. By Lemma~\ref{lem-weaklydown2}, (\ref{weakdown2}), there exists
$y_2\in V$ such that $(y_1,v_1)\pairs{\sigequiv}{\updownequiv{k}}
(y_2,v_2)$. By the induction hypothesis, $(y_2,v_2)\in f(D)$. Hence,
$(v_2,v_2)\in \pi_2(f)(D)$.
\end{proof}

By combining Proposition~\ref{prop-weaklydown} with
Lemma~\ref{lem-weaklydown2}, we can establish the following.

\begin{corollary}
\label{cor-weaklydown1}

Let $k\geq 1$, and let $E$ be the set of all nonbasic operations in
Table~\ref{tab-binops}, except for upward navigation (``$\up$''),
inverse (``$.^{-1}$''),
and selection on at least $m$ children
(``$\ch{m}(.)$'') for $m>k$. Let $e$ be an expression in $\BL(E)$. Let
$D=(V,\textit{Ed},r,\lambda)$ be a document, let $v_1$ and $w_1$
be nodes of $D$ such that $w_1$ is a descendant of $v_1$ and
$(v_1,w_1)\in e(D)$. Then,
\begin{enumerate}

\item\label{corweakdown1} each node $v_2$ of $D$ for which
  $v_1\updownequiv{k} v_2$ has a descendant
$w_2$ in $D$ such that $(v_2,w_2)\in e(D)$; and

\item\label{corweakdown2} each node $w_2$ of $D$ for which
  $w_1\updownequiv{k} w_2$ has an an ancestor
$v_2$ in $D$ such that $(v_2,w_2)\in e(D)$.

\end{enumerate}
\end{corollary}

Finally, we infer the following from
Corollary~\ref{cor-strictlydown2}, (\ref{corweakdown1}):

\begin{corollary}
\label{cor-weaklydown2}

Let $k\geq 1$, and let $E$ be a set of nonbasic operations
not containing upward navigation (``$\up$''),
inverse (``$.^{-1}$''),
and selection on at least $m$ children
satisfying some condition (``$\ch{m}(.)$'') for $m>k$. 
Consider the language $\BL(E)$ or $\CL(E)$. Let
$D=(V,\textit{Ed},r,\lambda)$ be a document, and let $v_1$ and $v_2$
be nodes of $D$. If $v_1\updownequiv{k} v_2$, then $v_1\expequiv v_2$.

\end{corollary}
%-------------------------------------------------------------------------
\subsection{Necessary conditions for expression equivalence}
\label{subsec-weaklydownnecessary}

We now explore requirements on the set of nonbasic operations
expressible in the language under which downward-$k$-equivalence
($k\geq 1$) is a necessary condition for expression-equivalence.
As we have endeavored to make as few assumptions as possible,
Proposition~\ref{prop-weaklydownnecessary} also holds for a class of languages
that are \emph{not} downward.

\begin{proposition}
\label{prop-weaklydownnecessary}

Let $k\geq 1$, and let $E$ be a set of nonbasic operations containing
at least one navigation operation (``$\down$'' or ``$\up$'') and set
difference (``$-$''). Consider the language $\BL(E)$ or $\CL(E)$, and
assume that, in this language, 
first and second projection (``$\pi_1$'' and ``$\pi_2$'') can be
expressed, as well as selection on at least $m$ 
children satisfying some condition (``$\ch{m}(.)$''),
for all $m=1,\ldots,k$.  Let $D=(V,\textit{Ed},r,\lambda)$ be a
document, and let $v_1$ and $v_2$ be nodes of $D$. If $v_1\expequiv
v_2$, then $v_1\updownequiv{k} v_2$.

\end{proposition}

\begin{proof}
Without loss of generality, we may assume that the language under
consideration is $\CL(E)$. In Proposition~\ref{prop-downnecessary}, we
have already established that $v_1\expequiv v_2$ implies
$v_1\downequiv{k} v_2$. By induction on the length of the path from
$r$ to~$v_1$, we establish that, furthermore, $v_1\updownequiv{k}
v_2$. For the basis of the induction, consider the case that
$v_1=r$. Let $d$ be the length of a longest path from $r$ to a leaf
of~$D$ (i.e., the height of the tree). We distinguish two cases:
\begin{enumerate}

\item $\down\in E$. Then,
$\down^d(D)(v_1)\neq\emptyset$. Hence, $\down^d(D)(v_2)\neq\emptyset$,
which implies $v_2=r$. 

\item $\up\in E$ Then, $\pi_2(\up^d)(D)(v_1)\neq\emptyset$. Hence,
$\pi_2(\up^d)(D)(v_2)\neq\emptyset$, which implies $v_2=r$. 

\end{enumerate}
In both cases, it follows that $v_1\updownequiv{k} v_2$.
For the induction step, consider the case that
$v_1\neq r$. Again, we distinguish two cases:
\begin{enumerate}

\item $\down\in E$. Then, $\pi_2(\down)(D)(v_1)\neq\emptyset$, and, hence,
$\pi_2(\down)(D)(v_2)\neq\emptyset$. So, $v_2\neq r$. 

\item $\up\in E$. Then, $\up(D)(v_1)\neq\emptyset$, and, hence,
$\up(D)(v_2)\neq\emptyset$. So, $v_2\neq r$. 

\end{enumerate}
Now, let $u_1$ be the
parent of $v_1$ and $u_2$ be the parent of $v_2$. We show that
$u_1\expequiv u_2$. Thereto, let $e$ be an expression in the language
under consideration for which $e(D)(u_1)\neq\emptyset$. Again, we
distinguish two cases:
\begin{enumerate}

\item $\down\in E$. Then, $\pi_2(e/\down)(v_1)\neq\emptyset$. Since
  $v_1\expequiv v_2$,  $\pi_2(e/\down)(v_2)\neq\emptyset$. It follows that
$e(D)(u_2)\neq\emptyset$. 

\item $\up\in E$. Then, $\up/e(v_1)\neq\emptyset$. Since
  $v_1\expequiv v_2$,  $\up/e(v_2)\neq\emptyset$. It follows that
$e(D)(u_2)\neq\emptyset$. 

\end{enumerate}
By the induction hypothesis, we may now
conclude that, in both cases, $u_1\updownequiv{k} u_2$. Hence, also
$v_1\updownequiv{k} v_2$.
\end{proof}

We see that Proposition~\ref{prop-weaklydownnecessary} is as well 
applicable to weakly downward languages as to weakly upward languages
(see Section~\ref{subsec-weaklyupward}). We shall see in
Section~\ref{subsec-weaklyupward} that this is no coincidence. For
now, we suffice with concluding that $k$-equivalence is a
necessary condition for expression-equivalence under a weakly
downward language containing downward navigation (``$\down$''), both
projections (``$\pi_1$'' and ``$\pi_2$''), and set difference
(``$-$''), provided  selection on at least $m$ children satisfying
some condition (``$\ch{m}$'') for all $m=1,\ldots,k$ can be expressed.
%-------------------------------------------------------------------------
\subsection{Characterization of expression equivalence}
\label{subsec-weaklydowncharacterization}

The weakly downward languages containing downward navigation
(``$\down$'') and satisfying both Corollary~\ref{cor-weaklydown2} 
of Subsection~\ref{subsec-weaklydownsufficient}
and Proposition~\ref{prop-downnecessary} 
of Subsection~\ref{subsec-strictlydownnecessary} 
are  $$\BL(\down,\pi_1,\pi_2,\ch{1}(.),\ldots,\ch{k}(.),-) \text{~and~}
\CL(\down,\pi_1,\pi_2,\ch{1}(.),\ldots,\ch{k}(.),-),$$ which, moreover, are
equivalent, by Corollary~\ref{cor-downwardcorecoincides}. In addition,
we can eliminate set difference or intersection operations except
those used as operands of boolean combinations of subexpressions
inside a projection operation without loosing expressive
power. We call these languages
the \emph{weakly downward XPath algebra with counting up to~$k$\/} and the
\emph{weakly downward core XPath algebra with counting up to~$k$\/},
respectively. Combining the aforementioned results, we get the following.

\begin{theorem}
\label{theo-weaklydownequivalent}

Let $k\geq 1$, and consider the weakly downward (core) XPath algebra with
counting up to~$k$.  Let $D=(V,\textit{Ed},r,\lambda)$ be a document,
and let $v_1$ and $v_2$ be nodes of $D$. Then $v_1\expequiv v_2$, if
and only if $v_1\updownequiv{k} v_2$.

\end{theorem}

A special case arises when $k=1$, since selection on at least one
child satisfying some condition (``$\ch{1}(.)$'') can be expressed in
terms of the other operations required by
Theorem~\ref{theo-weaklydownequivalent}, by Proposition~\ref{prop-counting}.
The languages we then obtain are called the \emph{weakly downward XPath
  algebra} and the \emph{weakly downward core XPath algebra}, respectively.
We have the following.

\begin{corollary}
\label{cor-weaklydownequivalent}

Consider the weakly downward (core) XPath algebra.  Let
$D=(V,\textit{Ed},r,\lambda)$ be a document, and let $v_1$ and $v_2$
be nodes of $D$. Then $v_1\expequiv v_2$, if and only if
$v_1\updownequiv{1} v_2$.

\end{corollary}
%-------------------------------------------------------------------------
\subsection{Characterization of navigational expressiveness}

We start by proving a converse to Proposition~\ref{prop-weaklydown}.

\begin{proposition}
\label{prop-weaklydownconverse}

Let $k\geq 1$, and let $E$ be a set of nonbasic operations containing
downward navigation (``$\down$'') and set difference (``$-$'').  Consider
the language $\BL(E)$ or $\CL(E)$. Assume that, in this language,
first and second projection (``$\pi_1$'' and ``$\pi_2$'') can be
expressed, as well as selection on at least $m$ children satisfying
some condition (``$\ch{m}(.)$''), for all $m=1,\ldots,k$. Let
$D=(V,\textit{Ed},r,\lambda)$ be a document, and let $v_1$, $w_1$,
$v_2$, and $w_2$ be nodes of $D$ such that $w_1$ is a descendant
of~$v_1$ and $w_2$ is a descendant of~$v_2$. Assume furthermore that,
for each expression~$e$ in the language, $(v_1,w_1)\in e(D)$ if and
only if $(v_2,w_2)\in e(D)$. Then
$(v_1,w_1)\pairs{\sigequiv}{\updownequiv{k}}(v_2,w_2)$.

\end{proposition}

\begin{proof}
From Proposition~\ref{prop-strictlydownconverse}, we already know that
$(v_1,w_1)\pairs{\sigequiv}{\downequiv{k}}(v_2,w_2)$. In particular,
$(v_1,w_1)\sigequiv(v_2,w_2)$. By Lemma~\ref{lem-weaklydown1}, it
suffices to prove that $v_1\updownequiv{k} w_2$, or, by
Proposition~\ref{prop-kequivalent}, that
$(r,w_1)\pairs{\sigequiv}{\downequiv{k}}(r,w_2)$. In view of what we
already know, we only need to show that
$(r,v_1)\pairs{\sigequiv}{\downequiv{k}}(r,v_2)$. Since $(v_1,w_1)\in
\pi_2(\sig(r,v_1))/\sig(v_1,w_1)$, it follows that also $(v_2,w_2)\in
\pi_2(\sig(r,v_1))/\sig(v_1,w_1)$, for which we readily deduce that
$(r,v_1)\sigequiv (r,v_2)$. Let $u_1$ be a node
on the path from $r$ to~$v_1$, and let $u_2$ be the corresponding
node on the path from $r$ to $v_2$. Then,
$(r,u_1)\sigequiv (r,u_2)$ and $(u_1,v_1)\sigequiv
(u_2,u_2)$. Now, let $f$ be any expression in the language such that
$f(D)(u_1)\neq\emptyset$. Then, $(u_1,u_1)\in \pi_1(f)(D)$. Let
$e\ass \pi_2(\pi_1(f)/\sig(u_1,v_1))/\sig(v_1,w_1)$. By construction,
$(v_1,w_1)\in e(D)$. Hence, by assumption, $(v_2,w_2)\in e(D)$, which
implies $(u_2,u_2)\in \pi_1(f)(D)$ or $f(D)(u_2)\neq\emptyset$. The same
holds vice versa, and we may thus conclude that $u_1\expequiv u_2$, and,
hence, by Proposition~\ref{prop-downnecessary}, $u_1\downequiv{k} u_2$.
We may thus conclude that $(r,v_1)\downequiv{k}(r,v_2)$.
\end{proof}

Combining Propositions~\ref{prop-weaklydown}
and~\ref{prop-weaklydownconverse}, we obtain the following.

\begin{corollary}
\label{cor-weaklydown}

Let $k\geq 1$, and consider the weakly downward (core) XPath algebra with
counting up to~$k$. Let
$D=(V,\textit{Ed},r,\lambda)$ be a document, and let $v_1$, $w_1$,
$v_2$, and $w_2$ be nodes of $D$ such that $w_1$ is a descendant of
$v_1$ and $w_2$ is a descendant of $v_2$. 
Then, the property that, for each expression $e$ in the
language under consideration, $(v_1,w_1)\in e(D)$ if and only if
$(v_2,w_2)\in e(D)$ is equivalent to the property
$(v_1,w_1)\pairs{\sigequiv}{\updownequiv{k}} (v_2,w_2)$.

\end{corollary}

From here on, the derivation of a BP-result for the weakly downward
(core) XPath algebra with counting up to~$k$ follows the development
in Section~\ref{subsec-strictlybp} very closely, which is why we only
state the final results.

\begin{theorem}
\label{theo-weaklydown-bp}

Let $k\geq 1$.
Let $D=(V,\textit{Ed},r,\lambda)$ be a document, and
let $R\subseteq V\times V$. Then, there exists an expression $e$ in the
weakly downward (core) XPath algebra with counting up to~$k$ such that
$e(D)=R$ if and only if, 
\begin{enumerate}

\item for all $v,w\in V$, $(v,w)\in R$ implies $w$ is a descendant
  of~$v$; and,

\item for all $v_1,w_1,v_2,w_2\in V$ with $w_1$ a descendant of $v_1$,
  $w_2$ a descendant of $v_2$, and
  $(v_1,w_1)\pairs{\sigequiv}{\updownequiv{k}} (v_2,w_2)$, 
$(v_1,w_1)\in R$ implies $(v_2,w_2)\in R$. 

\end{enumerate}
\end{theorem}

The specialization to the weakly downward (core) XPath algebra is as
follows.

\begin{corollary}
\label{cor-weaklydown-bp}

Let $D=(V,\textit{Ed},r,\lambda)$ be a document, and let $R\subseteq
V\times V$. There exists an expression $e$ in the weakly downward (core)
XPath algebra such that $e(D)=R$ if and only if,
\begin{enumerate}

\item for all $v,w\in V$, $(v,w)\in R$ implies $w$ is a descendant
  of~$v$; and,

\item for all $v_1,w_1,v_2,w_2\in V$ with $w_1$ a descendant of $v_1$,
  $w_2$ a descendant of $v_2$, and
  $(v_1,w_1)\pairs{\sigequiv}{\updownequiv{1}} (v_2,w_2)$, 
  $(v_1,w_1)\in R$ implies $(v_2,w_2)\in R$. 

\end{enumerate}
\end{corollary}

We recast Theorem~\ref{theo-weaklydown-bp} and
Corollary~\ref{cor-weaklydown-bp} in terms of
node-level navigation.

\begin{theorem}
\label{theo-weaklydown-nodelevel}

Let $k\geq 1$.  Let $D=(V,\textit{Ed},r,\lambda)$ be a document, let
$v$ be a node of $D$, and let $W\subseteq V$. Then there exists an
expression $e$ in the weakly downward (core) XPath algebra with counting up
to~$k$ such that $e(D)(v)=W$ if and only if all nodes of~$W$ are
descendants of $v$, and, for all $w_1,w_2\in W$ with
$(v,w_1)\pairs{\sigequiv}{\updownequiv{k}} (v,w_2)$, $w_1\in W$ implies
$w_2\in W$.

\end{theorem}

\begin{corollary}
\label{cor-weaklydown-nodelevel}

Let $D=(V,\textit{Ed},r,\lambda)$ be a document, let $v$ be a node of
$D$, and let $W\subseteq V$. Then there exists an 
expression $e$ in the weakly downward (core) XPath algebra such that
$e(D)(v)=W$ if and only if all nodes of~$W$ are descendants of $v$,
and, for all $w_1,w_2\in W$ with
$(v,w_1)\pairs{\sigequiv}{\updownequiv{1}} (v,w_2)$, 
$w_1\in W$ implies $w_2\in W$. 

\end{corollary}

For $v=r$, the condition $(v,w_1)\pairs{\sigequiv}{\updownequiv{k}}
(v,w_2)$ reduces to $w_1\updownequiv{k} w_2$, by
Proposition~\ref{prop-kequivalent} and Lemma~\ref{lem-weaklydown1}. 
Comparing Theorem~\ref{theo-weaklydown-nodelevel} and
Corollary~\ref{cor-weaklydown-nodelevel} with, respectively,
Theorem~\ref{theo-strictlydown-root} and
Corollary~\ref{cor-strictlydown-root} then immediately yields the
following.

\begin{theorem}
\label{theo-weaklydown-root}

Let $D=(V,\textit{Ed},r,\lambda)$.
\begin{enumerate}

\item for each expression $e$ in the weakly downward (core) XPath algebra
  with counting up to~$k$, $k\geq 1$, there exists an expression $e'$ in
  the strictly downward (core) XPath algebra with counting up to~$k$
  such that $e(D)(r)=e'(D)(r)$; in particular,

\item for each expression $e$ in the weakly downward (core) XPath
  algebra, there exists an expression $e'$ in 
  the strictly downward (core) XPath algebra
  such that $e(D)(r)=e'(D)(r)$.

\end{enumerate}
\end{theorem}

Hence, the corresponding weakly downward and strictly downward
languages are navigationally equivalent if navigation always starts
from the root.

%-------------------------------------------------------------------------
\subsection{Weakly downward languages not containing set difference}
\label{subsec-weaklydownnodifference}

To find characterizations for weakly downward languages not containing
set difference, we can proceed in two ways: 
\begin{enumerate}

\item we proceed as in Section~\ref{subsubsec-downwardpositivexpath} for
  strictly downward languages without set difference, i.e., reviewing
  the results in
  Sections~\ref{subsec-weaklydownsufficient}--\ref{subsec-weaklydowncharacterization}
  and examine to which extent these results in the case where $k=1$
  allow replacing 1-equivalence by relatedness
  (Definition~\ref{def-weakequivalent}); or

\item we start from the results in
  Section~\ref{subsubsec-downwardpositivexpath} on strictly downward
  languages without set difference and ``bootstrap'' them to results on
  weakly downward languages without set difference in the same way as
  the results on strictly downward languages with set difference in
  Sections~\ref{subsec-strictlydownsufficient}--\ref{subsec-strictlydowncharacterization}
  were bootstrapped to results on weakly down ward languages with set
  difference in Sections~\ref{subsec-weaklydownsufficient}--\ref{subsec-weaklydowncharacterization}.

\end{enumerate}
Of course, both approaches lead to the same results. As the necessary
intermediate lemmas and all the proofs can readily be deduced in one
of the two ways described above, we limit ourselves to giving the
main results. Only one technical subtlety deserves mentioning here:
despite the absence of difference, both the property that a node is
the root and the property that a node is not the root can be
expressed, the latter using second projection. For more details, we
refer to the proof of Proposition~\ref{prop-weaklydownnecessary}.

Concretely, the language for which we provide characterizations in
this Section, are $\BL(\down,\pi_1,\pi_2,\cap)$ and
$\CL(\down,\pi_1,\pi_2,\cap)$, which , moreover, are equivalent, by
Corollary~\ref{cor-downwardcorecoincides}. 
We call these languages the \emph{weakly 
  downward positive XPath algebra\/} and the
\emph{weakly downward core positive XPath algebra\/},
respectively. In addition, we can eliminate intersection
\emph{altogether\/}. This 
follows from an earlier result by some of the present
authors~\cite{WuGGP11}. Although this result was stated in the context
of languages that allow both downward and upward navigation, a careful
examination of the elimination algorithm reveals that the results
still hold in the absence of upward navigation. Thus, we have the
following.

\begin{proposition}
\label{prop-eliminateintersection}

The weakly downward positive XPath algebra and the
weakly downward core positive XPath algebra are both equivalent to 
$\BL(\down,\pi_1,\pi_2)$.

\end{proposition}

We now summarize the characterization results.

\begin{theorem}
\label{theo-weakweaklydownequivalent}

Consider the weakly downward (core) positive XPath algebra.
Let $D=(V,\textit{Ed},r,\lambda)$ be a document, and let
$v_1$ and $v_2$ be nodes of $D$.  Then $v_1\expequiv v_2$ if and only
if $v_1\weakupdownequiv{} v_2$.

\end{theorem}

\begin{theorem}
\label{theo-weakweaklydown}

Consider the weakly downward (core) positive XPath algebra.
Let
$D=(V,\textit{Ed},r,\lambda)$ be a document, and let $v_1$, $w_1$,
$v_2$, and $w_2$ be nodes of $D$ such that $w_1$ is a descendant of
$v_1$ and $w_2$ is a descendant of $v_2$. 
Then, the property that, for each expression $e$ in the
language under consideration, $(v_1,w_1)\in e(D)$ if and only if
$(v_2,w_2)\in e(D)$ is equivalent to the property
$(v_1,w_1)\pairs{\sigequiv}{\weakupdownequiv{}} (v_2,w_2)$.

\end{theorem}

\begin{theorem}
\label{theo-weakweaklydown-bp}

Let $D=(V,\textit{Ed},r,\lambda)$ be a document, and let
$R\subseteq V\times V$. Then, there exists an expression $e$ in the
weakly downward (core) positive XPath algebra
such that $e(D)=R$ if and only if,
\begin{enumerate}

\item for all $v,w\in V$, $(v,w)\in R$ implies $w$ is a descendant
  of~$v$; and,

\item for all $v_1,w_1,v_2,w_2\in V$ with $w_1$ a descendant of $v_1$,
  $w_2$ a descendant of $v_2$, and
  $(v_1,w_1)\pairs{\sigequiv}{\updowngeq{}} (v_2,w_2)$, 
$(v_1,w_1)\in R$ implies $(v_2,w_2)\in R$. 

\end{enumerate}
\end{theorem}

\begin{corollary}
\label{theo-weakweaklydown-nodelevel}

Let $D=(V,\textit{Ed},r,\lambda)$ be a document, let
$v$ be a node of $D$, and let $W\subseteq V$. Then there exists an
expression $e$ in the weakly downward (core) positive XPath algebra
such that $e(D)(v)=W$ if and only if all nodes
of~$W$ are descendants of $v$, and, for all nodes $w_1$ and $w_2$ of $D$ with
$(v,w_1)\pairs{\sigequiv}{\updowngeq{}} (v,w_2)$, $w_1\in W$ implies
$w_2\in W$.

\end{corollary}

\begin{corollary}
\label{cor-weakweaklydown-root}

Let $D=(V,\textit{Ed},r,\lambda)$ be a document, 
and let $W\subseteq V$. Then there exists an 
expression $e$ in the weakly downward (core) positive XPath algebra such that
$e(D)(r)=W$ if and only if, for all nodes $w_1$ and $w_2$ of $D$ with
$w_1\updowngeq{} w_2$, $w_1\in W$ implies $w_2\in W$. 

\end{corollary}

Hence, the weakly downward positive (core) XPath algebra and the
strictly downward positive (core) XPath algebra are navigationally
equivalent if navigation always starts from the root.

%======================================================================
\section{Upward languages}
\label{sec-upward}

In analogy to downward languages, we call a language \emph{upward\/}
if, for any expression in that
language, and for any node $v$ of the document $D$ under
consideration, all nodes in $e(D)(v)$ are ancestors of $v$.
If in an addition, it is always the case that $e(D)(v)=e(D')$, where
$D'$ is the subtree of $D$ obtained by removing from $D$ all strict
descendants of~$v$, we call the language \emph{strictly
  upward\/}. Upward languages that are \emph{not\/} strictly upward
will be called \emph{weakly upward}.

For $E$ a set of nonbasic
  operations of Table~\ref{tab-binops}, $\BL(E)$ or $\CL(E)$ is
  upward if it does not contain downward navigation (``$\down$''), 
  and inverse
  (``$.^{-1}$''). Additionally, strictly upward languages do
  not contain second projection (``$\pi_2$'') and counting operations
  (``$\ch{k}(.)$'').

Of course, there is a distinct asymmetry between strictly upward
languages and strictly downward languages: while a node can have an
arbitrary number of children, it 
has at most one parent, making the analysis of strictly upward
languages much easier than the analysis of downward languages. We
shall see, however, that this asymmetry disappears for weakly upward
languages versus weakly downward languages.

Finally, we observe that the analogues of
Theorem~\ref{theo-downwardcorecoincides} and
Corollary~\ref{cor-downwardcorecoincides} still hold for upward
languages: set difference (``$-$'') and intersection (``$\cap$'') can be
eliminated, unless they are used as operations in a Boolean
combination of subexpressions of the language within a
projection. Hence, an upward language and its
corresponding core language coincide.
%-------------------------------------------------------------------------
\subsection{Strictly upward languages}
\label{subsec-strictlyupward}

The languages we consider here, are $\BL(\up,\pi_1,-)$ and
$\CL(\up,\pi_1,-)$, which are equivalent, and $\BL(\up,\pi_1,\cap)$ and
$\CL(\up,\pi_1,\cap)$, which are also equivalent. We refer to the
former as the \emph{strictly upward (core) XPath algebra\/} and the
\emph{strictly upward (core) positive XPath algebra\/}, respectively.
As the characterization results for these languages are easy to derive
along the lines set out in Section~\ref{sec-strictlydownward}, we
merely summarize the results.

\begin{theorem}
\label{theo-upequivalent}

Consider the strictly upward (core) (positive) XPath algebra.
Let $D=(V,\textit{Ed},r,\lambda)$ be a document,
and let $v_1$ and $v_2$ be nodes of $D$. Then $v_1\expequiv v_2$, if
and only if $v_1\upequiv v_2$.

\end{theorem}

\begin{theorem}
\label{theo-strictlyup}

Consider the strictly upward (core) (positive) XPath algebra.
Let $D=(V,\textit{Ed},r,\lambda)$ be a document, and
let $v_1$, $w_1$, $v_2$, and $w_2$ be nodes of $D$ such that $w_1$ is
an ancestor of $v_1$ and $w_2$ is an ancestor of $v_2$. 
Then, the property that, for each expression $e$ in the
language under consideration, $(v_1,w_1)\in e(D)$ if and only if
$(v_2,w_2)\in e(D)$ is equivalent to the property
$(v_1,w_1)\pairs{\sigequiv}{\upequiv} (v_2,w_2)$.

\end{theorem}

\begin{theorem}
\label{theo-strictlyup-bp}

Let $D=(V,\textit{Ed},r,\lambda)$ be a document, and
let $R\subseteq V\times V$. Then, there exists an expression $e$ in the
strictly upward (core) XPath algebra 
such that $e(D)=R$
if and only if, 
\begin{enumerate}

\item for all $v,w\in V$, $(v,w)\in R$ implies $w$ is a ancestor
  of~$v$; and,

\item for all $v_1,w_1,v_2,w_2\in V$ with $w_1$ a ancestor of $v_1$,
  $w_2$ a ancestor of $v_2$, and
  $(v_1,w_1)\pairs{\sigequiv}{\upequiv} (v_2,w_2)$, 
$(v_1,w_1)\in R$ implies $(v_2,w_2)\in R$. 

\end{enumerate}
\end{theorem}

\begin{theorem}
\label{theo-weakstrictlyup-bp}

Let $D=(V,\textit{Ed},r,\lambda)$ be a document, and
let $R\subseteq V\times V$. Then, there exists an expression $e$ in the
strictly upward (core) positive XPath algebra 
such that $e(D)=R$
if and only if, 
\begin{enumerate}

\item for all $v,w\in V$, $(v,w)\in R$ implies $w$ is a ancestor
  of~$v$; and,

\item for all $v_1,w_1,v_2,w_2\in V$ with $w_1$ a ancestor of $v_1$,
  $w_2$ a ancestor of $v_2$, and
  $(v_1,w_1)\pairs{\sigequiv}{\upgeq} (v_2,w_2)$, 
$(v_1,w_1)\in R$ implies $(v_2,w_2)\in R$. 

\end{enumerate}
\end{theorem}

The difference between the strictly downward (core) XPath algebra and
the strictly downward (core) positive XPath algebra becomes only
apparent in the BP-characterization: in
Theorem~\ref{theo-strictlyup-bp}, $R$ is a union of equivalence
classes under $\pairs{\sigequiv}{\upequiv}$, whereas in 
Theorem~\ref{theo-weakstrictlyup-bp}, $R$ is merely closed under
the relation $\pairs{\sigequiv}{\upgeq}$.
%-------------------------------------------------------------------------
\subsection{Weakly upward languages}
\label{subsec-weaklyupward}

Weakly upward languages are closely related to weakly downward
languages, by the following result.

\begin{theorem}
\label{theo-weakupisweakdowninverse}

Let $E$ be a set of nonbasic operations not containing
downward navigation (``$\down$''),
and inverse (``$.^{-1}$''). Let $E'$ be the set
of nonbasic operations obtained from $E$ by replacing downward
navigation by upward navigation (``$\up$''), first projection
(``$\pi_1$'') by second projection (``$\pi_2$''), and second
projection by first projection. Then, for each expression $e$ in
$\BL(E)$ (respectively, $\CL(E)$), there is an expression $e'$ in
$\BL(E')$ (respectively, $\CL(E')$) such that $e^{-1}$ and $e'$ are
equivalent at the level of queries, and vice versa.

\end{theorem}

\begin{proof}
Starting from $e^{-1}$, we eliminate inverse (``$.^{-1}$'') using the
identities in the proof of Proposition~\ref{prop-eliminate}, and the
additional identities $\pi_1(e)^{-1}(D)=\pi_2(e^{-1})(D)$
and $\pi_2^{-1}(D)=\pi_1(e^{-1})(D)$, for $D$ an arbitrary
document.\footnote{Observe that $\pi_1(e)^{-1}(D)=\pi_1(e)(D)$ and
$\pi_2(e)^{-1}(D)=\pi_2(e)(D)$ are also valid identities if the sole
  purpose was to eliminate inverse; however, these identities will not
  lead to the desired result.} This elimination process yields the
desired expression $e'$.
\end{proof}

Together with the fact that, in a subsumption or congruence, the order
of the nodes in the pairs on the left- and right-hand sides may be swapped
simultaneously (Proposition~\ref{prop-subsumption},
(\ref{subsumption-2}) and (\ref{subsumption-4})),
Theorem~\ref{theo-weakupisweakdowninverse} has the following immediate
consequences:
\begin{enumerate}

\item Each characterization for a weakly downward language in
  Section~\ref{sec-weaklydownward}---which in each instance contains
  both projections---yields a characterization for the corresponding
  weakly upward language (i.e., obtained by substituting upward
  navigation for downward navigation) by replacing ``descendant'' by
  ``ancestor''; and

\item Each characterization for a strictly downward language in
  Section~\ref{sec-strictlydownward}---which in each instance contains
  the first projection---yields a characterization for the
  corresponding weakly upward language (i.e., obtained by
  substituting upward navigation for downward navigation and second
  for first projection) by replacing ``descendant'' by
  ``ancestor''.

\end{enumerate}
Moreover, Theorem~\ref{theo-weakupisweakdowninverse} gives us for free
characterizations for some additional weakly \emph{downward\/}
languages \emph{not\/} considered in Section~\ref{sec-weaklydownward}:
\begin{enumerate}

\item[~] Each characterization for a strictly upward language in
  Section~\ref{subsec-strictlyupward}---which in each instance contains
  the second projection---yields a characterization for the
  corresponding weakly downward language (i.e., obtained by
  substituting downward navigation for downward navigation and first
  for second projection) by replacing ``ancestor'' by
  ``descendant''.

\end{enumerate}
In view of space considerations, however, we refrain from explicitly writing down these new
characterization results.

%=======================================================================
\section{Languages for two-way navigation}
\label{sec-general}

We finally consider languages which are neither downward nor upward,
i.e., in which navigation in both directions (``$\down$'' and
``$\up$'') is possible. A notable difference in this case is that
standard languages no longer always coincide with their associated core
languages in expressive power. Below we distinguish languages with and
without difference. In the first case, we discuss the standard languages
and the core languages separately 
(Sections~\ref{subsec-generaldiffnocore} and~\ref{subsec-generaldiffcore}).
In the second case, there is no need for this distinction
(Section~\ref{subsec-generalnodiff}).
%-----------------------------------------------------------------------
\subsection{Standard languages with difference for two-way navigation}
\label{subsec-generaldiffnocore}

First, we state analogues to Lemmas~\ref{lem-weaklydown1}
and~\ref{lem-weaklydown2} for pairs
of nodes that are \emph{not\/} necessarily ancestor-descendant pairs.

\begin{lemma}
\label{lem-updown1}

Let $D=(V,\textit{Ed},r,\lambda)$ be a document, let $v_1$, $w_1$,
$v_2$, and $w_2$ be nodes of $D$, and let
$k\geq 1$. Then, $(v_1,w_1)\pairs{\sigequiv}{\updownequiv{k}}
(v_2,w_2)$ if and only if $(v_1,w_1)\sigequiv(v_2,w_2)$, 
$v_1\updownequiv{k} v_2$, and $w_1\updownequiv{k} w_2$. 

\end{lemma}

\begin{proof}
As the ``only if'' is obvious, we focus on the ``if.'' Obviously, 
$(v_1,w_1)\sigequiv(v_2,w_2)$ implies that
$(\top(v_1,w_1),v_1)\sigequiv (\top(v_2,w_2),w_2)$. By
Lemma~\ref{lem-weaklydown1}, \\
$(\top(v_1,w_1),v_1)\pairs{\sigequiv}{\updownequiv{k}}
(\top(v_2,w_2),v_2)$. In the same way, we derive \\
$(\top(v_1,w_1),w_1)\pairs{\sigequiv}{\updownequiv{k}}
(\top(v_2,w_2),w_2)$. Applying Proposition~\ref{prop-subsumption},
(\ref{subsumption-2}), (\ref{subsumption-3}) and
(\ref{subsumption-4}), yields the desired result.
\end{proof}

\begin{lemma}
\label{lem-updown2}

Let $D=(V,\textit{Ed},r,\lambda)$ be a document, let $v_1$ and $w_1$
be nodes of $D$, and let $k\geq 2$. Then,
\begin{enumerate}

\item\label{updown1} for each node $v_2$ of $D$ for which
  $v_1\updownequiv{k} v_2$ there is a node
$w_2$ in $D$ such that $(v_1,w_1)\pairs{\sigequiv}{\updownequiv{k}}
  (v_2,w_2)$; and

\item\label{updown2} for each node $w_2$ of $D$ for which
  $w_1\updownequiv{k} w_2$ there is a node
$v_2$ in $D$ such that $(v_1,w_1)\pairs{\sigequiv}{\updownequiv{k}}
  (v_2,w_2)$.

\end{enumerate}
\end{lemma}

\begin{proof}
We only prove (\ref{updown1}); the proof of (\ref{updown2}) is
completely analogous. By Lemma \ref{lem-weaklydown2},
(\ref{weakdown2}), there exists a node $t_2$ in $D$ such that
$(\top(v_1,w_1),v_1) \pairs{\sigequiv}{\updownequiv{k}} (t_2,v_2)$,
and, hence, also that
$(v_1,\top(v_1,w_1))\pairs{\sigequiv}{\updownequiv{k}} (v_2,t_2)$, by
Proposition~\ref{prop-subsumption}, (\ref{subsumption-2}) and
(\ref{subsumption-4}). Let $y_1$ be the child of $\top(v_1,w_1)$ on
the path to $w_1$. Since $k\geq 2$, there is a child $y_2$ of $t_2$
such that (1) $y_1\updownequiv{k} y_2$ and (2) $y_2$ is not on the
path from $t_2$ to $v_2$.\footnote{To see the latter claim, observe
  that $t_2$ must have two different $k$-equivalent children when
  $\top(v_1,w_1)$ has.} By Lemma~\ref{lem-weaklydown2},
(\ref{weakdown2}), there exists a node $w_2$ in $D$ such that
$(y_1,w_1)\pairs{\sigequiv}{\updownequiv{k}} (y_2,w_2)$. In
particular, $w_1\updownequiv{k} w_2$. By construction,
  $t_2=\top(v_2,w_2)$, and, hence, $(v_1,w_1)\sigequiv (v_2,w_2)$. The
  result now follows from Lemma~\ref{lem-updown1}.
\end{proof}

The mutual position of the nodes in the statement and the proof of
Lemma~\ref{lem-updown2}, (\ref{updown1}),
 is illustrated in Figure~\ref{fig-updown2}.

\begin{figure}[!htb]
\begin{center}
\resizebox{0.9\textwidth}{!}{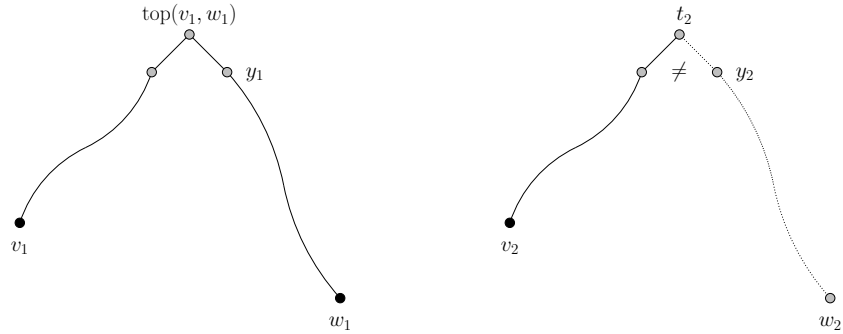}
\end{center}
\caption{Mutual position of the nodes in the statement and the proof of
Lemma~\ref{lem-updown2}, (\ref{updown1}).}%
\label{fig-updown2}
\end{figure}

Before we can start with establishing the relationship between
$\updownequiv{k}$-congruence and expression equivalence under
languages allowing two-way navigation, we need one more lemma to be
able to deal with the composition operator.

\begin{lemma}
\label{lem-updown3}

Let $D=(V,\textit{Ed},r,\lambda)$ be a document, let $v_1$, $w_1$,
$v_2$, and $w_2$ be nodes of $D$ such that
$(v_1,w_1)\pairs{\sigequiv}{\updownequiv{k}}(v_2,w_2)$, and let $k\geq
3$. Then, for every node $y_1$ of $D$, there exists a node $y_2$ of
$D$ such that $(v_1,y_1)\pairs{\sigequiv}{\updownequiv{k}}(v_2,y_2)$, and
$(y_1,w_1)\pairs{\sigequiv}{\updownequiv{k}}(y_2,w_2)$.

\end{lemma}

\begin{proof}
The proof is essentially a case analysis. In each case description, we
assume implicitly that the cases that were already dealt with before
are excluded.
\begin{enumerate}

\item \emph{$y_1$ is on the path from $v_1$ to $w_1$\/}. In that case,
  let $y_2$ be the node corresponding to $y_1$ on the path from $v_2$
  to $w_2$. The result now follows immediately.

\item \emph{$y_1$ is a strict descendant of $v_1$\/}. By
  Lemma~\ref{lem-weaklydown2}, (\ref{weakdown1}), there is a
  (strict) descendant $y_2$ of $v_2$ such that
  $(v_1,y_1)\pairs{\sigequiv}{\updownequiv{k}}(v_2,y_2)$. The result now
  follows immediately.

\item \emph{$y_1$ is a strict descendant of $w_1$\/}. Analogous to the
  previous case.

\item \emph{$y_1$ is a strict ancestor of $\top(v_1,w_1)$\/}. By
  Lemma~\ref{lem-weaklydown2}, (\ref{weakdown2}), there is a
  (strict) ancestor $y_2$ of $\top(v_2,w_2)$ such that
  $(\top(v_1,w_1),y_1)\pairs{\sigequiv}{\updownequiv{k}}(\top(v_2,w_2),y_2)$. The
  result now follows immediately.

\item\label{updown-case5} \emph{$\top(v_1,y_1)$ is an internal node on
  the path from $v_1$ to $\top(v_1,w_1)$\/}. By
  Lemma~\ref{lem-updown2}, (\ref{updown1}), there exists a node $y_2$
  in $D$ such that
  $(v_1,y_1)\pairs{\sigequiv}{\updownequiv{k}}(v_2,y_2)$. Since, in
  this case, $\top(y_1,w_1)=\top(v_1,w_1)$, and, therefore, an
  ancestor of $v_1$, we may apply Proposition~\ref{prop-subsumption},
  (\ref{subsumption-2})--(\ref{subsumption-4}), to obtain that
  $(y_1,w_1)\sigequiv (y_2,w_2)$. Since, moreover, $y_1\updownequiv{k}
  y_2$ and $w_1\updownequiv{k} w_2$, the desired result now follows from
  Lemma~\ref{lem-updown1}.

Figure~\ref{fig-updown3-1} illustrates this case and the constructions
therein.

\begin{figure}[!htb]
\begin{center}
\resizebox{0.85\textwidth}{!}{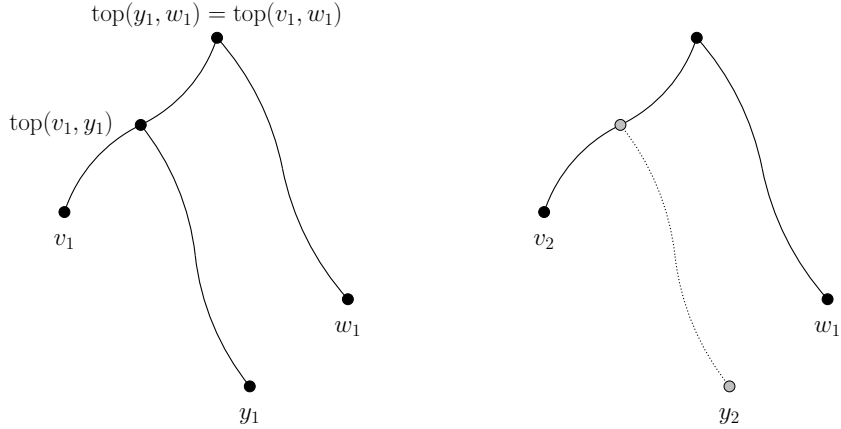}
\end{center}
\caption{Mutual position of the nodes in Case~\ref{updown-case5} of
  the proof of Lemma~\ref{lem-updown3}.}%
\label{fig-updown3-1}
\end{figure}

\item\label{updown-case6} \emph{$\top(y_1,w_1)$ is an internal node on
  the path from $\top(v_1,w_1)$ to $w_1$\/}. Analogous to the
  previous case.

\item\label{updown-case7} \emph{$\top(v_1,y_1)=\top(y_1,w_1)$ is a strict
ancestor of $\top(v_1,w_1)$\/}.  By
  Lemma~\ref{lem-updown2}, (\ref{updown1}), there exists a node $y_2$
  in $D$ such that
  $(v_1,y_1)\pairs{\sigequiv}{\updownequiv{k}}(v_2,y_2)$. Since, in
  this case, $\top(y_1,w_1)$ is a (strict) ancestor of
  $\top(v_1,w_1)$, and, therefore, an ancestor of $v_1$, we may apply
  Proposition~\ref{prop-subsumption},
  (\ref{subsumption-2})--(\ref{subsumption-4}), to obtain that
  $(y_1,w_1)\sigequiv (y_2,w_2)$. Since, moreover, $y_1\updownequiv{k} 
  y_2$ and $w_1\updownequiv{k} w_2$, the desired result now follows from
  Lemma~\ref{lem-updown1}.

Figure~\ref{fig-updown3-2} illustrates this case and the constructions
therein.

\begin{figure}[!htb]
\begin{center}
\resizebox{0.85\textwidth}{!}{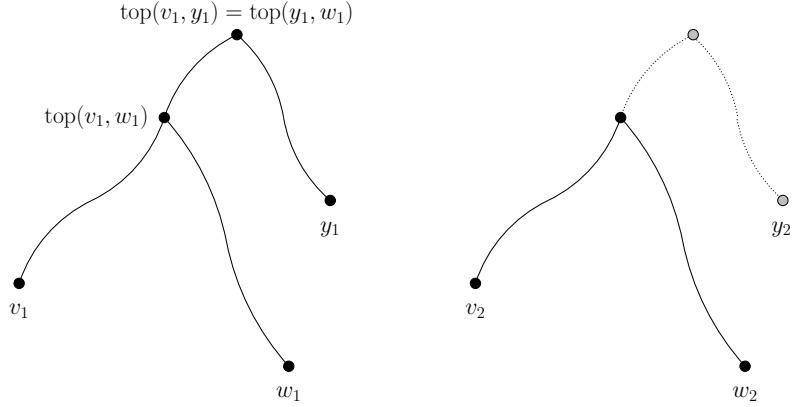}
\end{center}
\caption{Mutual position of the nodes in Case~\ref{updown-case7} of
  the proof of Lemma~\ref{lem-updown3}.}%
\label{fig-updown3-2}
\end{figure}

\item\label{updown-case8}
  \emph{$\top(v_1,y_1)=\top(y_1,w_1)=\top(v_1,w_1)$\/}. Let $z_1$ be
  the child of the top node on the path to $y_1$. By assumption, $z_1$
  is not on the path from $v_1$ to $w_1$. Since $k\geq 3$, there is a
  node $z_2$ in $D$ not on the path from $v_2$ to $w_2$ such that
  $z_1\updownequiv{k} z_2$. (For example, in the subcase where the
  children of $\top(v_1,w_1)$ on the paths to $v_1$, $w_1$, and $y_1$,
  the last of which is $z_1$, are all three $k$-equivalent, we know
  that $\top(v_2,w_2)$ must also have at least three children that are
  $k$-equivalent to $z_1$. Hence, at least one of these is not on the
  path from $v_1$ to $w_1$.) By   Lemma~\ref{lem-updown2},
  (\ref{updown1}), there exists a node $y_2$ in $D$ such that
  $(z_1,y_1)\pairs{\sigequiv}{\updownequiv{k}}(z_2,y_2)$. The result
  now follows readily.

Figure~\ref{fig-updown3-3} illustrates this case and the constructions
therein.

\begin{figure}[!htb]
\begin{center}
\resizebox{0.85\textwidth}{!}{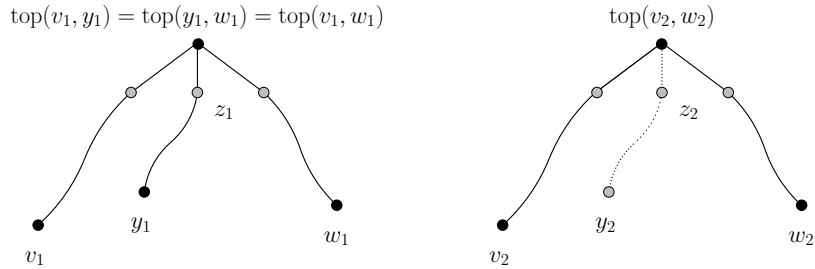}
\end{center}
\caption{Mutual position of the nodes in Case~\ref{updown-case8} of
  the proof of Lemma~\ref{lem-updown3}.}%
\label{fig-updown3-3}
\end{figure}

\end{enumerate}
\end{proof}

We are now ready to state the analogue of
Proposition~\ref{prop-weaklydown} for languages with two-way navigation.

\begin{proposition}
\label{prop-updown}

Let $k\geq 3$, and let $E$ be the set of all nonbasic operations in
Table~\ref{tab-binops}, except for selection on at least $m$ children 
satisfying some condition (``$\ch{m}(.)$'') for $m>k$. Let $e$ be an
expression in $\BL(E)$. 
Let $D=(V,\textit{Ed},r,\lambda)$ be a
document, and let $v_1$, $w_1$, $v_2$, 
and $w_2$ be nodes of $D$.  Assume furthermore that
$(v_1,w_1)\pairs{\sigequiv}{\updownequiv{k}} (v_2,w_2)$. Then,
$(v_1,w_1)\in e(D)$ if and only if $(v_2,w_2)\in e(D)$.

\end{proposition}

\begin{proof}
The proof goes along the same lines as the proofs of
Propositions~\ref{prop-weaklydown} and~\ref{prop-strictlydown}. 
The base case, for the atomic operators, remains straightforward.
In the induction step, we must now rely on Lemma~\ref{lem-updown3} to
make the case for composition (``$./.$''). To make the case for first,
respectively, second projection (``$\pi_1$,'' respectively,
``$\pi_2$''), we must rely on Lemma~\ref{lem-updown2}, 
(\ref{updown1}), respectively, (\ref{updown2}). The arguments for the
counting operations (``$\ch{m}(.)$,'' $m\leq k$), union (''$\cup$''),
intersection (''$\cap$''), and set difference (``$-$'') in the proof
of Proposition~\ref{prop-strictlydown} carry over to the present
setting. Finally, the case for inverse (``$.^{-1}$'') is straightforward.
\end{proof}

As in Section~\ref{subsec-weaklydownnecessary}, we can in two steps infer the
following result from Proposition~\ref{prop-updown}.

\begin{corollary}
\label{cor-updown}

Let $k\geq 3$, and let $E$ be a set of nonbasic operations
not containing selection on at least $m$ children
satisfying some condition (``$\ch{m}(.)$'') for $m>k$. 
Consider the language $\BL(E)$. Let
$D=(V,\textit{Ed},r,\lambda)$ be a document, and let $v_1$ and $v_2$
be nodes of $D$. If $v_1\updownequiv{k} v_2$, then $v_1\expequiv v_2$.

\end{corollary}

Now, notice that Proposition~\ref{prop-weaklydownnecessary} of
Section~\ref{subsec-weaklydownnecessary} is also
applicable to an important class of languages allowing two-way navigation.
The standard language for two-way navigation satisfying both
Corollary~\ref{cor-updown} and
Proposition~\ref{prop-weaklydownnecessary} is
$\BL(\down,\up,\ch{1}(.),\ldots,\ch{k}(.),-)$.\footnote{All other operations
  are redundant, by Proposition~\ref{prop-eliminate}.} We call this
language the \emph{XPath algebra with counting up to~$k$\/}. Combining the
aforementioned results, we obtain the following.

\begin{theorem}
\label{theo-updown}

Let $k\ge 3$, and consider the XPath algebra with counting up to~$k$.
Let $D=(V,\textit{Ed},r,\lambda)$ be a document, and let $v_1$ and $v_2$
be nodes of $D$. Then, $v_1\expequiv v_2$ if and only if
$v_1\updownequiv{k} v_2$.

\end{theorem}

By Proposition~\ref{prop-counting}, selection on up to three children
satisfying some condition (``$\ch{m}(.)$,'' $1\leq m\leq 3$) can be
expressed in the XPath algebra. Hence, a special case arises for
$k=3$:

\begin{corollary}
\label{cor-updownxpath}

Consider the XPath algebra. Let
$D=(V,\textit{Ed},r,\lambda)$ be a document, and let $v_1$ and $v_2$
be nodes of $D$. Then, $v_1\expequiv v_2$ if and only if
$v_1\updownequiv{3} v_2$.

\end{corollary}

We next prove a converse to Proposition~\ref{prop-updown}.

\begin{proposition}
\label{prop-updownconverse}

Let $k\geq 3$, and consider the XPath algebra with counting up to~$k$.
Let $D=(V,\textit{Ed},r,\lambda)$ be a document, and let $v_1$, $w_1$,
$v_2$, and $w_2$ be nodes of $D$. Assume furthermore that,
for each expression~$e$ in the language, $(v_1,w_1)\in e(D)$ if and
only if $(v_2,w_2)\in e(D)$. Then
$(v_1,w_1)\pairs{\sigequiv}{\updownequiv{k}}(v_2,w_2)$.

\end{proposition}

\begin{proof}
Since $\sig(v_1,w_1)$ is an expression in the language under
consideration, and since $(v_1,w_1)\in\sig(v_1,w_1)$,
$(v_2,w_2)\in\sig(v_1,w_1)$. Similarly, 
$(v_1,w_1)\in\sig(v_2,w_2)$. We may thus conclude that
$(v_1,w_1)\sigequiv (v_2,w_2)$. Now, let $f$ be any expression in the
language such that $f(D)(v_1)\neq\emptyset$. Then, $(v_1,v_1)\in
\pi_1(f)(D)$. Let $e\ass \pi_1(f)/\sig(v_1,w_1)$. By construction,
$(v_1,w_1)\in e(D)$. Hence, by assumption, $(v_2,w_2)\in e(D)$, which
implies $(v_2,v_2)\in \pi_1(f)(D)$ or $f(D)(v_2)\neq\emptyset$. The same
holds vice versa, and we may thus conclude that $v_1\expequiv v_2$, and,
hence, by Theorem~\ref{theo-updown}, $v_1\updownequiv{k} v_2$.
In a similar way, we prove that $w_1\updownequiv{k} w_2$.
By Lemma~\ref{lem-updown1}, we may now conclude that
$(v_1,w_1)\pairs{\sigequiv}{\updownequiv{k}} (v_2,w_2)$.
\end{proof}

Combining Propositions~\ref{prop-updown}
and~\ref{prop-updownconverse}, we obtain the following
characterization. 

\begin{corollary}
\label{cor-updownpairs}

Let $k\geq 3$, and consider the XPath algebra with
counting up to~$k$. Let
$D=(V,\textit{Ed},r,\lambda)$ be a document, and let $v_1$, $w_1$,
$v_2$, and $w_2$ be nodes of $D$.
Then, the property that, for each expression $e$ in the
language under consideration, $(v_1,w_1)\in e(D)$ if and only if
$(v_2,w_2)\in e(D)$ is equivalent to the property
$(v_1,w_1)\pairs{\sigequiv}{\updownequiv{k}} (v_2,w_2)$.

\end{corollary}

Using Theorem~\ref{theo-updown} instead of
Theorem~\ref{theo-downequivalent}, we can recast the proof of
Lemma~\ref{lem-separationdown1} into a proof of

\begin{lemma}
\label{lem-separationupdown1}

Let $k\geq 3$. Let $D=(V,\textit{Ed},r,\lambda)$ be a document, and
let $v_1$ be a node of $D$.  There exists an expression $e_{v_1}$ in
the XPath algebra with counting up to~$k$ such that, for
each node $v_2$ of $D$, $e_{v_1}(D)(v_2)\neq\emptyset$ if and only if
$v_1\updownequiv{k} v_2$. 

\end{lemma}

We can now bootstrap Lemma~\ref{lem-separationupdown1} to the
following result.

\begin{lemma}
\label{lem-separationupdown2}

Let $k\geq 3$.  Let $D=(V,\textit{Ed},r,\lambda)$ be a document, and
let $v_1$ and $w_1$ be a nodes of $D$.
There exists an expression $e_{v_1,w_1}$ in the 
XPath algebra with counting up to~$k$ such that, for all nodes
$v_2$ and $w_2$ of $D$, $(v_2,w_2)\in e_{v_1,w_1}(D)$ if and only if
$(v_1,w_1)\pairs{\sigequiv}{\updownequiv{k}} (v_2,w_2)$.

\end{lemma}

\begin{proof}
From Lemma~\ref{lem-separationupdown1}, we know that, for node $y_1$
of $D$, there exists an expression $e_{y_1}$ in the XPath algebra with
counting up to~$k$ such that, for each node $y_2$ of $D$,
$e_{y_1}(D)(y_2)\neq\emptyset$ if and only if $y_1\updownequiv{k}
y_2$. Now, let $v_1$ and $w_1$ be nodes of $D$.
Let $\sig(v_1,w_1)=\up^u/\down^d$, with $u,d\geq 0$, and
define
$$e_{v_1,w_1}\ass
\pi_1(e_{v_1})/\sig(v_1,w_1)/\pi_1(e_{w_1})-\up^{u-1}/\down^{d-1},$$ 
where, for an expression $f$, we define $f^{-1}\ass\emptyset$.
Clearly, $e_{v_1,w_1}$ is also in the XPath algebra with counting up to~$k$.
Let $v_2$ and $w_2$ be nodes of $D$. Suppose $(v_2,w_2)\in
e_{v_1,w_1}(D)$. Then, by Proposition~\ref{prop-congruence}, 
$\sig(v_1,w_1)=\sig(v_2,w_2)$. Furthermore, it follows that
$(v_2,v_2)\in e_{v_1}(D)$ and $(w_2,w_2)\in e_{w_1}(D)$. By
Lemma~\ref{lem-separationupdown1}, $v_1\updownequiv{k} v_2$ and
$w_1\updownequiv{k} w_2$. It now follows from Lemma~\ref{lem-updown1}
that $(v_1,w_2)\pairs{\sigequiv}{\updownequiv{k}}(v_2,w_2)$.
As $(v_1,w_1)\in e_{v_1,w_1}(D)$, the converse follows from 
Corollary~\ref{cor-updownpairs}.
\end{proof}

The BP characterization results now follow readily.

\begin{theorem}
\label{theo-updown-bp}

Let $k\geq 3$. Let $D=(V,\textit{Ed},r,\lambda)$ be a document, and let
$R\subseteq V\times V$. Then, there exists an expression $e$ in the
XPath algebra with counting up to $k$ such that $e(D)=R$ if and only if,
for all $v_1,w_1,v_2,w_2\in V$ with
  $(v_1,w_1)\pairs{\sigequiv}{\updownequiv{k}} (v_2,w_2)$, 
$(v_1,w_1)\in R$ implies $(v_2,w_2)\in R$. 

\end{theorem}

The specialization to the XPath algebra is as
follows.

\begin{corollary}
\label{cor-updown-bp}

Let $D=(V,\textit{Ed},r,\lambda)$ be a document, and let $R\subseteq
V\times V$. There exists an expression $e$ in the XPath algebra such
that $e(D)=R$ if and only if, 
for all $v_1,w_1,v_2,w_2\in V$ with
  $(v_1,w_1)\pairs{\sigequiv}{\updownequiv{3}} (v_2,w_2)$, 
  $(v_1,w_1)\in R$ implies $(v_2,w_2)\in R$. 

\end{corollary}

We recast Theorem~\ref{theo-updown-bp} and
Corollary~\ref{cor-updown-bp} in terms of
node-level navigation.

\begin{theorem}
\label{theo-updown-nodelevel}

Let $k\geq 3$.  Let $D=(V,\textit{Ed},r,\lambda)$ be a document, let
$v$ be a node of $D$, and let $W\subseteq V$. Then there exists an
expression $e$ in the XPath algebra
with counting up to~$k$ such that $e(D)(v)=W$ if and only if,
for all $w_1,w_2\in W$ with $(v,w_1)\pairs{\sigequiv}{\updownequiv{k}}
(v,w_2)$, $w_1\in W$ implies $w_2\in W$.

\end{theorem}

The specialization to the XPath algebra is as
follows.

\begin{corollary}
\label{cor-updown-nodelevel}

Let $D=(V,\textit{Ed},r,\lambda)$ be a document, let
$v$ be a node of $D$, and let $W\subseteq V$. Then there exists an
expression $e$ in the XPath algebra
such that $e(D)(v)=W$ if and only if,
for all $w_1,w_2\in W$ with $(v,w_1)\pairs{\sigequiv}{\updownequiv{3}}
(v,w_2)$, $w_1\in W$ implies $w_2\in W$.

\end{corollary}

Finally, we consider the special case where navigation starts from the
root. For $v=r$, the condition $(v,w_1)\pairs{\sigequiv}{\updownequiv{k}}
(v,w_2)$ reduces to $w_1\updownequiv{k} w_2$, by
Proposition~\ref{prop-kequivalent} and Lemma~\ref{lem-weaklydown1}. 
Comparing Theorem~\ref{theo-updown-nodelevel} and
Corollary~\ref{cor-updown-nodelevel} with, respectively,
Theorem~\ref{theo-strictlydown-root} and
Corollary~\ref{cor-strictlydown-root} then immediately yields the
following.

\begin{theorem}
\label{theo-updown-root}

Let $D=(V,\textit{Ed},r,\lambda)$.
\begin{enumerate}

\item for each expression $e$ in the XPath algebra
  with counting up to~$k$, $k\geq 3$, there exists an expression $e'$ in
  the strictly downward (core) XPath algebra with counting up to~$k$
  such that $e(D)(r)=e'(D)(r)$.

\item for each expression $e$ in the XPath
  algebra, there exists an expression $e'$ in 
  the strictly downward (core) XPath algebra with counting up to~3
  such that $e(D)(r)=e'(D)(r)$.

\end{enumerate}
\end{theorem}

Theorem~\ref{theo-updown-root} extends
Theorem~\ref{theo-weaklydown-root}. When navigating from the root, the
only thing that the full XPath algebra adds compared to using the
strictly downward (core) XPath algebra is its ability to select on at
least~2 and on at least~3 children satisfying some condition.
%-----------------------------------------------------------------------
\subsection{Core languages with difference for two-way navigation}
\label{subsec-generaldiffcore}

We now investigate what changes if we replace a standard language with
difference for two-way navigation by the corresponding core
language. The most important observation is that both languages are
\emph{not\/} equivalent, unlike in the cases of downward or upward
navigation.

\begin{figure}
\begin{center}
\resizebox{0.5\textwidth}{!}{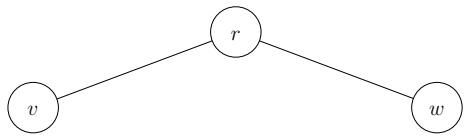}
\end{center}
\caption[a]{Document of Example~\ref{ex-notequivalent}.}
\label{fig-bp-notequivalent}
\end{figure}

\begin{example}
\label{ex-notequivalent}

Let $D=(V,\textit{Ed},r,\lambda)$ be the
 very simple document in
Figure~\ref{fig-bp-notequivalent}. For every value of $k\geq
2,$\footnote{We will not consider $k=1$, because both $\ch{1}(.)$ and
  $\ch{2}(.)$ can be expressed in the core XPath algebra, by
  Proposition~\ref{prop-counting}.} $e\ass\up/\down-\varepsilon$ is an
expression in the XPath algebra with counting up to~$k$. We have that
$e(D)=\{(v,w),(w,v)\}$. From Proposition~\ref{prop-updowncore}, it will
follow, however, that, for every expression $e'$ in the corresponding
core language, $(v,w)\in e'(D)$ implies that not only $(w,v)\in e'(D)$,
but also $(v,v)\in e'(D)$ and $(w,w)\in e'(D)$.

\end{example}

We now explore which changes occur when we try to make the same
reasoning as in Section~\ref{subsec-generaldiffnocore}. 

As Example~\ref{ex-notequivalent} suggests, there is no hope that we
can express congruence in the core XPath algebra with counting up
to~$k$,\footnote{This is the name we give to the core language
  corresponding to the (standard) XPath algebra with counting up
  to~$k$} for any $k\geq 2$. Therefore we shall have to work with
subsumption instead of congruence.

Lemma~\ref{lem-updown1} still holds if we replace congruence by
subsumption. We may of course still use Lemma~\ref{lem-updown2}
(as replacing congruence by subsumption here would 
yield a weaker statement). Lemma~\ref{lem-updown3} also survives
replacing congruence by subsumption, except that we can then strengthen
its statement, as follows.

\begin{lemma}
\label{lem-updowncore3}

Let $D=(V,\textit{Ed},r,\lambda)$ be a document, let $v_1$, $w_1$,
$v_2$, and $w_2$ be nodes of $D$ such that
$(v_1,w_1)\pairs{\siggeq}{\updownequiv{k}}(v_2,w_2)$, and let $k\geq
2$. Then, for every node $y_1$ of $D$, there exists a node $y_2$ of
$D$ such that $(v_1,y_1)\pairs{\siggeq}{\updownequiv{k}}(v_2,y_2)$, and
$(y_1,w_1)\pairs{\siggeq}{\updownequiv{k}}(y_2,w_2)$.

\end{lemma}

\begin{proof}
The only case in the proof of Lemma~\ref{lem-updown3} where we used
$k\geq 3$ is Case~\ref{updown-case8}
($\top(v_1,y_1)=\top(y_1,w_1)=\top(v_1,w_1)$) to guarantee that the
path from $\top(v_2,w_2)$ to $y_2$ has no overlap with both the path
from $\top(v_2,w_2)$ to $v_2$ and the path 
from $\top(v_2,w_2)$ to $w_2$. As this is no concern anymore when
we consider subsumption rather than congruence, the condition $k\geq
2$ suffices to recast the proof of Lemma~\ref{lem-updown3} into a
proof of Lemma~\ref{lem-updowncore3}.
\end{proof}

Lemma~\ref{lem-updown3} was used to complete the induction step for
composition (``$/$'') in the proof of Proposition~\ref{prop-updown}. If we
replace Lemma~\ref{lem-updown3} by Lemma~\ref{lem-updowncore3}, we can also
avoid making the assumption $\geq 3$ here. Thanks to the restricted use of
difference in core languages, we can also get away with subsumption
instead of congruence.

\begin{proposition}
\label{prop-updowncore}

Let $k\geq 2$, and let $E$ be the set of all nonbasic operations in
Table~\ref{tab-binops}, except for selection on at least $m$ children 
satisfying some condition (``$\ch{m}(.)$'') for $m>k$. Let $e$ be an
expression in $\CL(E)$. 
Let $D=(V,\textit{Ed},r,\lambda)$ be a
document, and let $v_1$, $w_1$, $v_2$, 
and $w_2$ be nodes of $D$.  Assume furthermore that
$(v_1,w_1)\pairs{\siggeq}{\updownequiv{k}} (v_2,w_2)$. Then,
$(v_1,w_1)\in e(D)$ implies $(v_2,w_2)\in e(D)$.

\end{proposition}

\begin{proof}
The proof goes along the same lines as the proof of
Proposition~\ref{prop-updown}, except that, in the induction step, we
need not consider the case of set difference (``$-$''). However, we
must consider instead the case where the expression is of the form
$e\ass\pi_1(f)$ or $e\ass\pi_2(f)$ with $f$ a Boolean combination of
expressions of $\CL(E)$ satisfying the induction hypothesis. For
reasons of symmetry, we only consider the case $e\ass\pi_1(f)$.
Without loss of generality, we may assume that $f$ is
union-free. Indeed, we can always rewrite $f$ in disjunctive normal
form, and, for $f=f_1\cup f_2$, $\pi_1(f)=\pi_1(f_1)\cup\pi_1(f_2)$. If,
for an expression $g$ in $\CL(E)$, we define $\overline{g}$ by
$\overline{g}(D)\ass V\times V-g(D)$, we can write $f=f_1\cap\ldots
f_p\cap \overline{g_1}\cap\ldots \overline{g_q}$ for some $p\geq 1$
and $q\geq 0$, with $f_1,\ldots,f_p,g_1,\ldots,g_q$ in $\CL(E)$ and
satisfying the induction hypothesis. In particular, if
$(v_1,v_1)\in\pi_1(f)(D)$,\footnote{In this case, $v_1=w_1$ and
  $v_2=w_2$.} there exists a node $y_1$ in $D$ such that 
$(v_1,y_1)\in f_1(D),\ldots,(v_1,y_1)\in f_p(D)$ and
$(v_1,y_1)\notin g_1(D),\ldots,(v_1,y_1)\notin g_p(D)$.
By Lemma~\ref{lem-updown2}, there exists a node $y_2$ in $D$ such that
$(v_1,y_1)\pairs{\sigequiv}{\updownequiv{k}}(v_2,y_2)$. 
Hence, $(v_1,y_1)\pairs{\siggeq}{\updownequiv{k}}(v_2,y_2)$ \textit{and\/}
$(v_2,y_2)\pairs{\siggeq}{\updownequiv{k}}(v_1,y_1)$.
By the induction hypothesis, 
$(v_2,y_2)\in f_1(D),\ldots,(v_2,y_2)\in f_p(D)$. Now, assume that, for
some $j$, $1\leq j\leq q$, $(v_2,y_2)\in g_j(D)$. Then, again by
the induction hypothesis, $(v_1,y_1)\in g_j(D)$, a contradiction. Hence, 
$(v_2,y_2)\notin g_1(D),\ldots,(v_2,y_2)\notin g_p(D)$. We may thus
conclude that $(v_2,v_2)\in \pi_1(f)$.
\end{proof}

By applying Proposition~\ref{prop-updowncore} twice, we obtain the following.

\begin{corollary}
\label{cor-updowncore1}

Let $k\geq 2$, and let $E$ be the set of all nonbasic operations in
Table~\ref{tab-binops}, except for selection on at least $m$ children 
satisfying some condition (``$\ch{m}(.)$'') for $m>k$. Let $e$ be an
expression in $\CL(E)$. 
Let $D=(V,\textit{Ed},r,\lambda)$ be a
document, and let $v_1$, $w_1$, $v_2$, 
and $w_2$ be nodes of $D$.  Assume furthermore that
$(v_1,w_1)\pairs{\sigequiv}{\updownequiv{k}} (v_2,w_2)$. Then,
$(v_1,w_1)\in e(D)$ if and only if $(v_2,w_2)\in e(D)$.

\end{corollary}

As in Section~\ref{subsec-weaklydownnecessary}, we can in two steps infer the
following result from Corollary~\ref{cor-updowncore1}.

\begin{corollary}
\label{cor-updowncore2}

Let $k\geq 2$, and let $E$ be a set of nonbasic operations
not containing selection on at least $m$ children
satisfying some condition (``$\ch{m}(.)$'') for $m>k$. 
Consider the language $\CL(E)$. Let
$D=(V,\textit{Ed},r,\lambda)$ be a document, and let $v_1$ and $v_2$
be nodes of $D$. If $v_1\updownequiv{k} v_2$, then $v_1\expequiv v_2$.

\end{corollary}

Notice that, for $k\geq 3$, Corollary~\ref{cor-updowncore2} is also an
immediate consequence of Corollary~\ref{cor-updown}. Because we are
dealing with a weaker language, we can also include the case $k=2$, however.

We already observed that Proposition~\ref{prop-weaklydownnecessary} of
Section~\ref{subsec-weaklydownnecessary} is also
applicable to an important class of languages allowing two-way navigation.
The core language for two-way navigation satisfying both
Corollary~\ref{cor-updowncore2} and
Proposition~\ref{prop-weaklydownnecessary} is
$\CL(\down,\up,\pi_1,\pi_2,\ch{1}(.),\ldots,\ch{k}(.),-)$.\footnote{Inverse
  (``$\mathstrut^{-1}$'') is redundant, by the identities in the proof
  of Proposition~\ref{prop-eliminate}, complemented by
  $\pi_1(e)^{-1}(D)=\pi_1(e)(D)$ and $\pi_2(e)^{-1}(D)=\pi_2(e)(D)$.}
We call this language the \emph{core XPath algebra with counting up
  to~$k$\/}. Combining the aforementioned results, we obtain the following.

\begin{theorem}
\label{theo-updowncore}

Let $k\geq 2$, and consider the core XPath algebra with counting up to~$k$.
Let $D=(V,\textit{Ed},r,\lambda)$ be a document, and let $v_1$ and $v_2$
be nodes of $D$. Then, $v_1\expequiv v_2$ if and only if
$v_1\updownequiv{k} v_2$.

\end{theorem}

By Proposition~\ref{prop-counting}, selection on up to two children
satisfying some condition (``$\ch{m}(.)$,'' $1\leq m\leq 2$) can be
expressed in the core XPath algebra. Hence, a special case arises for
$k=2$:

\begin{corollary}
\label{cor-updownxpathcore}

Consider the core XPath algebra. Let
$D=(V,\textit{Ed},r,\lambda)$ be a document, and let $v_1$ and $v_2$
be nodes of $D$. Then, $v_1\expequiv v_2$ if and only if
$v_1\updownequiv{2} v_2$.

\end{corollary}

The proof of Proposition~\ref{prop-updownconverse} can be recast to
a proof of the following converse to
Proposition~\ref{prop-updowncore}.

\begin{proposition}
\label{prop-updownconversecore}

Let $k\geq 2$, and consider the core XPath algebra with counting up to~$k$.
Let $D=(V,\textit{Ed},r,\lambda)$ be a document, and let $v_1$, $w_1$,
$v_2$, and $w_2$ be nodes of $D$. Assume furthermore that,
for each expression~$e$ in the language, $(v_1,w_1)\in e(D)$ implies
$(v_2,w_2)\in e(D)$. Then $(v_1,w_1)\pairs{\siggeq}{\updownequiv{k}}(v_2,w_2)$.

\end{proposition}

Combining Propositions~\ref{prop-updowncore}
and~\ref{prop-updownconversecore}, we obtain the following
characterization. 

\begin{corollary}
\label{cor-updownpairscore}

Let $k\geq 2$, and consider the core XPath algebra with
counting up to~$k$. Let
$D=(V,\textit{Ed},r,\lambda)$ be a document, and let $v_1$, $w_1$,
$v_2$, and $w_2$ be nodes of $D$.
Then, 
\begin{enumerate}

    \item the property that, for each expression $e$ in the language under
        consideration, $(v_1,w_1)\in e(D)$ implies $(v_2,w_2)\in e(D)$ is
        equivalent to the property $(v_1,w_1)\pairs{\siggeq}{\updownequiv{k}}
        (v_2,w_2)$; and,

    \item the property that, for each expression $e$ in the language under
        consideration, $(v_1,w_1)\in e(D)$ if and only if $(v_2,w_2)\in e(D)$ is
        equivalent to the property $(v_1,w_1)\pairs{\sigequiv}{\updownequiv{k}}
        (v_2,w_2)$.

\end{enumerate}
\end{corollary}

%\begin{corollary}
%\label{cor-updownpairscore2}
%
%Let $k\geq 2$, and consider the XPath algebra with
%counting up to~$k$. Let
%$D=(V,\textit{Ed},r,\lambda)$ be a document, and let $v_1$, $w_1$,
%$v_2$, and $w_2$ be nodes of $D$.
%Then, the property that, for each expression $e$ in the
%language under consideration, $(v_1,w_1)\in e(D)$ if and only if
%$(v_2,w_2)\in e(D)$ is equivalent to the property
%$(v_1,w_1)\pairs{\sigequiv}{\updownequiv{k}} (v_2,w_2)$.
%
%\end{corollary}

Lemma~\ref{lem-separationupdown1} also holds for the core XPath
algebra (with the condition $k\geq 3$ replaced by $k\geq 2$).
Lemma~\ref{lem-separationupdown2} is a different story,
unfortunately. Example~\ref{ex-notequivalent} already indicates that,
given nodes $v_1$, $w_1$, $v_2$, and $w_2$ of a document~$D$, we can
in general not hope for an expression $e_{v_1,w_1}$ such that
$(v_2,w_2)\in e_{v_1,w_1}(D)$ if and only if
$(v_1,w_1)\pairs{\sigequiv}{\updownequiv{k}}(v_2,w_2)$. The version
with congruence replaced by subsumption does hold, however.

\begin{lemma}
\label{lem-separationupdowncore}

Let $k\geq 2$.  Let $D=(V,\textit{Ed},r,\lambda)$ be a document, and
let $v_1$ and $w_1$ be a nodes of $D$.
There exists an expression $e_{v_1,w_1}$ in the 
core XPath algebra with counting up to~$k$ such that, for all nodes
$v_2$ and $w_2$ of $D$, $(v_2,w_2)\in e_{v_1,w_1}(D)$ if and only if
$(v_1,w_1)\pairs{\siggeq}{\updownequiv{k}} (v_2,w_2)$.

\end{lemma}

\begin{proof}
The proof follows the lines of the proof of
Proposition~\ref{lem-separationupdown2} very closely, the main
difference being that, from the proposed expression, the minus term
must be omitted.
\end{proof}

The BP characterization results now follow readily.

\begin{theorem}
\label{theo-updown-bpcore}

Let $k\geq 2$. Let $D=(V,\textit{Ed},r,\lambda)$ be a document, and let
$R\subseteq V\times V$. Then, there exists an expression $e$ in the
core XPath algebra with counting up to $k$ such that $e(D)=R$ if and only if,
for all $v_1,w_1,v_2,w_2\in V$ with
  $(v_1,w_1)\pairs{\siggeq}{\updownequiv{k}} (v_2,w_2)$, 
$(v_1,w_1)\in R$ implies $(v_2,w_2)\in R$. 

\end{theorem}

The specialization to the core XPath algebra is as
follows.

\begin{corollary}
\label{cor-updown-bpcore}

Let $D=(V,\textit{Ed},r,\lambda)$ be a document, and let $R\subseteq
V\times V$. There exists an expression $e$ in the core XPath algebra such
that $e(D)=R$ if and only if, 
for all $v_1,w_1,v_2,w_2\in V$ with
  $(v_1,w_1)\pairs{\siggeq}{\updownequiv{3}} (v_2,w_2)$, 
  $(v_1,w_1)\in R$ implies $(v_2,w_2)\in R$. 

\end{corollary}

We recast Theorem~\ref{theo-updown-bpcore} and
Corollary~\ref{cor-updown-bpcore} in terms of
node-level navigation.

\begin{theorem}
\label{theo-updown-nodelevelcore}

Let $k\geq 2$.  Let $D=(V,\textit{Ed},r,\lambda)$ be a document, let
$v$ be a node of $D$, and let $W\subseteq V$. Then there exists an
expression $e$ in the core XPath algebra
with counting up to~$k$ such that $e(D)(v)=W$ if and only if,
for all $w_1,w_2\in W$ with $(v,w_1)\pairs{\siggeq}{\updownequiv{k}}
(v,w_2)$, $w_1\in W$ implies $w_2\in W$.

\end{theorem}

The specialization to the core XPath algebra is as
follows.

\begin{corollary}
\label{cor-updown-nodelevelcore}

Let $D=(V,\textit{Ed},r,\lambda)$ be a document, let
$v$ be a node of $D$, and let $W\subseteq V$. Then there exists an
expression $e$ in the core XPath algebra
such that $e(D)(v)=W$ if and only if,
for all $w_1,w_2\in W$ with $(v,w_1)\pairs{\siggeq}{\updownequiv{2}}
(v,w_2)$, $w_1\in W$ implies $w_2\in W$.

\end{corollary}

Finally, for the special case where navigation starts from the
root, Theorem~\ref{theo-updown-nodelevelcore} and
Corollary~\ref{cor-updown-nodelevelcore} reduce to the following.

\begin{theorem}
\label{theo-updown-rootcore}

Let $D=(V,\textit{Ed},r,\lambda)$.
\begin{enumerate}

\item for each expression $e$ in the core XPath algebra
  with counting up to~$k$, $k\geq 2$, there exists an expression $e'$ in
  the strictly downward (core) XPath algebra with counting up to~$k$
  such that $e(D)(r)=e'(D)(r)$.

\item for each expression $e$ in the core XPath
  algebra, there exists an expression $e'$ in 
  the strictly downward (core) XPath algebra with counting up to~2
  such that $e(D)(r)=e'(D)(r)$.

\end{enumerate}
\end{theorem}

Together with Theorem~\ref{theo-updown-root},
Theorem~\ref{theo-updown-rootcore} extends
Theorem~\ref{theo-weaklydown-root}. When navigating from the root, the
only thing that the core XPath algebra adds compared to using the
strictly downward (core) XPath algebra is its ability to select on at
least~2 children satisfying some condition.
%-----------------------------------------------------------------------
\subsection{Languages without difference for two-way navigation}
\label{subsec-generalnodiff}

As before with languages not containing difference, we do not consider
counting operations, corresponding to considering the various
syntactic notions of relatedness or equivalence between nodes only for
the case $k=1$. Taking into account Proposition~\ref{prop-eliminate},
and recognizing that the techniques used in this paper to establish
characterizations heavily use intersection,
this means that only the following two languages must be considered:
\begin{itemize}

\item the language $\BL(\down,\up,\cap)$, which we call the
  \emph{positive XPath algebra}; and

\item the language $\CL(\down,\up,\pi_1,\pi_2,\cap)$, which we call the
  \emph{core positive XPath algebra}.

\end{itemize}

Some of the present authors showed in \cite{WuGGP11} that $\BL(\down,\up,\cap)$
and $\BL(\down,\up,\pi_1,\pi_2)$ are equivalent in expressive power
(even at the level of queries). Since obviously 
$\BL(\down,\up,\pi_1,\pi_2)=\CL(\down,\up,\pi_1,\pi_2)$, it follows
readily that the positive XPath algebra and the core positive XPath
algebra are equivalent. 

The following results were already proved in \cite{WuGGP11}, and are only
repeated for completeness' sake.

\begin{theorem}
\label{theo-weakupdownequivalentnodes}

Consider the (core) positive XPath algebra.
Let $D=(V,\textit{Ed},r,\lambda)$ be a document, and let
$v_1$ and $v_2$ be nodes of $D$.  Then,
\begin{enumerate}

\item  $v_1\expgeq v_2$ if and only
if $v_1\updowngeq{1} v_2$; and

\item  $v_1\expequiv v_2$ if and only
if $v_1\weakupdownequiv{} v_2$.

\end{enumerate}
\end{theorem}

\begin{theorem}
\label{theo-weakupdownequivalentpairs}

Consider the (core) positive XPath algebra.
Let $D=(V,\textit{Ed},r,\lambda)$ be a document, and let
$v_1$, $v_2$, $w_1$, and $w_2$ be nodes of $D$.  Then,
\begin{enumerate}

\item  the property that, for each (core) positive XPath expression $e$,
$(v_1,w_1)\in e(D)$ implies $(v_2,w_2)\in e(D)$ is equivalent to 
$(v_1,w_1)\pairs{\siggeq}{\updowngeq{1}}(v_2,w_2)$; and

\item  the property that, for each expression $e$ in the language
  under consideration,
$(v_1,w_1)\in e(D)$ if and only if $(v_2,w_2)\in e(D)$ is equivalent to 
the property $(v_1,w_1)\pairs{\sigequiv}{\weakupdownequiv{}}(v_2,w_2)$.

\end{enumerate}
\end{theorem}

As in Section~\ref{subsec-weaklydownnodifference}, we can bootstrap
these results to BP-type characterizations.

\begin{theorem}
\label{theo-posxpathbp-bp}

Let $D=(V,\textit{Ed},r,\lambda)$ be a document, and let
$R\subseteq V\times V$. Then, there exists an expression $e$ in the
(core) positive XPath algebra
such that $e(D)=R$ if and only if,
for all $v_1,w_1,v_2,w_2\in V$ with
$(v_1,w_1)\pairs{\siggeq}{\updowngeq{}} (v_2,w_2)$, 
$(v_1,w_1)\in R$ implies $(v_2,w_2)\in R$. 

\end{theorem}

Finally, Theorem~\ref{theo-posxpathbp-bp} can be specialized to the
node level, as follows.

\begin{corollary}
\label{cor-posxpathbp-nodelevel}

Let $D=(V,\textit{Ed},r,\lambda)$ be a document, let
$v$ be a node of $D$, and let $W\subseteq V$. Then there exists an
expression $e$ in the (core) positive XPath algebra
such that $e(D)(v)=W$ if and only if, for all nodes $w_1$ and $w_2$ of
$D$ with $(v,w_1)\pairs{\siggeq}{\updowngeq{}} (v,w_2)$, $w_1\in W$ implies
$w_2\in W$.

\end{corollary}

\begin{corollary}
\label{cor-posxpathbp-root}

Let $D=(V,\textit{Ed},r,\lambda)$ be a document, 
and let $W\subseteq V$. Then there exists an 
expression $e$ in the (core) positive XPath algebra such that
$e(D)(r)=W$ if and only if, for all nodes $w_1$ and $w_2$ of $D$
with $w_1\updowngeq{} w_2$, $w_1\in W$ implies $w_2\in W$. 

\end{corollary}

Hence, the (core) positive XPath algebra, the weakly downward positive
(core) XPath algebra, and the strictly downward positive (core) XPath
algebra are all navigationally equivalent if navigation always starts
from the root.
%======================================================================
\section{Discussion}
\label{sec-conclusions}

% expression-related & \expgeq         & Definition \ref{def-ee}\\
% expression-equivalent & \expequiv       & Definition \ref{def-ee}\\
% downward-$k$-equivalent & \downequiv{k}   & Definition \ref{def-kdownwardequivalent}\\
% upward-equivalent & \upequiv        & Definition \ref{def-upwardequivalent}\\
%    $k$-equivalent & \updownequiv{k} & Definition \ref{def-kequivalent}\\
% $\vartheta$-subsumes & \pairs{\siggeq}{\vartheta} & Definition \ref{def-pairsofnodes}\\
% $\vartheta$-congruent & \pairs{\sigequiv}{\vartheta} &  Definition \ref{def-pairsofnodes}\\
% downward-related & \downgeq{} & Definition \ref{def-downrelated}\\
% weakly-downward-equivalent & \weakdownequiv{} & Definition \ref{def-downrelated}\\
% related & \updowngeq{} & Definition \ref{def-weakequivalent}\\
% weakly-equivalent & \weakupdownequiv{} & Definition \ref{def-weakequivalent}\\
\begin{sidewaystable}
    \centering
    \caption{Summary of main results.}\label{table:summary}
\footnotesize{
    \begin{tabular}{|p{.21\textwidth}ccccc|}
        \hline
        {\em Language} & {\em Node relationship} & {\em Node Coupling Theorem}  & {\em Path Relationship} & {\em Path Coupling Theorem} & {\em BP Result}\\
        \hline
        strictly downward (core) XPath algebra with counting up to $k$ & \downequiv{k} & Theorem \ref{theo-downequivalent} 
                                                                       &\pairs{\sigequiv}{\downequiv{k}} &Corollary \ref{cor-strictlydown}& Theorem \ref{theo-strictlydown-bp} \\
        \hline
        strictly downward (core) positive XPath algebra &\weakdownequiv{} & Theorem \ref{theo-weakdownequivalent} 
                                                        &\pairs{\sigequiv}{\weakdownequiv{}} &Theorem \ref{theo-weakstrictlydown}& Theorem \ref{theo-weakstrictlydown-bp} \\
        \hline
        weakly downward (core) XPath algebra with counting up to $k$ & \updownequiv{k} & Theorem \ref{theo-weaklydownequivalent} 
                                                                     &\pairs{\sigequiv}{\updownequiv{k}}&Corollary \ref{cor-weaklydown}& Theorem \ref{theo-weaklydown-bp} \\
        \hline
        weakly downward (core) positive XPath algebra & \weakupdownequiv{}  & Theorem \ref{theo-weakweaklydownequivalent} 
                                                      &\pairs{\sigequiv}{\weakupdownequiv{}}&Theorem \ref{theo-weakweaklydown}& Theorem \ref{theo-weakweaklydown-bp} \\
        \hline
        strictly upward (core) XPath algebra & \upequiv & Theorem \ref{theo-upequivalent} 
                                             &\pairs{\sigequiv}{\upequiv}&Theorem \ref{theo-strictlyup}& Theorem \ref{theo-strictlyup-bp} \\
        \hline
        strictly upward (core) positive XPath algebra &\upequiv  & Theorem \ref{theo-upequivalent} 
                                                      &\pairs{\sigequiv}{\upequiv}&Theorem \ref{theo-strictlyup}& Theorem \ref{theo-weakstrictlyup-bp} \\
        \hline
        weakly upward languages & & & {\em see Section \ref{subsec-weaklyupward}}   &&  \\
        \hline
        XPath algebra with counting up to $k$ & \updownequiv{k} & Theorem \ref{theo-updown} 
                                              &\pairs{\sigequiv}{\updownequiv{k}}&Corollary \ref{cor-updownpairs}& Theorem \ref{theo-updown-bp} \\
        \hline
        core XPath algebra with counting up to $k$ & \updownequiv{k} & Theorem \ref{theo-updowncore} 
                                                   &\pairs{\sigequiv}{\updownequiv{k}}&Corollary \ref{cor-updownpairscore}& Theorem \ref{theo-updown-bpcore} \\
        \hline
        (core) positive XPath algebra (\cite{WuGGP11})& \weakupdownequiv{}  & Theorem \ref{theo-weakupdownequivalentnodes}  
                                                      &\pairs{\sigequiv}{\weakupdownequiv{}} &Theorem \ref{theo-weakupdownequivalentpairs}& Theorem \ref{theo-posxpathbp-bp} \\
        \hline
    \end{tabular}
}
\end{sidewaystable}

In this paper, we characterized the expressive power of several natural
fragments of XPath at the document level, as summarized in Table
\ref{table:summary}.  Of course, it is possible to
consider other fragments or extensions of the XPath algebra and its
data model. Analyzing these using our two-step methodology in order
to further improve our understanding of the instance expressivity of Tarski's
algebra is one possible research
direction which we have pursued recently \cite{FletcherGLBGVW11,FletcherGLBGVW12,WuGGP11}.  

Another future research direction is
refining the links between XPath 
and finite-variable first-order logics \cite{LibkinFMT}.
Indeed, such links have been established at the level of query semantics.
For example, Marx~\cite{MarxTODS05}
has shown that an extended version of Core XPath is equivalent to
$\textrm{FO}^2_{\mathtt{tree}}$---first-order logic using at
most two variables over \emph{ordered} node-labeled
trees---interpreted in the signature
\texttt{child}, 
\texttt{descendant}, and
\texttt{following\_sibling}. 

Our results
establish new links to finite-variable first-order logics at the
document level. 
For example, we can
show that, on a given document, the XPath algebra and
$\textrm{FO}^3$---first-order logic with at most three variables---are
equivalent in expressive power.
Indeed, as we discussed above, at the document level, the
XPath-algebra is equivalent with Tarski's relation
algebra~\cite{Tarski41} over trees. Tarski and
Givant~\cite{Givant06,TarskiGivant} established the link between
Tarski's algebra and $\textrm{FO}^3$. Corollary~\ref{cor-updownxpath} can then
be used to give a new characterization, other than via
pebble-games \cite{LibkinFMT,Krzeszczakowski03}, of when two nodes in
an unordered tree are indistinguishable in $\textrm{FO}^3$. In this
light, connections between other fragments of the XPath algebra and
finite-variable logics must be examined.

The connection between the XPath algebra and $\textrm{FO}^3$ also has
ramifications with regard to complexity issues. Indeed, using a result of
Grohe~\cite{Grohe99} which establishes that expression equivalence for
$\textrm{FO}^3$ is decidable in polynomial time, it follows readily from
Corollaries~\ref{cor-updown-bp} and \ref{cor-updown-nodelevel} that the global and local
definability problems for the XPath algebra are decidable in polynomial time.
Using the syntactic characterizations in this paper, one
can also establish that the global and local definability problems for the other
fragments of the XPath algebra are decidable in polynomial time.  As mentioned
in the Introduction, this feasibility suggests efficient partitioning and
reduction techniques on the set of nodes and the set of paths in a document.
Such techniques might be successfully applied towards various aspects of XML
document processing such as indexing, access control, and document compression.
This is another research direction which we are currently pursuing
\cite{Sofia,FletcherGWGBP09}.

%======================================================================
\bibliographystyle{model1-num-names}
\bibliography{xpathbib}

\end{document}

%% file: document.pdf_t
\begin{picture}(0,0)%
\includegraphics{document.pdf}%
\end{picture}%
\setlength{\unitlength}{4144sp}%
\begingroup\makeatletter\ifx\SetFigFont\undefined%
\gdef\SetFigFont#1#2#3#4#5{%
  \reset@font\fontsize{#1}{#2pt}%
  \fontfamily{#3}\fontseries{#4}\fontshape{#5}%
  \selectfont}%
\fi\endgroup%
\begin{picture}(9241,5059)(218,-4634)
\put(7201,-1096){\makebox(0,0)[b]{\smash{{\SetFigFont{14}{16.8}{\rmdefault}{\mddefault}{\updefault}{\color[rgb]{0,0,0}$b$}%
}}}}
\put(3151,254){\makebox(0,0)[b]{\smash{{\SetFigFont{14}{16.8}{\familydefault}{\mddefault}{\updefault}{\color[rgb]{0,0,0}$a$}%
}}}}
\put(3150,-128){\makebox(0,0)[b]{\smash{{\SetFigFont{12}{14.4}{\familydefault}{\mddefault}{\updefault}{\color[rgb]{0,0,0}$v_1$}%
}}}}
\put(5850,-2828){\makebox(0,0)[b]{\smash{{\SetFigFont{12}{14.4}{\familydefault}{\mddefault}{\updefault}{\color[rgb]{0,0,0}$v_8$}%
}}}}
\put(6526,-4561){\makebox(0,0)[b]{\smash{{\SetFigFont{14}{16.8}{\rmdefault}{\mddefault}{\updefault}{\color[rgb]{0,0,0}$d$}%
}}}}
\put(7876,-4561){\makebox(0,0)[b]{\smash{{\SetFigFont{14}{16.8}{\rmdefault}{\mddefault}{\updefault}{\color[rgb]{0,0,0}$d$}%
}}}}
\put(9226,-4561){\makebox(0,0)[b]{\smash{{\SetFigFont{14}{16.8}{\rmdefault}{\mddefault}{\updefault}{\color[rgb]{0,0,0}$c$}%
}}}}
\put(9225,-4178){\makebox(0,0)[b]{\smash{{\SetFigFont{12}{14.4}{\rmdefault}{\mddefault}{\updefault}{\color[rgb]{0,0,0}$v_{13}$}%
}}}}
\put(7875,-4178){\makebox(0,0)[b]{\smash{{\SetFigFont{12}{14.4}{\rmdefault}{\mddefault}{\updefault}{\color[rgb]{0,0,0}$v_{12}$}%
}}}}
\put(6525,-4178){\makebox(0,0)[b]{\smash{{\SetFigFont{12}{14.4}{\rmdefault}{\mddefault}{\updefault}{\color[rgb]{0,0,0}$v_{11}$}%
}}}}
\put(5806,-3211){\makebox(0,0)[b]{\smash{{\SetFigFont{14}{16.8}{\rmdefault}{\mddefault}{\updefault}{\color[rgb]{0,0,0}$c$}%
}}}}
\put(7200,-2828){\makebox(0,0)[b]{\smash{{\SetFigFont{12}{14.4}{\rmdefault}{\mddefault}{\updefault}{\color[rgb]{0,0,0}$v_9$}%
}}}}
\put(7201,-3211){\makebox(0,0)[b]{\smash{{\SetFigFont{14}{16.8}{\rmdefault}{\mddefault}{\updefault}{\color[rgb]{0,0,0}$c$}%
}}}}
\put(8551,-3211){\makebox(0,0)[b]{\smash{{\SetFigFont{14}{16.8}{\rmdefault}{\mddefault}{\updefault}{\color[rgb]{0,0,0}$b$}%
}}}}
\put(8550,-2828){\makebox(0,0)[b]{\smash{{\SetFigFont{12}{14.4}{\rmdefault}{\mddefault}{\updefault}{\color[rgb]{0,0,0}$v_{10}$}%
}}}}
\put(4050,-2828){\makebox(0,0)[b]{\smash{{\SetFigFont{12}{14.4}{\rmdefault}{\mddefault}{\updefault}{\color[rgb]{0,0,0}$v_7$}%
}}}}
\put(4051,-3211){\makebox(0,0)[b]{\smash{{\SetFigFont{14}{16.8}{\rmdefault}{\mddefault}{\updefault}{\color[rgb]{0,0,0}$c$}%
}}}}
\put(2251,-3211){\makebox(0,0)[b]{\smash{{\SetFigFont{14}{16.8}{\rmdefault}{\mddefault}{\updefault}{\color[rgb]{0,0,0}$c$}%
}}}}
\put(2250,-2828){\makebox(0,0)[b]{\smash{{\SetFigFont{12}{14.4}{\rmdefault}{\mddefault}{\updefault}{\color[rgb]{0,0,0}$v_6$}%
}}}}
\put(450,-2828){\makebox(0,0)[b]{\smash{{\SetFigFont{12}{14.4}{\rmdefault}{\mddefault}{\updefault}{\color[rgb]{0,0,0}$v_5$}%
}}}}
\put(451,-3211){\makebox(0,0)[b]{\smash{{\SetFigFont{14}{16.8}{\rmdefault}{\mddefault}{\updefault}{\color[rgb]{0,0,0}$c$}%
}}}}
\put(450,-1478){\makebox(0,0)[b]{\smash{{\SetFigFont{12}{14.4}{\rmdefault}{\mddefault}{\updefault}{\color[rgb]{0,0,0}$v_2$}%
}}}}
\put(451,-1096){\makebox(0,0)[b]{\smash{{\SetFigFont{14}{16.8}{\rmdefault}{\mddefault}{\updefault}{\color[rgb]{0,0,0}$b$}%
}}}}
\put(3150,-1478){\makebox(0,0)[b]{\smash{{\SetFigFont{12}{14.4}{\rmdefault}{\mddefault}{\updefault}{\color[rgb]{0,0,0}$v_3$}%
}}}}
\put(3151,-1861){\makebox(0,0)[b]{\smash{{\SetFigFont{14}{16.8}{\rmdefault}{\mddefault}{\updefault}{\color[rgb]{0,0,0}$b$}%
}}}}
\put(7200,-1478){\makebox(0,0)[b]{\smash{{\SetFigFont{12}{14.4}{\rmdefault}{\mddefault}{\updefault}{\color[rgb]{0,0,0}$v_4$}%
}}}}
\end{picture}%

%% file: bp-nottrans.pdf_t
\begin{picture}(0,0)%
\includegraphics{bp-nottrans.pdf}%
\end{picture}%
\setlength{\unitlength}{4144sp}%
\begingroup\makeatletter\ifx\SetFigFont\undefined%
\gdef\SetFigFont#1#2#3#4#5{%
  \reset@font\fontsize{#1}{#2pt}%
  \fontfamily{#3}\fontseries{#4}\fontshape{#5}%
  \selectfont}%
\fi\endgroup%
\begin{picture}(6766,3841)(3368,-2994)
\put(7200,-2828){\makebox(0,0)[b]{\smash{{\SetFigFont{12}{14.4}{\familydefault}{\mddefault}{\updefault}{\color[rgb]{0,0,0}$v_2$}%
}}}}
\put(9900,-1478){\makebox(0,0)[b]{\smash{{\SetFigFont{12}{14.4}{\familydefault}{\mddefault}{\updefault}{\color[rgb]{0,0,0}$w_2$}%
}}}}
\put(8775,-2153){\makebox(0,0)[b]{\smash{{\SetFigFont{12}{14.4}{\familydefault}{\mddefault}{\updefault}{\color[rgb]{0,0,0}$z_2$}%
}}}}
\put(6300,-1478){\makebox(0,0)[b]{\smash{{\SetFigFont{12}{14.4}{\familydefault}{\mddefault}{\updefault}{\color[rgb]{0,0,0}$w_1$}%
}}}}
\put(5400,-2153){\makebox(0,0)[b]{\smash{{\SetFigFont{12}{14.4}{\familydefault}{\mddefault}{\updefault}{\color[rgb]{0,0,0}$z_1$}%
}}}}
\put(3600,-2828){\makebox(0,0)[b]{\smash{{\SetFigFont{12}{14.4}{\familydefault}{\mddefault}{\updefault}{\color[rgb]{0,0,0}$v_1$}%
}}}}
\end{picture}%

%% file: subpairs.pdf_t
\begin{picture}(0,0)%
\includegraphics{subpairs.pdf}%
\end{picture}%
\setlength{\unitlength}{4144sp}%
\begingroup\makeatletter\ifx\SetFigFont\undefined%
\gdef\SetFigFont#1#2#3#4#5{%
  \reset@font\fontsize{#1}{#2pt}%
  \fontfamily{#3}\fontseries{#4}\fontshape{#5}%
  \selectfont}%
\fi\endgroup%
\begin{picture}(2640,2511)(2056,-3334)
\put(4681,-3256){\makebox(0,0)[b]{\smash{{\SetFigFont{12}{14.4}{\rmdefault}{\mddefault}{\updefault}{\color[rgb]{0,0,0}$w_2$}%
}}}}
\put(3961,-2806){\makebox(0,0)[b]{\smash{{\SetFigFont{12}{14.4}{\rmdefault}{\mddefault}{\updefault}{\color[rgb]{0,0,0}$z_2$}%
}}}}
\put(2071,-1006){\makebox(0,0)[b]{\smash{{\SetFigFont{12}{14.4}{\rmdefault}{\mddefault}{\updefault}{\color[rgb]{0,0,0}$v_1$}%
}}}}
\put(2836,-1681){\makebox(0,0)[b]{\smash{{\SetFigFont{12}{14.4}{\rmdefault}{\mddefault}{\updefault}{\color[rgb]{0,0,0}$y_1$}%
}}}}
\put(2161,-2356){\makebox(0,0)[b]{\smash{{\SetFigFont{12}{14.4}{\rmdefault}{\mddefault}{\updefault}{\color[rgb]{0,0,0}$z_1$}%
}}}}
\put(2881,-2806){\makebox(0,0)[b]{\smash{{\SetFigFont{12}{14.4}{\rmdefault}{\mddefault}{\updefault}{\color[rgb]{0,0,0}$w_1$}%
}}}}
\put(4636,-2131){\makebox(0,0)[b]{\smash{{\SetFigFont{12}{14.4}{\rmdefault}{\mddefault}{\updefault}{\color[rgb]{0,0,0}$y_2$}%
}}}}
\put(3871,-1456){\makebox(0,0)[b]{\smash{{\SetFigFont{12}{14.4}{\rmdefault}{\mddefault}{\updefault}{\color[rgb]{0,0,0}$v_2$}%
}}}}
\end{picture}%

%% file: weak.pdf_t
\begin{picture}(0,0)%
\includegraphics{weak.pdf}%
\end{picture}%
\setlength{\unitlength}{4144sp}%
\begingroup\makeatletter\ifx\SetFigFont\undefined%
\gdef\SetFigFont#1#2#3#4#5{%
  \reset@font\fontsize{#1}{#2pt}%
  \fontfamily{#3}\fontseries{#4}\fontshape{#5}%
  \selectfont}%
\fi\endgroup%
\begin{picture}(4066,2491)(4268,-1644)
\put(8086,-803){\makebox(0,0)[b]{\smash{{\SetFigFont{12}{14.4}{\familydefault}{\mddefault}{\updefault}{\color[rgb]{0,0,0}$x_2$}%
}}}}
\put(5400,-803){\makebox(0,0)[b]{\smash{{\SetFigFont{12}{14.4}{\familydefault}{\mddefault}{\updefault}{\color[rgb]{0,0,0}$y_1$}%
}}}}
\put(4500,-803){\makebox(0,0)[b]{\smash{{\SetFigFont{12}{14.4}{\familydefault}{\mddefault}{\updefault}{\color[rgb]{0,0,0}$x_1$}%
}}}}
\put(8093,-128){\makebox(0,0)[b]{\smash{{\SetFigFont{12}{14.4}{\familydefault}{\mddefault}{\updefault}{\color[rgb]{0,0,0}$v_2$}%
}}}}
\put(4500,-128){\makebox(0,0)[b]{\smash{{\SetFigFont{12}{14.4}{\familydefault}{\mddefault}{\updefault}{\color[rgb]{0,0,0}$v_1$}%
}}}}
\end{picture}%

%% file: bp-lem-mgc1.pdf_t
\begin{picture}(0,0)%
\includegraphics{bp-lem-mgc1.pdf}%
\end{picture}%
\setlength{\unitlength}{3947sp}%
\begingroup\makeatletter\ifx\SetFigFont\undefined%
\gdef\SetFigFont#1#2#3#4#5{%
  \reset@font\fontsize{#1}{#2pt}%
  \fontfamily{#3}\fontseries{#4}\fontshape{#5}%
  \selectfont}%
\fi\endgroup%
\begin{picture}(13066,5362)(1118,-4547)
\put(3901,464){\makebox(0,0)[b]{\smash{{\SetFigFont{20}{24.0}{\rmdefault}{\mddefault}{\updefault}{\color[rgb]{0,0,0}$\top(v_1,w_1)$}%
}}}}
\put(6301,-4411){\makebox(0,0)[b]{\smash{{\SetFigFont{20}{24.0}{\rmdefault}{\mddefault}{\updefault}{\color[rgb]{0,0,0}$w_1$}%
}}}}
\put(11701,-436){\makebox(0,0)[b]{\smash{{\SetFigFont{20}{24.0}{\rmdefault}{\mddefault}{\updefault}{\color[rgb]{0,0,0}$\neq$}%
}}}}
\put(11776,464){\makebox(0,0)[b]{\smash{{\SetFigFont{20}{24.0}{\rmdefault}{\mddefault}{\updefault}{\color[rgb]{0,0,0}$t_2$}%
}}}}
\put(1201,-3211){\makebox(0,0)[b]{\smash{{\SetFigFont{20}{24.0}{\rmdefault}{\mddefault}{\updefault}{\color[rgb]{0,0,0}$v_1$}%
}}}}
\put(9001,-3211){\makebox(0,0)[b]{\smash{{\SetFigFont{20}{24.0}{\rmdefault}{\mddefault}{\updefault}{\color[rgb]{0,0,0}$v_2$}%
}}}}
\put(14101,-4411){\makebox(0,0)[b]{\smash{{\SetFigFont{20}{24.0}{\rmdefault}{\mddefault}{\updefault}{\color[rgb]{0,0,0}$w_2$}%
}}}}
\put(4801,-436){\makebox(0,0)[lb]{\smash{{\SetFigFont{20}{24.0}{\rmdefault}{\mddefault}{\updefault}{\color[rgb]{0,0,0}$y_1$}%
}}}}
\put(12601,-436){\makebox(0,0)[lb]{\smash{{\SetFigFont{20}{24.0}{\rmdefault}{\mddefault}{\updefault}{\color[rgb]{0,0,0}$y_2$}%
}}}}
\end{picture}%

%% file: bp-lem-key2.pdf_t
\begin{picture}(0,0)%
\includegraphics{bp-lem-key2.pdf}%
\end{picture}%
\setlength{\unitlength}{3947sp}%
\begingroup\makeatletter\ifx\SetFigFont\undefined%
\gdef\SetFigFont#1#2#3#4#5{%
  \reset@font\fontsize{#1}{#2pt}%
  \fontfamily{#3}\fontseries{#4}\fontshape{#5}%
  \selectfont}%
\fi\endgroup%
\begin{picture}(10666,5967)(2318,-5752)
\put(2401,-3211){\makebox(0,0)[b]{\smash{{\SetFigFont{20}{24.0}{\rmdefault}{\mddefault}{\updefault}{\color[rgb]{0,0,0}$v_1$}%
}}}}
\put(9001,-3211){\makebox(0,0)[b]{\smash{{\SetFigFont{20}{24.0}{\rmdefault}{\mddefault}{\updefault}{\color[rgb]{0,0,0}$v_2$}%
}}}}
\put(4951,-5611){\makebox(0,0)[b]{\smash{{\SetFigFont{20}{24.0}{\rmdefault}{\mddefault}{\updefault}{\color[rgb]{0,0,0}$y_1$}%
}}}}
\put(11551,-5611){\makebox(0,0)[b]{\smash{{\SetFigFont{20}{24.0}{\rmdefault}{\mddefault}{\updefault}{\color[rgb]{0,0,0}$y_2$}%
}}}}
\put(6301,-4411){\makebox(0,0)[b]{\smash{{\SetFigFont{20}{24.0}{\rmdefault}{\mddefault}{\updefault}{\color[rgb]{0,0,0}$w_1$}%
}}}}
\put(12901,-4411){\makebox(0,0)[b]{\smash{{\SetFigFont{20}{24.0}{\rmdefault}{\mddefault}{\updefault}{\color[rgb]{0,0,0}$w_1$}%
}}}}
\put(3076,-1636){\makebox(0,0)[rb]{\smash{{\SetFigFont{20}{24.0}{\rmdefault}{\mddefault}{\updefault}{\color[rgb]{0,0,0}$\top(v_1,y_1)$}%
}}}}
\put(4501,-136){\makebox(0,0)[b]{\smash{{\SetFigFont{20}{24.0}{\rmdefault}{\mddefault}{\updefault}{\color[rgb]{0,0,0}$\top(y_1,w_1)=\top(v_1,w_1)$}%
}}}}
\end{picture}%

%% file: bp-lem-key1.pdf_t
\begin{picture}(0,0)%
\includegraphics{bp-lem-key1.pdf}%
\end{picture}%
\setlength{\unitlength}{3947sp}%
\begingroup\makeatletter\ifx\SetFigFont\undefined%
\gdef\SetFigFont#1#2#3#4#5{%
  \reset@font\fontsize{#1}{#2pt}%
  \fontfamily{#3}\fontseries{#4}\fontshape{#5}%
  \selectfont}%
\fi\endgroup%
\begin{picture}(11213,5962)(2318,-4547)
\put(2401,-3211){\makebox(0,0)[b]{\smash{{\SetFigFont{20}{24.0}{\rmdefault}{\mddefault}{\updefault}{\color[rgb]{0,0,0}$v_1$}%
}}}}
\put(9001,-3211){\makebox(0,0)[b]{\smash{{\SetFigFont{20}{24.0}{\rmdefault}{\mddefault}{\updefault}{\color[rgb]{0,0,0}$v_2$}%
}}}}
\put(6301,-4411){\makebox(0,0)[b]{\smash{{\SetFigFont{20}{24.0}{\rmdefault}{\mddefault}{\updefault}{\color[rgb]{0,0,0}$w_1$}%
}}}}
\put(12901,-4411){\makebox(0,0)[b]{\smash{{\SetFigFont{20}{24.0}{\rmdefault}{\mddefault}{\updefault}{\color[rgb]{0,0,0}$w_2$}%
}}}}
\put(6901,-2011){\makebox(0,0)[b]{\smash{{\SetFigFont{20}{24.0}{\rmdefault}{\mddefault}{\updefault}{\color[rgb]{0,0,0}$y_1$}%
}}}}
\put(13426,-2011){\makebox(0,0)[b]{\smash{{\SetFigFont{20}{24.0}{\rmdefault}{\mddefault}{\updefault}{\color[rgb]{0,0,0}$y_2$}%
}}}}
\put(4276,-436){\makebox(0,0)[rb]{\smash{{\SetFigFont{20}{24.0}{\rmdefault}{\mddefault}{\updefault}{\color[rgb]{0,0,0}$\top(v_1,w_1)$}%
}}}}
\put(5551,1064){\makebox(0,0)[b]{\smash{{\SetFigFont{20}{24.0}{\rmdefault}{\mddefault}{\updefault}{\color[rgb]{0,0,0}$\top(v_1,y_1)=\top(y_1,w_1)$}%
}}}}
\end{picture}%

%% file: bp-lem-key3b.pdf_t
\begin{picture}(0,0)%
\includegraphics{bp-lem-key3b.pdf}%
\end{picture}%
\setlength{\unitlength}{4144sp}%
\begingroup\makeatletter\ifx\SetFigFont\undefined%
\gdef\SetFigFont#1#2#3#4#5{%
  \reset@font\fontsize{#1}{#2pt}%
  \fontfamily{#3}\fontseries{#4}\fontshape{#5}%
  \selectfont}%
\fi\endgroup%
\begin{picture}(10513,3727)(2037,-3392)
\put(9811,-2806){\makebox(0,0)[b]{\smash{{\SetFigFont{20}{24.0}{\rmdefault}{\mddefault}{\updefault}{\color[rgb]{0,0,0}$y_2$}%
}}}}
\put(10936,-1276){\makebox(0,0)[b]{\smash{{\SetFigFont{20}{24.0}{\rmdefault}{\mddefault}{\updefault}{\color[rgb]{0,0,0}$z_2$}%
}}}}
\put(4276,-16){\makebox(0,0)[b]{\smash{{\SetFigFont{20}{24.0}{\rmdefault}{\mddefault}{\updefault}{\color[rgb]{0,0,0}$\top(v_1,y_1)=\top(y_1,w_1)=\top(v_1,w_1)$}%
}}}}
\put(2071,-3256){\makebox(0,0)[b]{\smash{{\SetFigFont{20}{24.0}{\rmdefault}{\mddefault}{\updefault}{\color[rgb]{0,0,0}$v_1$}%
}}}}
\put(3511,-2806){\makebox(0,0)[b]{\smash{{\SetFigFont{20}{24.0}{\rmdefault}{\mddefault}{\updefault}{\color[rgb]{0,0,0}$y_1$}%
}}}}
\put(4636,-1276){\makebox(0,0)[b]{\smash{{\SetFigFont{20}{24.0}{\rmdefault}{\mddefault}{\updefault}{\color[rgb]{0,0,0}$z_1$}%
}}}}
\put(6166,-2986){\makebox(0,0)[b]{\smash{{\SetFigFont{20}{24.0}{\rmdefault}{\mddefault}{\updefault}{\color[rgb]{0,0,0}$w_1$}%
}}}}
\put(8371,-3256){\makebox(0,0)[b]{\smash{{\SetFigFont{20}{24.0}{\rmdefault}{\mddefault}{\updefault}{\color[rgb]{0,0,0}$v_2$}%
}}}}
\put(12466,-2986){\makebox(0,0)[b]{\smash{{\SetFigFont{20}{24.0}{\rmdefault}{\mddefault}{\updefault}{\color[rgb]{0,0,0}$w_2$}%
}}}}
\put(10576,-16){\makebox(0,0)[b]{\smash{{\SetFigFont{20}{24.0}{\rmdefault}{\mddefault}{\updefault}{\color[rgb]{0,0,0}$\top(v_2,w_2)$}%
}}}}
\end{picture}%

%% file: bp-notequivalent.pdf_t
\begin{picture}(0,0)%
\includegraphics{bp-notequivalent.pdf}%
\end{picture}%
\setlength{\unitlength}{4144sp}%
\begingroup\makeatletter\ifx\SetFigFont\undefined%
\gdef\SetFigFont#1#2#3#4#5{%
  \reset@font\fontsize{#1}{#2pt}%
  \fontfamily{#3}\fontseries{#4}\fontshape{#5}%
  \selectfont}%
\fi\endgroup%
\begin{picture}(4066,1141)(4261,-294)
\put(4493,-128){\makebox(0,0)[b]{\smash{{\SetFigFont{12}{14.4}{\familydefault}{\mddefault}{\updefault}{\color[rgb]{0,0,0}$v$}%
}}}}
\put(6300,547){\makebox(0,0)[b]{\smash{{\SetFigFont{12}{14.4}{\familydefault}{\mddefault}{\updefault}{\color[rgb]{0,0,0}$r$}%
}}}}
\put(8093,-128){\makebox(0,0)[b]{\smash{{\SetFigFont{12}{14.4}{\familydefault}{\mddefault}{\updefault}{\color[rgb]{0,0,0}$w$}%
}}}}
\end{picture}%